\documentclass{lmcs}
\pdfoutput=1
\usepackage[utf8]{inputenc}

\usepackage{lastpage}
\lmcsdoi{17}{2}{1}
\lmcsheading{}{\pageref{LastPage}}{}{}%
{May~07,~2020}{Apr.~12,~2021}{}

\pdfoutput=1


\PassOptionsToPackage{hyphens}{url}
\usepackage[hyphenbreaks]{breakurl}

\usepackage[T1]{fontenc}
\usepackage[misc]{ifsym}
\usepackage{amsmath}
\usepackage{amssymb}
\usepackage{bm}
\usepackage{booktabs}
\usepackage{cite}
\usepackage{color}
\usepackage{stmaryrd}
\usepackage{textcomp}
\usepackage{url}
\usepackage{mathtools}
\usepackage{microtype}
\usepackage{array}
\usepackage{nicefrac}
\usepackage{subcaption}
\usepackage{tabu}
\usepackage{prftree}
\usepackage{pifont}
\usepackage{float}
\usepackage{bm}
\usepackage{xr}

\usepackage{wasysym} %

\makeatletter
\def\@seccntformat#1{%
  \protect\textup{\protect\@secnumfont
    \ifnum\pdfstrcmp{subsection}{#1}=0 \bfseries\fi
    \csname the#1\endcsname
    \protect\@secnumpunct
  }%
}

\newcommand\Section{Section}

\g@addto@macro{\UrlBreaks}{\UrlOrds\do\=\do\_}

\DeclareFontFamily{OT1}{pzc}{}
\DeclareFontShape{OT1}{pzc}{m}{it}{<-> s * [1.10] pzcmi7t}{}
\DeclareMathAlphabet{\mathcalx}{OT1}{pzc}{m}{it}

\def\negvthinspace{\kern-0.083333em}
\def\vthinspace{\kern+0.083333em}
\def\vvthinspace{\kern+0.0416667em}
\def\negvvthinspace{\kern-0.0416667em}
\def\hypsep{\quad}
\def\rulesep{\kern1.75em}

\newcommand\medrightarrow{\mathrel{{{\color{black}\relbar}\kern-0.9ex\rlap{\color{white}\ensuremath{\blacksquare}}\kern-0.9ex}\joinrel{\color{black}\rightarrow}}}
\newcommand\medleftarrow{\mathrel{{\color{black}\leftarrow}\kern-0.9ex\rlap{\color{white}\ensuremath{\blacksquare}}\kern-0.9ex\joinrel{{\color{black}\relbar}}}}
\newcommand\medleftrightarrow{\mathrel{\leftarrow\kern-1.685ex\rightarrow}}

\newcommand{\rewrite}{\medrightarrow}

\newcommand{\leftrightrewrite}{\medleftrightarrow}

\def\cpp{C\nobreak\raisebox{.1ex}{+}\nobreak\raisebox{.1ex}{+}}

\definecolor{light-gray}{gray}{0.9}
\definecolor{darker-gray}{gray}{0.45}

\let\oldSigma=\Sigma
\renewcommand\Sigma{\mathrm{\oldSigma}}

\newcommand\eqIH{\overset{\smash{\scriptscriptstyle\text{IH}}}{=}}

\newcommand{\flooronly}{\mathcalx{F}}
\newcommand{\ceilonly}{\flooronly^{-1}}
\newcommand{\floor}[1]{\flooronly(#1)}
\newcommand{\ceil}[1]{\ceilonly(#1)}

\newcommand{\III}{\mathcal{I}}

\newcommand{\II}{\mathcalx{J}}
\newcommand{\IIty}{\II_\mathsf{ty}}
\newcommand{\IIIty}{\III_\mathsf{ty}}

\newcommand{\Sigmaty}{\Sigma_\mathsf{ty}}
\newcommand{\VV}{\mathcalx{V}}
\newcommand{\Vty}{\VV_\mathsf{ty}}

\newcommand{\UU}{\mathcalx{U}}
\newcommand{\U}{U}
\newcommand{\EE}{\mathcalx{E}}

\newcommand{\uho}{\UU^{\smash{\GH}}}
\newcommand{\iho}{\II^{\smash{\GH}}}
\newcommand{\eho}{\EE^{\smash{\GH}}}
\newcommand{\IIIho}{\III^{\smash{\GH}}}

\newcommand{\ifo}{\II}

\newcommand{\RfN}{R}

\newcommand\foralltynospace[1]{\mathsf{\Pi}#1.}
\newcommand\forallty[1]{\foralltynospace{#1}\;}

\newcommand{\infname}[1]{\textsc{#1}}

\newcommand{\unified}[1]{\smash{\setlength{\fboxsep}{.3ex}\colorbox{light-gray}{\ensuremath{\vphantom{('q}{#1}}}}}

\newcommand{\lefthosubterm}{[}
\newcommand{\righthosubterm}{]}
\newcommand{\hosubterm}[2]{#1\lefthosubterm#2\righthosubterm}

\newcommand{\lang}{\begin{picture}(5,7)
\put(1.1,3){\rotatebox{45}{\line(1,0){6.0}}}
\put(1.1,3){\rotatebox{315}{\line(1,0){6.0}}}
\end{picture}}
\newcommand{\rang}{\begin{picture}(5,7)
\put(.1,3){\rotatebox{135}{\line(1,0){6.0}}}
\put(.1,3){\rotatebox{225}{\line(1,0){6.0}}}
\end{picture}}
\newcommand{\leftfosubterm}{\lang\vthinspace}
\newcommand{\rightfosubterm}{\rang\,}
\newcommand{\fosubterm}[2]{#1\leftfosubterm #2\rightfosubterm}
\newcommand{\subterm}[2]{#1[#2]}

\newcommand{\leftinterpret}{\llbracket}
\newcommand{\rightinterpret}{\rrbracket}
\newcommand{\interpret}[3]{\smash{\leftinterpret #1\rightinterpret_{#2}^{\smash{#3}}}}
\newcommand{\interpreta}[1]{\interpret{#1}{\III}{}}
\newcommand{\interpretaxi}[1]{\interpret{#1}{\III}{\xi}}
\newcommand{\interpretfo}[2]{\interpret{#1}{\RfN}{#2}}
\newcommand{\interpretho}[2]{\interpret{#1}{\IIIho}{#2}}

\newcommand{\interpretfog}[1]{\interpretfo{#1}{}}
\newcommand{\interprethog}[1]{\interpretho{#1}{}}

\renewcommand{\doteq}{\mathrel{\dot\eq}}
\newcommand{\eq}{\approx}
\newcommand{\noteq}{\not\eq}

\newcommand{\namedinference}[3]{\prftree[r]{\small{\infname{#1}}}{\strut#2}{\strut#3}}

\DeclareMathOperator{\mgu}{mgu}

\newcommand{\tuple}[1]{\bar{#1}}
\newcommand\succL{\succ}
\newcommand{\precC}{\prec}
\newcommand{\preceqC}{\preceq}
\newcommand{\succC}{\succ}
\newcommand{\succeqC}{\succeq}

\newcommand{\cst}[1]{{\mathsf{#1}}}
\newcommand{\var}[1]{{\mathit{#1}}}

\newcommand{\typ}[1]{{\cst{#1}}}

\newcommand\defeq{=}

\newcommand\fun{\rightarrow}

\newcommand\fofun{\Rightarrow}

\newcommand{\gnd}{\mathcalx{G}}

\newcommand{\typeargs}[1]{{\langle#1\rangle}}

\newcommand\oftype{:}
\newcommand\oftypedecl{:}

\DeclareMathOperator{\pure}{\mathcalx{pure}}
\DeclareMathOperator{\pureext}{\pure}
\DeclareMathOperator{\pureint}{\pure}

\newcommand{\diff}{\mathsf{diff}}

\newcommand{\llor}{\mathrel\lor}
\newcommand{\ccup}{\mathrel\cup}
\newcommand{\ccap}{\mathrel\cap}

\newcommand{\TT}{\mathcalx{T}}
\newcommand{\Ty}{\mathcalx{Ty}}
\newcommand{\CC}{\mathcalx{C}}
\newcommand{\HH}{{\mathrm{H}}}
\newcommand{\GH}{{\mathrm{GH}}}
\newcommand{\GF}{{\mathrm{GF}}}
\newcommand{\THH}{\TT_\HH}
\newcommand{\TGH}{\TT_\GH}
\newcommand{\TGF}{\TT_\GF}
\newcommand{\TyHH}{\Ty_\HH}
\newcommand{\TyGH}{\Ty_\GH}
\newcommand{\TyGF}{\Ty_\GF}
\newcommand{\CHH}{\CC_\HH}
\newcommand{\CGH}{\CC_\GH}
\newcommand{\CGF}{\CC_\GF}

\newcommand{\Gmodels}{\models_\gnd}
\newcommand{\Inf}{\mathit{Inf}}
\newcommand{\HInf}{\mathit{HInf}}

\newcommand{\GFInf}{\mathit{GFInf}}
\newcommand{\GHInf}{\mathit{GHInf}}

\newcommand{\concl}{\mathit{concl}}
\newcommand{\Undef}{\mathit{undef}}
\newcommand{\preconcl}{\mathit{preconcl}}
\newcommand{\prem}{\mathit{prems}}
\newcommand{\mprem}{\mathit{mprem}}
\newcommand{\HSel}{\mathit{HSel}}
\newcommand{\GHSel}{\mathit{GHSel}}
\newcommand{\GFSel}{\mathit{GFSel}}
\newcommand{\RedC}{\mathit{Red}_{\mathrm{C}}}
\newcommand{\RedI}{\mathit{Red}_{\mathrm{I}}}
\newcommand{\HRedC}{{\mathit{HRed}_{\mathrm{C}}}}
\newcommand{\HRedI}{\mathit{HRed}_{\mathrm{I}}}
\newcommand{\GHRedC}{\mathit{GHRed}_{\mathrm{C}}}
\newcommand{\GHRedI}{\mathit{GHRed}_{\mathrm{I}}}
\newcommand{\GFRedC}{\mathit{GFRed}_{\mathrm{C}}}
\newcommand{\GFRedI}{\mathit{GFRed}_{\mathrm{I}}}
\newcommand{\wrt}{\hbox{w.r.t.}}

\newcommand{\soundmodels}{\mathrel{|\kern-.1ex}\joinrel\approx}

\newcommand{\cmark}{\ding{51}}%
\let\xmark=\relax

\newcommand{\paper}{article}

\newcommand{\OK}[1]{#1}

\floatstyle{plain}
\newfloat{calculus}{p}{calc}

\newtheoremstyle{plainx}
  {6pt}   %
  {6pt}   %
  {\slshape}  %
  {0pt}       %
  {\bfseries} %
  {.}         %
  {5pt plus 1pt minus 1pt} %
  {}

\theoremstyle{plainx}
\newtheorem{theoremx}[thm]{Theorem}
\newtheorem{lemmax}[thm]{Lemma}
\newtheorem{corollaryx}[thm]{Corollary}

\theoremstyle{definition}

\newtheorem{definitionx}[thm]{Definition}
\newtheorem{notationx}[thm]{Notation}
\newtheorem{examplex}[thm]{Example}
\newtheorem{remarkx}[thm]{Remark}

\begin{document}

\title[Superposition for Lambda-Free Higher-Order Logic]{Superposition for
Lambda-Free Higher-Order Logic} %
\titlecomment{Extended version of Bentkamp et al., ``Superposition for lambda-free higher-order logic'' \cite{bentkamp-et-al-2018}}

\author[A.~Bentkamp]{Alexander Bentkamp\rsuper{a}}	%
%\address{Vrije Universiteit Amsterdam,
%Department of Computer Science,
%De Boelelaan 1111, 1081 HV Amsterdam, The Netherlands}	%
%

\author[J.~Blanchette]{Jasmin Blanchette\rsuper{a}}	%
\address{\lsuper{a}Vrije Universiteit Amsterdam,
Department of Computer Science,
De Boelelaan 1111, 1081 HV Amsterdam, The Netherlands}	%
\email{a.bentkamp@vu.nl}  %
\email{j.c.blanchette@vu.nl}  %

\author[S.~Cruanes]{Simon Cruanes\rsuper{b}}	%
\address{\lsuper{b}Aesthetic Integration,
600 Congress Ave., Austin, Texas, 78701, USA}	%
\email{simon@imandra.ai}

\author[U. Waldmann]{Uwe Waldmann\rsuper{c}}	%
\address{\lsuper{c}Max-Planck-Institut f\"ur Informatik, Saarland Informatics Campus, Saarbr\"ucken, Germany}
\email{uwe@mpi-inf.mpg.de}

\begin{abstract}
We introduce refutationally complete superposition calculi for
intentional and extensional clausal $\lambda$-free higher-order logic, two
formalisms that allow partial application and applied variables.
The calculi are parameterized by a term order that need not be fully
monotonic, making it possible to employ the $\lambda$-free higher-order
lexicographic path and Knuth--Bendix orders. We implemented the calculi in the
Zipperposition prover and evaluated them on Isabelle/HOL and TPTP benchmarks. They appear
promising as a stepping stone towards complete, highly efficient automatic theorem
provers for full higher-order logic.
\end{abstract}

\maketitle

\setcounter{footnote}{0}

\section{Introduction}
\label{sec:introduction}

Superposition is a highly successful calculus for reasoning about first-order
logic with equality. We are interested in \emph{graceful} generalizations to
higher-order logic:\ calculi that, as much as possible, coincide with standard
superposition on first-order problems and that scale up to arbitrary
higher-order problems.

As a stepping stone towards full higher-order logic,
in this \paper{} we restrict our attention to a
clausal $\lambda$-free fragment of polymorphic higher-order logic
that supports partial application and application of variables
(\Section~\ref{sec:logic}). This formalism is expressive
enough to permit the axiomatization of higher-order combinators such as
$\cst{pow} : \forallty{\alpha} \typ{nat} \to (\alpha \to \alpha) \to \alpha \to \alpha$
(intended to denote the iterated application~$h^n\:x$):
\begin{align*}
\cst{pow}\typeargs{\alpha}\; \cst{Zero}\; h & \approx \cst{id}\typeargs{\alpha}
&
\cst{pow}\typeargs{\alpha}\; (\cst{Succ}\; n)\; h\; x & \approx h\; (\cst{pow}\typeargs{\alpha}\; n\; h\; x)
\end{align*}
Conventionally,
functions are applied without
parentheses and commas, and variables are italicized. %
Notice the variable number of arguments to $\cst{pow}\typeargs{\alpha}$ %
and the application of~$h.$
The expressiveness
of full higher-order logic can be recovered
by introducing 
$\cst{SK}$%
-style%
\ combinators to represent $\lambda$-abstractions and proxies for the logical
symbols \cite{kerber-1991,robinson-1970}.

A widespread technique to support partial application and application of
variables in first-order logic is to make all symbols nullary and to represent
application of functions %
by a distinguished binary symbol $\cst{app} \oftypedecl \forallty{\alpha,\beta}
\typ{fun}(\alpha, \beta) \times \alpha \to \beta$,
where $\typ{fun}$ is an uninterpreted binary type constructor.
Following this scheme,
the higher-order term $\cst{f}\>(h\;\cst{f})$,
where $\cst{f} \oftype \kappa\to\kappa'$, is translated to
$\cst{app}(\cst{f}, \cst{app}(h, \cst{f}))$---or rather
$\cst{app}\typeargs{\kappa,\kappa'}(\cst{f},
\cst{app}\typeargs{\typ{fun}(\kappa,\kappa'),\kappa}(h, \cst{f}))$
if we specify the type arguments.
We call this the \emph{applicative encoding.}
The existence of such a reduction to first-order logic explains why
$\lambda$-free higher-order terms are also called ``applicative first-order
terms.''
Unlike for
full higher-order logic, most general unifiers are unique for
our \hbox{$\lambda$-free} fragment, just as they are for
applicatively encoded first-order terms.

Although the applicative encoding is complete \cite{kerber-1991} and is
employed fruitfully in tools such as HOLyHammer and Sledgehammer
\cite{blanchette-et-al-2016-qed}, it suffers from a number of weaknesses, all
related to its gracelessness. Transforming all the function symbols into
constants considerably restricts what can be achieved with term
orders; for example, argument tuples cannot easily be compared
using different methods for different symbols
\cite[\Section~2.3.1]{kop-2012-phd}. In a prover, the
encoding also clutters the data structures, slows down the algorithms,
and neutralizes the heuristics that look at the terms' root symbols. But our
chief objection is the sheer clumsiness of encodings and their
poor integration with interpreted symbols. %
And they quickly accumulate; for example, using the traditional encoding of
polymorphism relying on a distinguished binary function symbol $\cst{t}$
\cite[\Section~3.3]{blanchette-et-al-2016-types} in conjunction with the
applicative encoding, the term $\cst{Succ}\; \var{x}$ becomes
$\cst{t}(\cst{nat},\allowbreak \cst{app}(\cst{t}(\cst{fun}(\cst{nat},\allowbreak \cst{nat}),\allowbreak
\cst{Succ}),\allowbreak \cst{t}(\cst{nat}, \var{x})))$.
The term's simple structure is lost in translation.

Hybrid schemes have been proposed to strengthen the applicative encoding: If
a given symbol always occurs with at least $k$ arguments, these
can be passed directly \cite{meng-paulson-2008-trans}. However, this relies on
a closed-world assumption:\ that all terms that will ever be compared
arise in the initial problem. This noncompositionality conflicts with the
need for complete higher-order calculi to synthesize arbitrary terms during
proof search \cite{benzmueller-miller-2014}.
As a result, hybrid encodings are not an ideal basis for
higher-order automated reasoning.

Instead, we propose to generalize the superposition calculus to \emph{intensional} and
\emph{extensional} clausal $\lambda$-free higher-order logic. For the extensional
version of the logic, the property $\vthinspace(\forall x.\; h\>x \approx k\>x) \medrightarrow
h \approx k\vthinspace$ holds for all functions $h{,}\> k$ of the same type. For each logic, we
present two calculi (\Section~\ref{sec:the-calculi}). The intentional
calculi perfectly coincide with standard superposition on first-order
clauses; the extensional calculi depend on an extra axiom.

Superposition is parameterized by a term order, which is used to prune the
search space. If we assume that the term order is a simplification order
enjoying totality on ground terms (i.e., terms containing no term or type
variables), the standard calculus rules and completeness proof can be lifted
verbatim. The only necessary changes concern the basic definitions of terms and
substitutions.
However, there is one monotonicity property that is hard to obtain
unconditionally:\ \emph{compatibility with arguments}. It states that $s' \succ s$ implies $s'\>t \succ
s\;t$ for all terms $s{,}\> s'\!{,}\> t$ such that $s\;t$ and $s'\;t$ are well typed.
Blanchette, Waldmann, and colleagues recently introduced graceful
generalizations of the lexicographic path order (LPO)
\cite{blanchette-et-al-2017-rpo} and the Knuth--Bendix order (KBO)
\cite{becker-et-al-2017} with argument coefficients, but they both lack this
property. For example, given a KBO with $\cst g \succ \cst f$, it may well be
that $\cst g \; \cst a \prec \cst f \; \cst a$ if
$\cst f$ has a large enough multiplier %
on its argument.

Our superposition calculi are designed to be refutationally complete for such
nonmonotonic orders (\Section~\ref{sec:refutational-completeness}).
To achieve this, they include an inference rule for argument congruence, which
derives $C \llor s\;\var{x} \approx t\;\var{x}$ from $C \llor s \approx t$.
The redundancy criterion is defined in such a way that the larger,
derived clause is not subsumed by the premise. In the completeness proof, the
most difficult case is the one
that normally excludes superposition at or below variables using the induction
hypothesis. With nonmonotonicity, this approach no longer works, and we propose two
alternatives: Either perform some superposition inferences into higher-order
variables or ``purify'' the clauses to circumvent the issue. We refer to the
corresponding calculi as \emph{nonpurifying} and \emph{purifying.}

The calculi are implemented in the Zipperposition prover \cite{cruanes-2017}
(\Section~\ref{sec:implementation}). We evaluate them on
first-~and higher-order Isabelle/HOL \cite{boehme-nipkow-2010} and TPTP %
benchmarks \cite{sutcliffe-et-al-2012-tff,sutcliffe-et-al-2009} and compare
them with the applicative encoding (\Section~\ref{sec:evaluation}). We find
that there is a substantial cost associated with the applicative
encoding, that the nonmonotonicity is not particularly expensive, and that the
nonpurifying calculi outperform the purifying calculi.

\looseness=-1
An earlier version of this work was presented at IJCAR 2018 \cite{bentkamp-et-al-2018}.
This \paper{} extends the conference paper with detailed soundness and
completeness proofs and more explanations.
Because of too weak selection restrictions on the purifying calculi,
our claim of refutational completeness in the conference version was not entirely correct.
We now strengthened the selection restrictions accordingly.
Moreover, we extended the logic with polymorphism, leading
to minor modifications to the calculus. We also simplified the
presentation of the clausal fragment of the logic that interests us.
In particular, we removed mandatory arguments. %
The redundancy criterion also differs slightly from the conference version.
Finally,
we updated the empirical evaluation to reflect recent improvements in the
Zipperposition prover.

Since the publication of the conference paper, 
two research groups have built on our work.
Bhayat and Reger~\cite{bhayat-reger-2020-combsup}
extended our intensional nonpurifying calculus 
to a version of higher-order logic with combinators.
They use it with a nonmonotonic order that orients
the defining equations of $\cst{SK}$-style combinators.
Bentkamp et al.~\cite{bentkamp-et-al-lamsup-journal}
extended our extensional nonpurifying calculus
to a calculus on $\lambda$-terms.
For them, support for nonmonotonic orders is crucial as well
because no ground-total simplification order exists
for $\lambda$-terms up to $\beta$-conversion.
Based on these two extensions of our calculi, 
the provers Zipperposition and Vampire
achieved first and third place, respectively, in the higher-order category of
the 2020 edition of the CADE ATP System Competition (CASC-J10).

\section{Logic}
\label{sec:logic}

Our logic is intended as an intermediate step on the way towards full
higher-order logic
(also called simple type theory)
\cite{church-1940,gordon-melham-1993}.
Refutational completeness of calculi for higher-order logic is usually
stated in terms of Henkin semantics
\cite{henkin-1950,benzmueller-miller-2014}, in which the universes used to
interpret functions need only contain
the functions that can be expressed as terms. Since the terms of \hbox{$\lambda$-free}
higher-order logic exclude $\lambda$-abstractions, in ``\hbox{$\lambda$-free} Henkin
semantics'' the universes interpreting functions %
can be even smaller.
In that sense, our semantics resemble Henkin prestructures \cite[\Section~5.4]{leivant-1994}. %
In contrast to other higher-order logics \cite{vaananen-2019},
there are no comprehension principles, and we disallow nesting of Boolean
formulas inside terms.

\subsection{Syntax}

We fix a set $\Sigma_\mathsf{ty}$ of type constructors with arities and a set
$\VV_\mathsf{ty}$ of type variables.
We require at least one nullary type constructor
and a binary type constructor ${\fun}$ to be present in $\Sigma_\mathsf{ty}$.
We inductively define a $\lambda$-free higher-order type to be
either a type variable
$\alpha\in\Vty$ or of the form $\kappa(\tuple{\tau}_n)$
for an $n$-ary type constructor $\kappa\in\Sigmaty$ and types $\tuple{\tau}_n$.
Here and elsewhere, we write $\tuple{a}_n$ or $\tuple{a}$ to abbreviate the
tuple $(a_1,\dots,a_n)$ or product $a_1 \times \dots \times a_n$, for $n \ge 0$.
We write $\kappa$ for $\kappa()$ and $\tau\fun\upsilon$ for ${\fun}(\tau,\upsilon)$.
A \relax{type declaration} is
an expression of the form $\forallty{\tuple{\alpha}_m}\tau$ (or simply $\tau$
if $m=0$), where all type variables occurring in $\tau$ belong to~$\tuple{\alpha}_m$.

We fix a set
$\Sigma$ of symbols with type declarations,
written as $\cst{f}\oftypedecl\forallty{\tuple{\alpha}_m}\tau $ or~$\cst{f}$,
and a set $\VV$ of typed variables, written as
$\var{x}\oftype\tau$ or $\var{x}$.
We require $\Sigma$ to contain
a symbol with type declaration $\forallty{\alpha}\alpha$,
to ensure that the (Herbrand) domain of every type is nonempty.
The sets $(\Sigmaty,\Sigma)$ form the logic's signature.
We reserve the letters $s, t, u, v, w$ for terms and
$x, y, z$ for variables and write ${\oftype}\;\tau$
to indicate their type. The set of \hbox{$\lambda$-free} higher-order
terms is defined inductively as
follows. Every variable $x\oftype\tau\in\VV$ is a term. 
If $\cst{f}\oftypedecl
\forallty{\tuple{\alpha}_m}\tau$ is a symbol and
$\tuple{\upsilon}_m$ are types,
then $\cst{f}\typeargs{\tuple{\upsilon}_m}\oftype\tau\{\tuple{\alpha}_m \mapsto \tuple{\upsilon}_m\}$ 
is a term. If
$t\oftype\tau\fun\upsilon$ and $u\oftype\tau$, then $t\;u\oftype\upsilon$ is a term,
called an \emph{application}. Nonapplication terms are called
\emph{heads}.
Application is left-associative, and correspondingly
the function type constructor $\fun$ is right-associative.
Using the spine notation \cite{cervesato-pfenning-2003}, terms can be
decomposed in a unique way as a head $t$ applied to zero or more
arguments: $t\>s_1\dots s_n$ or $t\>\tuple{s}_n$ (abusing notation).
A term is \emph{ground} if it is built without using type or term variables.

Substitution and unification are generalized in the obvious way, without the
difficulties caused by $\lambda$-abstractions.
A substitution has the form $\{\tuple{\alpha}_m, \tuple{x}_n
\mapsto \tuple{\upsilon}_m, \tuple{s}_n\}$, where each $x_{\!j}$ has type $\tau_{\!j}$
and each $s_{\!j}$ has type $\tau_{\!j}\{\tuple{\alpha}_m \mapsto \tuple{\upsilon}_m\}$,
mapping $m$~type variables to $m$~types and $n$~term variables to $n$~terms.
A unifier of two terms $s$ and $t$ is a substitution $\rho$ such that $s\rho = t\rho$.
A most general unifier $\mgu(s,t)$ of two terms $s$ and $t$ is
a unifier $\sigma$ of $s$ and $t$ such that for every other unifier $\theta$, there exists a
substitution $\rho$ such that $\alpha\theta = \alpha\sigma\rho$ and $x\theta = x\sigma\rho$ 
for all $\alpha\in\Vty$ and all $x\in\VV$.
As in first-order logic, the most general unifier is unique up to variable renaming.
For example, $\mgu(x\>\cst{b}\>z,\>\cst{f}\>\cst{a}\>y\>\cst{c}) =
\{x\mapsto\cst{f}\>\cst{a}{,}\>y\mapsto \cst{b}{,}\> z\mapsto \cst{c}\}$,
and $\mgu(y\>(\cst{f}\>\cst{a}),\>\cst{f}\>(y\>\cst{a})) = \{y\mapsto\cst{f}\}$,
assuming that the types of the unified subterms are equal.

An equation $s \eq t$ is formally an unordered pair of terms $s$ and $t$ of the same type. A
literal is an equation or a negated equation, written $s\noteq t$. A clause $L_1 \lor \dots \lor L_n$ is a finite multiset of literals
$L_j$. The empty clause is written as $\bot$.

\subsection{Semantics}

A \emph{type interpretation} $\IIIty = (\UU, \IIty)$ is defined as follows.
The set $\UU$ is a nonempty collection of nonempty sets, called
\emph{universes}.
The function $\IIty$ associates a function
$\IIty(\kappa) : \UU^n \rightarrow \UU$
with each $n$-ary type constructor~$\kappa$.
A \emph{type valuation} $\xi$ is a function that maps every type variable to a universe.
The \emph{denotation} of a type for a type interpretation $\IIIty$
and a type valuation $\xi$ is defined by
$\interpret{\alpha}{\IIIty}{\xi}=\xi(\alpha)$ and
$\interpret{\kappa(\tuple{\tau})}{\IIIty}{\xi}=
\IIty(\kappa)(\interpret{\tuple{\tau}}{\IIIty}{\xi})$.
Here and elsewhere, we abuse notation by applying an operation on a tuple when it must be applied elementwise;
thus, $\interpret{\tuple{\tau}_n}{\IIIty}{\xi}$ stands for $\interpret{\tau_1}{\IIIty}{\xi},\dots, \interpret{\tau_n}{\IIIty}{\xi}$.

A type valuation $\xi$ can be extended to be a \emph{valuation} by additionally
assigning an element $\xi(x)\in\interpret{\tau}{\IIIty}{\xi}$ to each variable $x \oftype \tau$.
An \emph{interpretation function} $\II$ for a type interpretation $\IIIty$
associates with each symbol
$\cst{f}\oftypedecl\forallty{\tuple{\alpha}_m}\tau$ and universe tuple $\tuple{\U}_m\in\UU^m$
a value
$\II(\cst{f},\tuple{\U}_m) \in \interpret{\tau}{\IIIty}{\xi}$,
where $\xi$ is the type valuation that maps each $\alpha_i$ to $\U_i$.
Loosely following Fitting~\cite[\Section~2.5]{fitting-2002},
an \emph{extension function} $\EE$
associates to any pair of universes $\U_1,\U_2\in\UU$ a
function
$\EE_{\U_1,\U_2} : \IIty({\fun})(\U_1,\U_2)\rightarrow\allowbreak
(\U_1 \rightarrow \U_2)$.
Together, a type interpretation, an interpretation function, and an extension function
form an \emph{interpretation} $\III=(\UU,\IIty,\II,\EE)$.

An interpretation is \emph{extensional} if $\EE_{\U_1,\U_2}$ is
injective for all $\U_1, \U_2$.
Both intensional and extensional logics are widely used for interactive
theorem proving; for example, Coq's
calculus of inductive constructions is intensional \cite{bertot-casteran-2004},
whereas Isabelle/HOL is extensional \cite{nipkow-et-al-2002}.
The semantics is \emph{standard} if
$\EE_{\U_1,\U_2}$ is bijective for all $\U_1, \U_2$.

For an interpretation $\III = (\UU,\IIty,\II,\EE)$ and a valuation $\xi$, 
the denotation of a term is defined
as follows:
For variables $x$, let $\interpretaxi{x} \defeq \xi(x)$.
For symbols $\cst{f}$, let
$\interpretaxi{\cst{f}\typeargs{\tuple{\tau}} } \defeq \II(\cst{f},\interpret{\tuple{\tau}}{\IIIty}{\xi})$.
For applications $s\;t$ of a term $s \oftype \tau \fun \upsilon$
to a term $t\oftype \tau$, let
$\U_1 = \interpret{\tau}{\IIIty}{\xi}$,
$\U_2 = \interpret{\upsilon}{\IIIty}{\xi}$, and
$\interpretaxi{s\;t} \defeq \EE_{\U_1,\U_2}(\interpretaxi{s}) (\interpretaxi{t})$.
If $t$ is a ground term, we also write $\interpreta{t}$ for the denotation of $t$ because
it does not depend on the valuation.

An equation $s\eq t$ is true in $\III$ for $\xi$ if $\interpretaxi{s} = \interpretaxi{t}$;
otherwise, it is false.
A disequation $s\noteq t$ is true if $s\eq t$ is false.
A clause is true if at least one of its literals is true.
The interpretation $\III$ is a model of a clause $C$,
written $\III \models C$, if
$C$ is true in $\III$ for all valuations~$\xi$.
It is a model of a set of clauses if it is a model of all contained clauses.

For example, given the signature
$(\{\kappa,\fun\},\{\cst{a}\oftype\kappa\})$
and a variable $h:\kappa\fun\kappa$,
the clause $\var{h}\;\cst{a}\noteq\cst{a}$ has
an extensional model with $\UU = \{\U_1,\U_2\}$, $\U_1=\{a,b\}$
($a \not= b$), $\U_2=\{f\}$,
$\IIty(\kappa) = \U_1$, $\IIty({\fun})(\U_1, \U_1) = \U_2$,
$\II(\cst{a})=a$,
$\EE_{\U_1,\U_1}(f)(a) = \EE_{\U_1,\U_1}(f)(b) = b$.

\section{The Calculi}
\label{sec:the-calculi}

We introduce four versions of the \emph{Boolean-free $\lambda$-free higher-order
superposition calculus}, articulated along two axes:\ intentional versus
extensional, and nonpurifying versus purifying. To avoid repetitions, our
presentation unifies them into a single framework.

\subsection{The Inference Rules}
\label{ssec:the-inference-rules}

To support nonmonotonic term orders, we
restrict superposition inferences to \emph{green subterms},
which are defined inductively as follows:
\begin{definitionx}[Green subterms and contexts] \label{def:green-subterm}
A term $t'$ is a \emph{green subterm} of $t$ if
$t = t'$ or if
$t = s\;\tuple{u}$
and $t'$ is a green subterm of $u_i$ for some $i$.
We write $\fosubterm{s}{u}$ to indicate that
the subterm $u$ of $\hosubterm{s}{u}$ is a green subterm.
Correspondingly, we call the context $\fosubterm{s}{\phantom{.}}$ around a green subterm a \emph{green context}.
\end{definitionx}
By this definition, $\cst{f}$ and $\cst{f}\;\cst{a}$ are subterms of
$\cst{f}\;\cst{a}\;\cst{b}$, but not green subterms.
The green subterms of $\cst{f}\;\cst{a}\;\cst{b}$ are $\cst{a}$, $\cst{b}$, and $\cst{f}\;\cst{a}\;\cst{b}$.
Thus, $[\phantom{i}]\;\cst{a}\;\cst{b}$ and $[\phantom{i}]\;\cst{b}$ are not
green contexts, but $\cst{f}\;[\phantom{i}]\;\cst{b}$, $\cst{f}\;\cst{a}\;[\phantom{i}]$,
and $[\phantom{i}]$ are.

The calculi are parameterized by a partial order $\succ$ on terms
that
\begin{itemize}
  \item is well founded on ground terms;
  \item is total on ground terms;
  \item has the subterm property on ground terms;
  \item is \emph{compatible with green contexts} on ground terms: 
    $t' \succ t$ implies $\fosubterm{s}{t'} \succ \fosubterm{s}{t}$;
  \item is stable under
  grounding substitutions: 
  $t \succ s$ implies $t\theta \succ s\theta$ for all
  substitutions~$\theta$ grounding $t$ and $s$.
\end{itemize}
The order need not be \emph{compatible with arguments}:\
$s' \succ s$ need not imply $s'\>\relax{t} \succ s\;\relax{t}$,
even on ground terms.
The literal and clause orders %
are defined from $\succ$ as multiset extensions in the
standard way \cite{bachmair-ganzinger-1994}.
Despite their names, the term, literal, and clause orders need not be
transitive on nonground entities.

The $\lambda$-free higher-order generalizations of LPO
\cite{blanchette-et-al-2017-rpo} and KBO~\cite{becker-et-al-2017} as well as EPO
\cite{bentkamp-2018-epo-formalization} fulfill these requirements, with the caveat that they
are defined on untyped terms. To use them on polymorphic terms, we can
encode type arguments as term arguments and ignore the remaining type
information.

Literal selection is supported.
The selection function maps each clause $C$ to a subclause of $C$
consisting of negative literals.
A literal $L$ is (\emph{strictly}) \emph{eligible} \wrt\ a substitution $\sigma$ in $C$ if it is
selected in $C$ or there are no selected literals in $C$ and $L\sigma$ is
(strictly) maximal in $C\sigma.$
If $\sigma$ is the identity substitution, we leave it implicit.

The following four rules are common to all four calculi.
We regard positive and negative superposition as two cases
of the same rule
\[\namedinference{Sup}
{\overbrace{D' \mathrel\lor { t \eq t'}}^{\vphantom{\cdot}\smash{D}} \hypsep
 \overbrace{C' \mathrel\lor s\leftfosubterm u\rightfosubterm \doteq s'}^{\smash{C}}}
{(D' \mathrel\lor C' \mathrel\lor s\leftfosubterm t'\rightfosubterm \doteq s')\sigma}\]
where $\doteq$ denotes $\eq$ or $\noteq$ and the following conditions are fulfilled:

\begin{enumerate}
  \item[1.] $\sigma = \mgu(t,u)$;%
  \hfill 2.\enskip $t\sigma \not\preceq t'\sigma$;%
  \hfill 3.\enskip $s\leftfosubterm u\rightfosubterm\sigma \not\preceq s'\sigma$ ;%
  \hfill\hbox{}
  \item[4.] $t \eq t'$ is strictly eligible \wrt\ $\sigma$ in $D$;%
  \hfill 5.\enskip $C\sigma \not\preceqC D\sigma$;%
  \hfill\hbox{} 
  \item[6.] $s\leftfosubterm u\rightfosubterm \doteq s'$ is eligible \wrt\ $\sigma$ in $C$
  and, if positive, even strictly eligible;%
  \item[7.] the \emph{variable condition} must hold, which 
   is specified individually for each calculus below.
\end{enumerate}
In each calculus, we will define the variable condition 
to coincide with the condition ``$u \not\in \VV$'' if the premises are first-order.

The equality resolution and equality factoring rules are almost identical to
their standard counterparts:
\[
\namedinference{ERes}
{\overbrace{C' \mathrel\lor {s \noteq s'}}^C}
{C'\sigma}
\]
where
 \begin{enumerate}
  \item[1.] $\sigma = \mgu(s,s')$;%
  \hfill 2.\enskip $s \noteq s'$ is eligible \wrt\ $\sigma$ in $C$.%
  \hfill\hbox{}
 \end{enumerate}
\medskip

\[
 \namedinference{EFact}
{\overbrace{C' \mathrel\lor {s'} \eq t' \mathrel\lor {s} \eq t}^C}
{(C' \mathrel\lor t \noteq t' \mathrel\lor s \eq t')\sigma}
\]
where
 \begin{enumerate}
   \item[1.] $\sigma = \mgu(s,s')$;%
  \hfill 2.\enskip $s\sigma \not\preceq t\sigma$;%
  \hfill 3.\enskip $s \eq t$ is eligible \wrt\ $\sigma$ in $C$.%
  \hfill\hbox{}
 \end{enumerate}
\medskip

The following \emph{argument congruence} rule compensates for the limitation
that the superposition rule applies only to green subterms:
\[\namedinference{ArgCong}
{\overbrace{C' \mathrel\lor s \eq s'}^C}
{C'\sigma \lor s\sigma\;\tuple{x} \eq s'\sigma\>\tuple{x}}\]
where
\begin{enumerate}
  \item[1.] $s \eq s'$ is strictly eligible \wrt\ $\sigma$ in $C$;%
  \item[2.] $\tuple{x}$ is a nonempty tuple of distinct fresh variables;%
  \item[3.] $\sigma$ is the most general type substitution that ensures
  well-typedness of the conclusion.
\end{enumerate}
In particular, if $s$ takes $m$ arguments,
there are $m$ \infname{ArgCong} conclusions for this literal,
for which $\sigma$ is the identity and $\tuple{x}$ is a tuple of 1, \dots, $m -1$, or $m$ variables.
If the result type of $s$ is a type variable, we have in addition infinitely many
\infname{ArgCong} conclusions,
for which $\sigma$ instantiates the type variable in the result type of $s$
with $\tuple{\alpha}_k \fun \beta$
for some $k>0$ and fresh type variables $\tuple{\alpha}_k$ and $\beta$ and
for which $\tuple{x}$ is a tuple of $m + k$ variables.
In practice, the enumeration of the infinitely many conclusions 
must be interleaved with other inferences via some form of dovetailing.

For the \textbf{intensional nonpurifying}
calculus, the variable condition
of the \infname{Sup} rule is as follows:
\begin{quote}
Either $u \notin \VV$ or there exists a grounding substitution~$\theta$ with $t\sigma\theta \succ t'\sigma\theta$
and $C\sigma\theta \precC C\{u\mapsto t'\}\sigma\theta$.
\end{quote}
This condition generalizes the standard condition that $u \notin \VV\!$.
The two coincide if $C$ is first-order or if the term order is monotonic. In some cases
involving nonmonotonicity, the variable condition effectively mandates
\infname{Sup} inferences at variable positions of the right premise, but never below.
We will call theses inferences \emph{at variables}.

For the \textbf{extensional nonpurifying} calculus,
the variable condition uses the following definition.
\begin{definitionx} \label{def:jell}
  A term of the form $x\;\tuple{s}_n$, for $n \ge 0$, \emph{jells} with a literal $t \eq t' \in D$ if
  $t=\tilde{t}\;\tuple{y}_n$ and $t'=\tilde{t}\vthinspace'\>\tuple{y}_n$
  for some terms $\tilde{t}$, $\tilde{t}\vthinspace'$ and distinct variables
  $\tuple{y}_n$ that do not occur elsewhere in $D.$
\end{definitionx}%
\noindent
Using the naming
convention from Definition~\ref{def:jell} for $\tilde{t}\vthinspace'$,
the variable condition can be stated as follows:
\begin{quote}
  If $u$ has a variable head $x$ and jells with the literal $t \eq t' \in D$,
  there must exist a grounding substitution~$\theta$ with $t\sigma\theta \succ t'\sigma\theta$ and
  $C\sigma\theta \precC C''\sigma\theta$, where $C'' = C\{x\mapsto \tilde{t}\vthinspace'\}$.
\end{quote}
If $C$ is first-order, this amounts to $u \notin \VV\!$.
Since the order is compatible with green contexts,
the substitution $\theta$ can exist only if $x$ occurs applied in $C$.

Moreover, the extensional nonpurifying calculus has one additional rule,
the positive extensionality rule, and one axiom, the extensionality axiom.
The rule is
    \[\namedinference{PosExt}
    {C' \mathrel\lor s\;\tuple{x} \eq s'\>\tuple{x}}
    {C' \mathrel\lor s \eq s'}\]
where 
\begin{enumerate}
\item[1.] $\tuple{x}$ is a tuple of distinct variables that do not occur in $C'$, $s$, or $s'$
\item[2.] $s\;\tuple{x} \eq s'\>\tuple{x}$ is strictly eligible in the premise.
\end{enumerate}

The extensionality axiom uses a
polymorphic
Skolem symbol~$\diff
  : \forallty{\alpha,\beta}{(\alpha\fun\beta)}^2\fun{\alpha}$ characterized by the axiom
\begin{align*}
  x\;(\diff\typeargs{\alpha,\beta}\;x\;y )
  \noteq y\>(\diff\typeargs{\alpha,\beta}\;x\;y ) \llor x \eq y
  \tag{\infname{Ext}}
\end{align*}

Unlike the nonpurifying calculi, the purifying calculi never perform superposition at variables.
Instead, they rely on purification
\cite{Brand1975,DigricoliHarrison1986,SchmidtSchauss1989,SnyderLynch1991}
(also called abstraction) to circumvent nonmonotonicity.
The idea
is to rename apart problematic occurrences of a variable~$x$ in
a clause to $x_1, \dots, x_n$ and to add \emph{purification literals}
$x_1 \noteq x$, \dots, $x_n \noteq x$ to connect the new variables to $x$.
We must then ensure
that all clauses are purified, by processing the initial clause set
and the conclusion of every inference or simplification.

In the \textbf{intensional purifying} calculus, the purification
$\pureext(C)$ of clause~$C$ is defined as the result of the following procedure.
Choose a variable $x$ that occurs applied in $C$ and also unapplied in
a literal of $C$ that is not of the form $x \noteq y$.
If no such variable exists, terminate.
Otherwise, replace all unapplied occurrences of
$x$ in $C$ by a fresh variable $x'$ and add the purification literal $x'\noteq x$.
Then repeat these steps with another variable.
The procedure terminates because the number of variables that can be chosen reduces with each step.
For example,
\[\pureint(x\;\cst{a} \eq x\;\cst{b} \mathrel\lor \cst{f}\;x \eq \cst{g}\;x)
\vthinspace=\vthinspace x\;\cst{a} \eq x\;\cst{b} \mathrel\lor \cst{f}\;x' \eq \cst{g}\;x' \mathrel\lor x \noteq x'\]
The variable condition is standard:
\begin{quote}
The term $u$ is not a variable.
\end{quote}
The conclusion $C$ of \infname{ArgCong} is changed to $\pureint(C)$; the other
rules preserve purity of their premises.

In the \textbf{extensional purifying} calculus, $\pureint(C)$ is defined
as follows.
Choose a variable~$x$ occurring in green subterms
$x\;\tuple{u}$ and $x\;\tuple{v}$ in literals of $C$
that are not of the form $x \noteq y$,
where $\tuple{u}$ and $\tuple{v}$ are distinct
(possibly empty) term tuples.
If no such variable exists, terminate.
Otherwise, replace all green subterms $x\;\tuple{v}$ with
$x'\>\tuple{v}$, where $x'$ is fresh, and add the purification literal $x'\noteq x$.
Then repeat the procedure until no variable fulfilling the requirements is left.
The procedure terminates because the number of variables fulfilling the requirements
does not increase and with each step, the number of distinct tuples
$\tuple{u}$ and $\tuple{v}$ fulfilling the requirements decreases for the chosen variable $x$.
For example,
\[\pureext(x\;\cst{a} \eq x\;\cst{b} \mathrel\lor \cst{f}\;x \eq \cst{g}\;x)
\vthinspace=\vthinspace x\;\cst{a} \eq x'\;\cst{b} \mathrel\lor \cst{f}\;x''
\eq \cst{g}\;x'' \mathrel\lor x' \noteq x \mathrel\lor x'' \noteq x\]
Like the extensional nonpurifying calculus, this calculus also
contains the \infname{PosExt} rule and axiom (\infname{Ext}) introduced above.
The variable condition is as follows: 
\begin{quote}
Either $u$ has a
nonvariable head or $u$ does not jell with the literal $t \eq t' \in D.$
\end{quote}
The conclusion $E$ of each rule is changed to $\pureext(E)$, except for
\infname{PosExt}, which preserves purity.

Finally, we impose further restrictions on literal selection.
In the nonpurifying calculi, a literal must not be
selected if $x\;\tuple{u}$ is a maximal term of the clause and the literal
contains a green subterm $x\;\tuple{v}$ with $\tuple{v}\neq\tuple{u}$.
In the purifying calculi, a literal must not be selected if it
contains a variable of functional type.
These restrictions are
needed in our completeness proof to show that superposition inferences below variables are redundant. 
For the purifying calculi, the restriction is also needed to avoid inferences 
whose conclusion is identical to their premise, such as
\[\namedinference{\infname{ERes}}
{z\>\cst{a} \eq \cst{b} \llor z \noteq x \llor \cst{f} \noteq x}
{z\>\cst{a} \eq \cst{b} \llor z' \noteq \cst{f} \llor z' \noteq z}\]
where $\cst{f} \noteq x$ is selected.
For the nonpurifying calculi, it might be possible to avoid the restriction at
the cost of a more elaborate argument.

\begin{examplex}
We illustrate the different variable condition with a few examples.
Consider a left-to-right LPO \cite{blanchette-et-al-2017-rpo} instance
with precedence $\cst{h} \succ \cst{g} \succ \cst{f} \succ \cst{c} \succ \cst{b} \succ
\cst{a}$. 
Given the clauses $D$ and $C$ of a superposition inference in which the
grayed
subterms are unified, we list below whether each calculus's
variable condition is fulfilled (\cmark) or not.

\vskip\abovedisplayskip

\centerline{%
\begin{tabular}{ c c c c c c c }
\toprule
$D$ & $C$ & 
\small\shortstack{\vphantom{X$^X$}intensional\\nonpurifying} & 
\small\shortstack{\vphantom{X}extensional\\nonpurifying} & 
\small\shortstack{\vphantom{X}intensional\\purifying} & 
\small\shortstack{\vphantom{X}extensional\\purifying}\\
\midrule
$\unified{\cst{h}} \eq \cst{g}$ &
$\cst{f}\>\unified{y} \eq \cst{c}$ 
& \xmark & \xmark 
& \xmark & \xmark \\[\jot]
$\unified{\cst{f}\;(\cst{h}\;\cst{a})} \eq \cst{h}$ &
$\cst{g}\;\unified{\var{y}} \eq \var{y}\;\cst{b}$
& \cmark %
& \cmark %
& \xmark %
& \xmark %
\\[\jot]
$\unified{\cst{h}\;\var{x}} \eq \cst{f}\>\var{x}$
& $\cst{g}\;(\unified{\var{y}\;\cst{b}})\;\var{y} \eq \cst{a}$ 
& \cmark %
& \xmark %
& \cmark %
& \xmark %
\\[\jot]
$x\eq\cst{c} \llor \unified{\cst{h}\;\var{x}} \eq \cst{f}\>\var{x}$
& $\cst{g}\;(\unified{\var{y}\;\cst{b}})\;\var{y} \eq \cst{a}$ 
& \cmark %
& \cmark %
& \cmark %
& \cmark %
\\[\jot]
  \bottomrule
\end{tabular}}

\vskip\belowdisplayskip

\noindent
For the purifying calculi,
the clauses would actually undergo purification first,
but this has no impact on the variable conditions.
In the last row, the term $y\>\cst{b}$ does not jell with $\cst{h}\>x \eq\cst{f}\>x \in D$ because $x$ occurs also in the first literal of $D$.
\end{examplex}

\begin{remarkx}
In descriptions of first-order logic with equality, the property
$y \eq y' \medrightarrow \cst{f}(\bar x, y, \bar z) \eq \cst{f}(\bar x, y'\!, \bar z)$
is often referred to as ``function congruence.'' It seems natural to use the
same label for the higher-order version
$t \eq t' \medrightarrow s\;t \eq s\;t'$
and to call the companion property
$s \eq s' \medrightarrow s\;t \eq s'\>t$
``argument congruence,'' whence the name \infname{ArgCong} for our inference
rule. This nomenclature is far from universal; for example, the Isabelle/HOL
theorem \textit{fun\_cong} captures argument congruence and \textit{arg\_cong}
captures function congruence.
\end{remarkx}

\subsection{Rationale for the Inference Rules}
\label{ssec:rationale-for-the-inference-rules}

A key restriction of all four calculi is that they superpose only at
green subterms, %
mirroring the term order's compatibility with green contexts. The
\infname{ArgCong} rule then makes it possible to simulate superposition at
nongreen subterms. However, in conjunction with the \infname{Sup} rules,
\infname{ArgCong} can exhibit an unpleasant behavior, which we call
\emph{argument congruence explosion}:
\begin{align*}
\prftree[l]{\small{\infname{Sup}}}{
  \prftree[l]{\small{\infname{ArgCong}}}
  {\strut\cst{g}\eq \cst{f}}
  {\strut\cst{g}\;\var{x} \eq \cst{f}\>\var{x} }
}
{\strut\var{h}\;\cst{a} \noteq \cst{b}}
{\strut\cst{f}\;\cst{a} \noteq \cst{b}}
\rulesep
\prftree[l]{\small{\infname{Sup}}}{
  \prftree[l]{\small{\infname{ArgCong}}}
  {\strut\cst{g}\eq \cst{f}}
  {\strut\cst{g}\;\var{x}\;\var{y}\;\var{z} \eq \cst{f}\>\var{x}\;\var{y}\;\var{z} }
}
{\strut\var{h}\;\cst{a} \noteq \cst{b}}
{\strut\cst{f}\>\var{x}\;\var{y}\;\cst{a} \noteq \cst{b}}
\end{align*}
In both derivation trees, the higher-order variable $\var{h}$ is effectively the target
of a \infname{Sup} inference. Such derivations essentially amount to
superposition at variable positions (as shown on the left) or even
superposition below variable positions (as shown on the right), both of which
can be extremely prolific. In standard superposition, the explosion
is averted by the condition on the \infname{Sup} rule that $u \notin
\VV\negvthinspace.$
In the extensional purifying calculus, the variable condition
tests that either $u$ has a nonvariable head or
$u$ does not jell with the literal $t \eq t' \in D,$
which %
prevents derivations such as the above.
In the corresponding nonpurifying calculus, some such derivations may need to
be performed when the term order exhibits nonmonotonicity for the terms of
interest.

In the intensional calculi, the explosion can arise because the variable conditions are weaker.
The following example shows that the intensional nonpurifying calculus would be incomplete
if we used the variable condition of the extensional nonpurifying calculus.

\begin{examplex}
Consider a left-to-right LPO \cite{blanchette-et-al-2017-rpo} instance
with precedence $\cst{h} \succ \cst{g} \succ \cst{f} \succ \cst{b} \succ
\cst{a}$, and consider the following unsatisfiable clause set:
\begin{align*}
\cst{h}\;\var{x} & \eq \cst{f}\>\var{x}
& \cst{g}\;(\var{x}\;\cst{b})\;\var{x} & \eq \cst{a}
& \cst{g}\;(\cst{f}\;\cst{b})\;\cst{h} & \noteq \cst{a}
\end{align*} %
The only possible inference is a \infname{Sup} inference of the first into the second clause,
but the variable condition of the extensional nonpurifying calculus is not met.
\end{examplex}

\noindent
It is unclear whether the variable condition
of the intensional purifying calculus could be strengthened,
but our completeness proof suggests that it cannot.

The variable conditions in the extensional calculi are
designed to prevent the argument congruence explosion shown above, but since they
consider only the shape of the clauses, they might also block \infname{Sup}
inferences whose side premises do not originate from \infname{ArgCong}.
This is why we need the \infname{PosExt} rule.
\begin{examplex}
In the following unsatisfiable clause set, the only possible inference from these clauses
in the extensional nonpurifying calculus is \infname{PosExt}, showing its
necessity:
\begin{align*}
\cst{g}\;\var{x} & \eq \cst{f}\>\var{x}
& \cst{g} & \noteq \cst{f}
& x\;(\diff\typeargs{\alpha,\beta}\;x\;y)
\noteq y\>(\diff\typeargs{\alpha,\beta}\;x\;y) \llor x \eq y
\end{align*} %
The same argument applies for the purifying calculus with the difference that
the third clause must be purified.
\end{examplex}

\noindent
Due to nonmonotonicity, for refutational completeness we need either to purify the clauses
or to allow some superposition at variable positions, as mandated by the respective
variable conditions. Without either of these measures, at least the extensional
calculi and presumably also the intensional calculi would be incomplete, as the next example demonstrates.

\begin{examplex}
Consider the following clause set:
\begin{equation*}
\begin{gathered}
\cst{k}\;(\cst{g}\;\var{x}) \eq \cst{k}\;(\var{x}\;\cst{b}) \quad
\cst{k}\;(\cst{f}\;(\cst{h}\;\cst{a})\;\cst{b}) \noteq \cst{k}\;(\cst{g}\;\cst{h}) \quad
\cst{f}\;(\cst{h}\;\cst{a}) \eq \cst{h} \quad
\cst{f}\;(\cst{h}\;\cst{a})\;\var{x} \eq \cst{h}\;\var{x}\\
x\;(\diff\typeargs{\alpha,\beta}\;x\;y)
\noteq y\>(\diff\typeargs{\alpha,\beta}\;x\;y) \llor x \eq y
\end{gathered}
\end{equation*}%
Using a left-to-right LPO \cite{blanchette-et-al-2017-rpo} instance
with precedence
$\cst{k} \succ \cst{h} \succ \cst{g} \succ \cst{f} \succ \cst{b} \succ \cst{a}$,
this clause set is saturated \wrt\  the extensional purifying calculus when omitting
purification. It also quickly saturates using the extensional
nonpurifying calculus when omitting \infname{Sup} inferences at variables.
By contrast, the intensional calculi derive $\bot$, even
without purification and without \infname{Sup} inferences at variables, because of the less restrictive variable conditions.
\end{examplex}

\noindent
This raises the question as to whether the intensional calculi actually need to purify
or to perform \infname{Sup} inferences at variables. We conjecture that omitting purification and
\infname{Sup} inferences at variables in the intensional calculi is complete when redundant
clauses are kept but that it is incomplete in general.

We initially considered inference rules instead of axiom (\infname{Ext}).
However, we did not find a set of inference rules that is complete and leads to
fewer inferences than (\infname{Ext}). We considered the $\infname{PosExt}$
rule described above in combination with the following rule:
\begin{align*}
&\namedinference{NegExt}{C \llor s \noteq t}
{C \llor s\>(\cst{sk}\typeargs{\tuple{\alpha}}\>\tuple{x}_n)
\noteq t\>(\cst{sk}\typeargs{\tuple{\alpha}}\>\tuple{x}_n)}\end{align*}
where $\cst{sk}$ is a fresh Skolem symbol and
$\tuple{\alpha}$ and $\tuple{x}_n$ are the type and term variables
occurring free in the literal $s \noteq t$.
However, these two rules do not suffice for a refutationally complete
calculus, as the following example demonstrates:

\begin{examplex}
Consider the clause set
\begin{align*}
&\cst{f}\;x\eq \cst{a}&
&\cst{g}\;x\eq \cst{a}&
&\cst{h}\;\cst{f}\eq \cst{b}&
&\cst{h}\;\cst{g}\noteq \cst{b}
\end{align*}
Assuming that all four equations are oriented from left to right, this set is
saturated \wrt\ the extensional calculi if (\infname{Ext})
is replaced by \infname{NegExt}; yet it is unsatisfiable in an extensional
logic.
\end{examplex}

\begin{examplex}
\label{ex:add-l-r}
A significant advantage of our calculi over the use of standard superposition
on applicatively encoded problems is the flexibility they offer in orienting
equations. The following equations provide two definitions of addition on Peano
numbers:
\begin{align*}
  \cst{add_L}\,\cst{Zero}\;\var{y} & \eq \var{y}
& \cst{add_R}\;\var{x}\;\cst{Zero} & \eq \var{x}
\\[-.5\jot] %
  \cst{add_L}\,(\cst{Succ}\;x)\;\var{y} & \eq \cst{add_L}\,x\;(\cst{Succ}\;\var{y})
& \cst{add_R}\;\var{x}\;(\cst{Succ}\;\var{y}) & \eq \cst{add_R}\;(\cst{Succ}\;\var{x})\;y
\end{align*}
Let $\cst{add_L}\,(\smash{\cst{Succ}^{100}}\;\cst{Zero})\;\cst{n} \noteq
\cst{add_R}\;\cst{n}\;(\smash{\cst{Succ}^{100}}\;\cst{Zero})$ be the negated conjecture.
With LPO, we can use a left-to-right comparison for $\cst{add_L}$'s arguments
and a right-to-left comparison for $\cst{add_R}$'s arguments to orient all
four equations from left to right. Then the negated conjecture can be simplified to
$\cst{Succ}^{100}\, \cst{n} \noteq \cst{Succ}^{100}\, \cst{n}$ by simplification
(demodulation), and $\bot$ can be derived with a single inference.
If we use the applicative encoding instead, there is no instance of LPO or KBO
that can orient both recursive equations from left to right. For at least one of
the two sides of the negated conjecture, simplification is replaced by
100~\infname{Sup} inferences, which is much less efficient, especially in the
presence of additional axioms.

\end{examplex}

\subsection{Soundness}
\label{ssec:soundness}

To show the inferences' soundness, we need the substitution lemma for our logic:

  \begin{lemmax}[Substitution lemma]
    \label{lem:subst-lemma-general}
    Let $\III = (\UU, \IIty,\II,\EE)$ be a $\lambda$-free higher-order interpretation. Then
    \[\interpret{\tau\negvthinspace\rho}{\IIIty}{\xi} = \interpret{\tau}{\IIIty}{\xi'}
    \text{ and }
    \interpret{t\rho}{\III}{\xi} = \interpret{t}{\III}{\xi'}\]
    for all terms $t$, all types $\tau$, and all substitutions
    $\rho$,
    where $\xi'(\alpha) = \interpret{\alpha\rho}{\IIIty}{\xi}$ for all type
    variables $\alpha$ and $\xi'(x) = \interpret{x\rho}{\III}{\xi}$ for all term
    variables $x$.
  \end{lemmax}
  \begin{proof}
    First, we prove that $\interpret{\tau\negvthinspace\rho}{\IIIty}{\xi} =
    \interpret{\tau}{\IIIty}{\xi'}$ by induction on the structure of $\tau$.
    If $\tau = \alpha$ is a type variable,
    \[\interpret{\alpha\rho}{\IIIty}{\xi} = \xi'(\alpha) = \interpret{\alpha}{\IIIty}{\xi'}\]
    If $\tau = \kappa(\tuple{\upsilon})$ for some type constructor $\kappa$ and types $\tuple{\upsilon}$,
    \[\interpret{\kappa(\tuple{\upsilon})\rho}{\IIIty}{\xi}
    = \IIty(\kappa)(\interpret{\tuple{\upsilon}\rho}{\IIIty}{\xi})
    \eqIH
      \IIty(\kappa)(\interpret{\tuple{\upsilon}}{\IIIty}{\xi'})
      = \interpret{\kappa(\tuple{\upsilon})}{\IIIty}{\xi'}\]

    Next, we prove $\interpret{t\rho}{\III}{\xi} = \interpret{t}{\III}{\xi'}$
    by structural induction on $t$.	
    If~$t = y$, then by the definition of the denotation of a variable
    \[\interpret{y\rho}{\III}{\xi} = \xi'(y) = \interpret{y}{\III}{\xi'}\]
    If $t = \cst{f}\typeargs{\tuple{\tau}}$, then by the definition of the term denotation
    \[\interpret{\cst{f}\typeargs{\tuple{\tau}}\rho}{\III}{\xi} =
    \II(\cst{f},\interpret{\tuple{\tau}\rho}{\IIIty}{\xi})
    =
    \II(\cst{f},\interpret{\tuple{\tau}}{\IIIty}{\xi'}) =
    \interpret{\cst{f}\typeargs{\tuple{\tau}}}{\III}{\xi'}\]
    If $t = u\>v$, then by the definition of the term denotation
    \[\interpret{(u\>v)\rho}{\III}{\xi}
    = \EE_{U_1,U_2}(\interpret{u\rho}{\III}{\xi})(\interpret{v\rho}{\III}{\xi})
    \eqIH  \EE_{U_1,U_2}(\interpret{u}{\III}{\xi'})(\interpret{v}{\III}{\xi'})
    =  \interpret{u\>v}{\III}{\xi'}\]
    where $u$ is of type $\tau \fun \upsilon$,
    $U_1 = \interpret{\tau\negvthinspace\rho}{\IIIty}{\xi} =
    \interpret{\tau}{\IIIty}{\xi'}$, and
    $U_2 = \interpret{\upsilon\rho}{\IIIty}{\xi} =
    \interpret{\upsilon}{\IIIty}{\xi'}$.
  \end{proof}

  \begin{lemmax}\label{lem:apply-subst}
  If $\III\models C$ for some interpretation $\III$ and some clause $C$, then $\III\models C\rho$ for all substitutions $\rho$.
  \end{lemmax}
  \begin{proof}
  We need to show that $C\rho$ is true in $\III$ for all valuations $\xi$. Given a valuation $\xi$,
  define $\xi'$ as in Lemma~\ref{lem:subst-lemma-general}.
  Then, by Lemma~\ref{lem:subst-lemma-general}, a literal in $C\rho$ is true in $\III$ for $\xi$ if and
  only if the corresponding literal in $C$ is true in $\III$ for $\xi'$.
  There must be at least one such literal because $\III \models C$ and hence $C$ is in particular true in $\III$ for $\xi'$.
  Therefore, $C\rho$ is true in $\III$ for $\xi$.
  \end{proof}

\begin{theoremx}[Soundness of the intensional calculi]
  The inference rules \infname{Sup}, \infname{ERes},
  \infname{EFact}, and \infname{Arg\-Cong} are sound
  {\upshape(}even without the variable condition and the side conditions on order and eligibility{\upshape)}.
  \label{thm:intensional-soundness}
\end{theoremx}
  \begin{proof}
  We fix an inference and an interpretation $\III$
  that is a model of the premises.
  We need to show that it is also a model of the conclusion.

  From the definition of the denotation of a term, it is obvious that
  congruence holds at all subterms in our logic.
  By Lemma~\ref{lem:apply-subst}, $\III$ is a model of the $\sigma$-instances of the premises as well,
  where $\sigma$ is the substitution used for the inference.
  Fix a valuation $\xi$.
  By making case distinctions on the truth in $\III$ under $\xi$ of the literals
  of the $\sigma$-instances of the premises, using the conditions that $\sigma$ is a unifier, and applying congruence,
  it follows that the conclusion is also true in $\III$ under $\xi$.
\end{proof}

\begin{theoremx}[Soundness of the extensional calculi]
  \label{thm:extensional-soundness}
  The inference rules \infname{Sup}, \infname{ERes},
  \infname{EFact}, \infname{Arg\-Cong}, and \infname{PosExt} are sound \wrt\ extensional interpretations
  {\upshape(}even without the variable condition and the side conditions on order and eligibility{\upshape)}.
\end{theoremx}
\begin{proof}
  We only need to prove \infname{PosExt} sound.
  For the other rules, we can proceed as in Theorem~\ref{thm:intensional-soundness}.
  By induction on the length of $\tuple{x}$,
  it suffices to prove \infname{PosExt} sound for one variable $x$ instead of a tuple $\tuple{x}$.
  We fix an inference and an extensional interpretation $\III$
  that is a model of the premise $C' \llor s\>x \eq s'\>x$.
  We need to show that it is also a model of the conclusion $C' \llor s \eq s'$.

 Let $\xi$ be a valuation.
 If $C'$ is true in $\III$ under $\xi$, the conclusion is clearly true as well.
 Otherwise $C'$ is false in $\III$ under $\xi$,
 and also under $\xi[x\mapsto a]$ for all $a$ because $x$ does not occur in $C'$.
 Since the premise is true in $\III$,
 $s\>x = s'\>x$ must be true in $\III$ under $\xi[x\mapsto a]$ for all $a$.
 Hence, for appropriate universes $\U_1,\U_2$, we have
 $
 \EE_{\U_1,\U_2}(\interpret{s}{\III}{\xi[x\mapsto a]}) (a) =
 \interpret{s\>x}{\III}{\xi[x\mapsto a]} = \interpret{s'\>x}{\III}{\xi[x\mapsto a]}
 = \EE_{\U_1,\U_2}(\interpret{s'}{\III}{\xi[x\mapsto a]}) (a)
 $.
 Since $s$ and $s'$ do not contain $x$, $\interpret{s}{\III}{\xi[x\mapsto a]}$ and $\interpret{s'}{\III}{\xi[x\mapsto a]}$
 do not depend on $a$.
 Thus, $\EE_{\U_1,\U_2}(\interpretaxi{s}) = \EE_{\U_1,\U_2}(\interpretaxi{s'})$.
 Since $\III$ is extensional, $\EE_{\U_1,\U_2}$ is injective and
 hence
 $\interpretaxi{s} = \interpretaxi{s'}$. It follows that $s \eq s'$ is true in $\III$ under $\xi$,
 and so is the entire conclusion of the inference.
\end{proof}

A problem expressed in higher-order logic must be transformed into
clausal normal form before the calculi can be applied. This process works as
in the first-order case, except for skolemization.
The issue is that skolemization, when performed naively, is unsound for
higher-order logic with a Henkin semantics \cite[\Section~6]{miller-1987},
because it introduces new functions that can be used to instantiate variables.

The core of this \paper{} is not affected by this because the
problems are given in clausal form.
For the implementation, we claim soundness
only \wrt\ models that satisfy the axiom of choice,
which is the
semantics mandated by the TPTP THF format \cite{sutcliffe-et-al-2009}.
By contrast, refutational completeness holds \wrt\ arbitrary models
as defined above.
Alternatively, skolemization can be made sound by introducing mandatory arguments
as described by Miller \cite[\Section~6]{miller-1987} and in the conference version of
this \paper{} \cite{bentkamp-et-al-2018}.

This issue also affects axiom (\infname{Ext}) because it contains the Skolem symbol $\diff$.
As a consequence, (\infname{Ext}) does not hold in all extensional interpretations.
The extensional calculi are thus only sound
\wrt\ interpretations in which (\infname{Ext}) holds.
However, we can prove that (\infname{Ext}) is compatible with our logic:

\noindent%
\begin{minipage}{\textwidth}%
  % force theorem statement and proof to the same page.
  % Unfortunately \begin{samepage} \end{samepage} also pulls
  % parts of the above text to the next page, so we do it manually.
\begin{theoremx}
  Axiom {\upshape(}\infname{Ext}{\upshape)} is satisfiable.
\end{theoremx}
\begin{proof}
  For a given signature, let $(\UU,\IIty,\II,\EE)$ be an Herbrand interpretation. That is,
  we define $\UU$ to contain the set $\UU_\tau$ of all terms of type $\tau$ for each ground type $\tau$,
  we define $\IIty$ by $\IIty(\kappa)(\tuple{\tau})= \kappa(\tuple{\tau})$,
  we define $\II$ by $\II(\cst{f},\UU_{\tuple{\tau}})=\cst{f}\typeargs{\tuple{\tau}}$, and
  we define $\EE$ by $\EE_{U_\tau,U_\upsilon}(\cst{f})(\cst{a}) = \cst{f}\>\cst{a}$.
  Then $\EE_{U_\tau,U_\upsilon}$ is clearly injective and hence $\III$ is extensional.
  To show that $\III\models(\infname{Ext})$, we need to show that (\infname{Ext}) is true
  under all valuations. Let $\xi$ be a valuation.
  If $x \eq y$ is true under $\xi$, (\infname{Ext}) is also true. Otherwise $x \eq y$ is false under $\xi$,
  and hence $\xi(x) \not= \xi(y)$.
  Then we have
$\interpretaxi{x\;(\diff\typeargs{\alpha,\beta}\;x\;y )}
= (\xi(x))\;(\diff\typeargs{\alpha,\beta}\;(\xi(x))\;(\xi(y)) )
\not=
(\xi(y))\;(\diff\typeargs{\alpha,\beta}\;(\xi(x))\;(\xi(y)) )
= \interpretaxi{y\>(\diff\typeargs{\alpha,\beta}\;x\;y)}$.
  Therefore,
  $x\;(\diff\typeargs{\alpha,\beta}\;x\;y )
  \noteq y\>(\diff\typeargs{\alpha,\beta}\;x\;y)$ is true in $\III$ under $\xi$ and so is (\infname{Ext}).
\end{proof}
\end{minipage}

\subsection{The Redundancy Criterion}
\label{ssec:the-redundancy-criterion}

For our calculi, a redundant clause cannot simply be defined as
a clause whose ground instances are entailed by smaller (\negvvthinspace$\prec$) ground
instances of
existing
clauses, because this would make all \infname{ArgCong} inferences redundant.
Our solution is to base the redundancy criterion on a weaker ground logic---%
ground monomorphic first-order logic---in
which argument congruence does not hold. This logic also plays a central role
in our refutational completeness proof.

We employ an encoding $\flooronly$ to translate
ground \hbox{$\lambda$-free} higher-order terms into
ground first-order terms. It indexes each symbol occurrence
with its type arguments and its term argument count. Thus,
$\floor{\cst{f}}=\cst{f}_0$,
$\floor{\cst{f}\>\cst{a}}=\cst{f}_1(\cst{a}_0)$,
and
$\floor{\cst{g}\typeargs{\kappa}}=\cst{g}_0^\kappa$.
This is enough to disable argument congruence; for example,
$\{\cst{f} \eq \cst{h}{,}\; \cst{f}\>\cst{a} \noteq \cst{h}\> \cst{a}\}$ is
unsatisfiable, whereas
its encoding $\{\cst{f}_0 \eq \cst{h}_0{,}\allowbreak\; \cst{f}_1(\cst{a}_0) \noteq \cst{h}_1(\cst{a}_0)\}$
is satisfiable.
For clauses built from fully applied ground terms, the two logics are
isomorphic, as we would expect from a graceful generalization.

Given a $\lambda$-free higher-order signature $(\Sigmaty,\Sigma)$,
we define a first-order signature
$(\Sigmaty,\Sigma_\GF)$ as follows.
The type constructors $\Sigmaty$ are the same in both signatures, but $\fun$ is uninterpreted
in first-order logic.
For each symbol
$\cst{f}\oftypedecl
\forallty{\tuple{\alpha}_m}\tau_1\fun\cdots\fun\tau_n\fun\tau$ in $\Sigma$, where $\tau$ is not functional,
we introduce symbols $\cst{f}_{l}^{\tuple{\upsilon}_m}\in\Sigma_\GF$
with argument types
$\tuple{\tau}_l\sigma$
and return type $(\tau_{l+1}\fun\cdots\fun\tau_n\fun\tau)\sigma$,
where $\sigma = \{\tuple{\alpha}_m \mapsto \tuple{\upsilon}_m\}$,
for
each tuple of ground types $\tuple{\upsilon}_m$ and each $l \in \{0,\dots,n\}$.

  For example, let
  $\Sigma = \{\cst{a}\oftypedecl \kappa{,}\; \cst{g}\oftypedecl\kappa\fun\kappa\fun\kappa\}$.
  The corresponding first-order signature is
  $\Sigma_\GF=\{
  \cst{a}_0\oftypedecl{\kappa}{,}\;
  \cst{g}_0\oftypedecl{\kappa\fun\kappa\fun\kappa}{,}\;
  \cst{g}_1\oftypedecl{\kappa\fofun\kappa\fun\kappa}{,}\;
  \cst{g}_2\oftypedecl{\kappa^2\fofun\kappa}
  \}$
   where
   $\cst{f}\oftypedecl{\tuple{\tau}\fofun\upsilon}$ denotes a first-order function symbol $\cst{f}$
   with argument types $\tuple{\tau}$ and return type $\upsilon$,
   and $\fun$ is an uninterpreted binary type constructor.
  The term $\floor{\cst{g}\;\cst{a}\;\cst{a}}
  = \cst{g}_2(\cst{a}_0, \cst{a}_0)$ has type ${\kappa}$,
  and $\floor{\cst{g}\;\cst{a}} =  \cst{g}_1(\cst{a}_0)$
  has type ${\kappa\fun\kappa}$.

Thus, we consider three levels of logics:\
the $\lambda$-free higher-order level $\HH$
over a given signature $(\Sigmaty,\allowbreak\Sigma)$,
the ground $\lambda$-free higher-order level $\GH$,
corresponding to $\HH$'s ground fragment,
and the ground monomorphic first-order level $\GF$
over the signature $(\Sigmaty,\allowbreak\Sigma_\GF)$ defined above.
We use $\THH$, $\TGH$, and $\TGF$ to denote the respective sets of terms,
$\TyHH$, $\TyGH$, and $\TyGF$ to denote the respective sets of types,
and $\CHH$, $\CGH$, and $\CGF$ to denote the respective sets of clauses.
In the purifying calculi, we exceptionally let $\CHH$ denote the set of
purified clauses.
Each of the three levels has an entailment relation $\models$.
A clause set~$N_1$ entails a clause set~$N_2$, denoted $N_1 \models N_2$,
if any model of~$N_1$ is also a model of~$N_2$. On $\HH$ and $\GH$,
we use $\lambda$-free higher-order models for the intensional calculi
and extensional $\lambda$-free higher-order models for the extensional calculi;
on $\GF$, we use first-order models.
This machinery may seem excessive, but it is essential to define redundancy of
clauses and inferences properly, and it will play an important role in the
refutational completeness proof (Section~\ref{sec:refutational-completeness}).

The three levels are connected by two functions,
$\gnd$ and $\flooronly$:
\begin{definitionx}[Grounding function ${\gnd}$ on terms and clauses]
  The grounding function $\gnd$ maps
  terms $t \in \THH$ to the set of their ground instances---i.e.,
  the set of all $t\theta \in \TGH$ where $\theta$ is a substitution.
  It also maps clauses $C\in\CHH$ to the set of their ground instances---i.e.,
  the set of all $C\theta \in \CGH$ where $\theta$ is a substitution.
\end{definitionx}
\begin{definitionx}[Encoding ${\flooronly}$ on terms and clauses]
  \,The encoding $\flooronly : \TGH \rightarrow \TGF$ is recursively defined
  as $\floor{\cst{f}\typeargs{\tuple\upsilon_m}\;\tuple{u}_l}
  = \cst{f}_{l}^{\tuple{\upsilon}_m} ({\floor{\tuple{u}_{l}}})$.
  The encoding $\flooronly$ is
  extended to map from $\CGH$ to $\CGF$ 
  by mapping each literal and each side of a literal individually.
\end{definitionx}
The encoding $\flooronly$ is bijective with inverse $\ceilonly$.
Using $\ceilonly$, the clause order $\succ$ on $\TGH$
can be transferred to $\TGF$ by defining
$t \succ s$ as equivalent to $\ceil{t} \succ \ceil{s}$.
The property that $\succC$ on clauses is the multiset extension of $\succL$ on literals,
which in turn is the multiset extension of $\succ$ on terms, is maintained because
$\ceilonly$ maps the multiset representations elementwise.

Schematically, the three levels are connected as follows:
\[
\begin{tikzpicture}[level distance=13em]
  \node[align=center]{$\HH$\\higher-order}[grow'=right]
    child {
      node[align=center]{$\GH$\\ground higher-order}
        child {
          node[align=center]{$\GF$\\ground first-order}
          edge from parent[->]
            node[above] {$\flooronly$}
        }
      edge from parent[->]
        node[above] {$\gnd$}
    };
\end{tikzpicture}
\]

Crucially, green subterms in $\TGH$ correspond to
subterms in $\TGF$ (Lemma~\ref{lem:subterm-correspondence1}),
whereas nongreen subterms
in $\TGH$ are not subterms at all in $\TGF$.
For example, $\cst{a}$ is a green subterm of $\cst{f}\>\cst{a}$,
and correspondingly $\floor{\cst{a}} = \cst{a}_0$ 
is a subterm of $\floor{\cst{f}\>\cst{a}} = \cst{f}_1(\cst{a}_0)$.
On the other hand, $\cst{f}$ is not a green subterm of $\cst{f}\>\cst{a}$,
and correspondingly $\floor{\cst{f}} = \cst{f}_0$ 
is not a subterm of $\floor{\cst{f}\>\cst{a}} = \cst{f}_1(\cst{a}_0)$.

To state the correspondence between green subterms in $\TGH$
and subterms in $\TGF$ explicitly, we define positions of green subterms as follows.
A term $s\in\THH$ is a green subterm at position $\epsilon$ of $s$.
If a term $s\in\THH$ is a green subterm at position $p$ of $u_i$ for some $1 \leq i \leq n$,
then $s$ is a green subterm at position $i.p$ of
$\cst{f}\typeargs{\tuple\tau\vthinspace}\; \tuple{u}_n$
and of
$x\;{\bar u}_n$.
For $\TGF$, positions are defined as usual in first-order logic.

\begin{lemmax}\label{lem:subterm-correspondence1}
  Let $s,t\in\TGH$.
  We have $\floor{\fosubterm{t}{s}_p} = \subterm{\floor{t}}{\floor{s}}_p$.
  In other words, $s$ is a green subterm of $t$ at position $p$ if and only if $\floor{s}$ is
  a subterm of $\floor{t}$ at position $p$.
\end{lemmax}
\begin{proof}
\begin{sloppypar}
  By induction on $p$.
  If $p=\varepsilon$, then $s=\hosubterm{t}{s}_p$.
  Hence $\floor{\hosubterm{t}{s}_p} = \floor{s} = \fosubterm{\floor{t}}{\floor{s}}_p $.
  If $p=i.p'$, then $\hosubterm{t}{s}_p = \cst{f}{\typeargs{\tuple{\tau}}}\;\tuple{u}_n$
  with $u_i = \hosubterm{u_i}{s}_{p'}$.
  Applying $\flooronly$, we obtain by the induction hypothesis
  that $\floor{u_i} = \fosubterm{\floor{u_i}}{\floor{s}}_{p'}$. Therefore,
  $\floor{\hosubterm{t}{s}_p} = \cst{f}_{n}^{\tuple{\tau}}(\floor{u_{1}},\dots,\allowbreak\floor{u_{i-1}},
  \allowbreak\fosubterm{\floor{u_i}}{\floor{s}}_{p'},\floor{u_{i+1}},\dots,\floor{u_n})$.
  It follows that $\floor{\hosubterm{t}{s}_p} = \fosubterm{\floor{t}}{\floor{s}}_p$. \qedhere
\end{sloppypar}
\end{proof}%
\begin{corollaryx}\label{cor:subterm-correspondence2}
  Given $s,t\in\TGF$, we have $\ceil{\subterm{t}{s}_p} = \fosubterm{\ceil{t}}{\ceil{s}}_p$.
\end{corollaryx}%
  \begin{lemmax}\label{lem:order-prop-transfer}
    Well-foundedness,
    totality,
    compatibility with
    contexts, and
    the subterm property
    hold for $\succ$ on $\TGF$.
  \end{lemmax}
  \begin{proof}
    \textsc{Compatibility with contexts:}\enskip
    We must show that $s \succ s'$ implies $\hosubterm{t}{s}_p \succ
    \hosubterm{t}{s'}_p$ for terms $t,s,s'\in\TGF$. Assuming $s \succ
    s'$, we have $\ceil{s} \succ \ceil{s'}$. By compatibility with green
    contexts on $\TGH$, 
    it follows that 
    $\fosubterm{\ceil{t}}{\ceil{s}}_p
    \succ \fosubterm{\ceil{t}}{\ceil{s'}}_p$. By
    Corollary~\ref{cor:subterm-correspondence2}, we have $\hosubterm{t}{s}_p \succ
    \hosubterm{t}{s'}_p$.

    \medskip

    \noindent
    \textsc{Well-foundedness:}\enskip Assume that there exists an infinite
    descending chain $t_1\succ t_2\succ \cdots$ in $\TGF$. By applying
    $\ceilonly$, we then obtain the chain
     $\ceil{t_1}\succ \ceil{t_2}\succ \cdots$ in $\TGH$, contradicting
    well-foundedness on $\TGH$.

    \medskip

    \noindent
    \textsc{Totality:}\enskip Let $s, t\in\TGF$. Then $\ceil{t}$ and $\ceil{s}$ must be comparable by totality
    in $\TGH$. Hence, $t$ and $s$ are comparable.

    \medskip

    \noindent
    \textsc{Subterm property:}\enskip By Corollary~\ref{cor:subterm-correspondence2} and
    the subterm property on $\TGH$, we have $\ceil{\hosubterm{t}{s}_p} =
    \fosubterm{\ceil{t}}{\ceil{s}}_p \succ \ceil{s}$. Hence, $\hosubterm{t}{s}_p
    \succ s$.
  \end{proof}

In standard superposition, redundancy relies on the entailment relation
$\models$ on ground clauses. We will define redundancy on $\GH$ and $\HH$ in the
same way, but using $\GF$'s entailment relation.
This notion of redundancy gracefully generalizes the first-order notion,
without making all \infname{ArgCong} inferences redundant.

The standard redundancy criterion for standard superposition cannot
justify subsumption deletion. Following Waldmann et
al.~\cite{waldmann-et-al-2020-saturation}, we 
incorporate subsumption into the redundancy criterion.
A clause $C$ \emph{subsumes}~$D$ if there exists a substitution
$\sigma$ such that $C\sigma \subseteq D$. A clause $C$ \emph{strictly subsumes}~$D$ if
$C$ subsumes $D$ but $D$ does not subsume $C$.
Let~$\sqsupset$ stand for ``is strictly subsumed by''.
Using the applicative encoding, it is easy to show that~$\sqsupset$
is well founded because strict subsumption is well founded in first-order logic.

We define the sets of redundant clauses \wrt\ a given clause set as follows:
\begin{itemize}
\item Given $C\in\CGF$ and $N\subseteq\CGF$, let $C\in\GFRedC(N)$ if
$\{D \in N \mid D \prec C\}\models C$.
\item Given $C\in\CGH$ and $N\subseteq\CGH$, let $C\in\GHRedC(N)$ if
$\floor{C} \in \GFRedC(\floor{N})$.
\item
Given $C\in\CHH$ and $N\subseteq\CHH$, let $C\in\HRedC(N)$ if
for every $D \in \gnd(C)$,
we have $D \in \GHRedC(\gnd(N))$ or
there exists $C' \in N$ such that $C \sqsupset C'$ and $D \in \gnd(C')$.
\end{itemize}

Along with the three levels of logics, we consider three inference systems:\
$\HInf$, $\GHInf$ and $\GFInf$. $\HInf$ is one of the four variants of
the inference system described in
\Section~\ref{ssec:the-inference-rules}. For uniformity, we regard
axiom~(\infname{Ext}) as a premise-free inference rule~\infname{Ext} whose
conclusion is (\infname{Ext}).
In the purifying calculi, the conclusion of $\infname{Ext}$ must be purified.
$\GHInf$
consists of all \infname{Sup}, \infname{ERes}, and \infname{EFact} inferences
from $\HInf$ whose premises and conclusion are ground,
a premise-free rule \infname{GExt}
whose infinitely many conclusions are the ground instances of (\infname{Ext}),
and the following ground variant of \infname{ArgCong}:
\[\namedinference{GArgCong}
{C' \llor s \eq s'}
{C' \llor s\>\tuple{u}_n \eq s'\>\tuple{u}_n}\]
where $s \eq s'$ is strictly eligible in the premise %
and $\tuple{u}_n$ is a nonempty tuple of ground terms.
$\GFInf$ contains all \infname{Sup}, \infname{ERes}, and \infname{EFact} inferences
from $\GHInf$ translated by $\flooronly$. It coincides exactly with standard first-order
superposition. Given a \infname{Sup}, \infname{ERes}, or \infname{EFact}
inference $\iota \in \GHInf$, let $\floor{\iota}$ denote the corresponding
inference in $\GFInf$.

Each of the three inference systems is parameterized by a selection function.
For $\HInf$, we globally fix a selection function $\HSel$.
For $\GHInf$ and $\GFInf$, we need to consider different selection functions $\GHSel$ and $\GFSel$.
We write $\GHInf^\GHSel$ for $\GHInf$ and $\GFInf^\GFSel$ for $\GFInf$
to make the dependency on the respective selection functions
explicit.
Let $\gnd(\HSel)$ denote the set of all selection functions on~$\CGH$
such that for each clause
in $C\in\CGH$, there exists a clause $D\in\CHH$ with $C\in\gnd(D)$ and corresponding selected literals.
For each selection function $\GHSel$ on $\CGH$, via the bijection $\flooronly$, we obtain a corresponding selection function
on $\CGF$, which we denote by $\floor{\GHSel}$.

\begin{notationx}
  Given an inference $\iota$, we write $\prem(\iota)$ for the tuple of premises,
  $\mprem(\iota)$ for the main (i.e., rightmost) premise, $\preconcl(\iota)$
  for the conclusion before purification, and $\concl(\iota)$ for
  the conclusion after purification. For the nonpurifying calculi,
  $\preconcl(\iota) = \concl(\iota)$ simply denotes the conclusion.
\end{notationx}

\begin{definitionx} [Encoding $\flooronly$ on inferences]\,
Given a \infname{Sup}, \infname{ERes}, or \infname{EFact}
inference $\iota \in \GHInf$, let $\floor{\iota}\in\GFInf$ denote 
the inference defined by $\prem(\floor{\iota}) = \floor{\prem(\iota)}$
and $\concl(\floor{\iota}) = \floor{\concl(\iota)}$.
\end{definitionx}

\begin{definitionx} [Grounding function $\gnd$ on inferences]\,
Given a selection function $\GHSel\in\gnd(\HSel)$,
and a non-\infname{PosExt} inference $\iota\in\HInf$,
we define the set $\gnd^\GHSel(\iota)$ of ground instances of $\iota$
to be all inferences $\iota'\in\GHInf^\GHSel$ such that $\prem(\iota') = \prem(\iota)\theta$
and $\concl(\iota') = \preconcl(\iota)\theta$ for some grounding substitution $\theta$.
For \infname{PosExt} inferences $\iota$, which cannot be grounded, we let
$\gnd^\GHSel(\iota) = \Undef$.
\end{definitionx}
This will map \infname{Sup} to \infname{Sup},
\infname{EFact} to \infname{EFact},
\infname{ERes} to \infname{ERes},
\infname{Ext} to \infname{GExt},
and \infname{ArgCong} to \infname{GArgCong} inferences, but it
is also possible that $\gnd^\GHSel(\iota)$ is the empty set for some inferences $\iota$.

We define the sets of redundant inferences \wrt\ a given clause
set as follows:
\begin{itemize}
\item Given $\iota\in\GFInf^\GFSel$ and $N\subseteq\CGF$, let $\iota\in\GFRedI^\GFSel(N)$
  if 
  $\prem(\iota) \ccap \GFRedC(N) \not= \varnothing$ or
  $\{D \in N \mid D \prec \mprem(\iota)\} \models \concl(\iota)$.
\item Given $\iota\in\GHInf^\GHSel$ and $N\subseteq\CGH$, let $\iota\in\GHRedI^\GHSel(N)$ if
\begin{itemize}
\item $\iota$ is not a \infname{GArgCong} or \infname{GExt} inference
  and $\floor{\iota}\in\smash{\GFRedI^{\floor{\GHSel}}}(\floor{N})$; or
\item the calculus is nonpurifying
  and $\iota$ is a \infname{GArgCong} or \infname{GExt} inference
  and $\concl(\iota)\in N\ccup\GHRedC(N)$; or
\item the calculus is purifying
  and $\iota$ is a \infname{GArgCong} or \infname{GExt} inference
  and $\floor{N} \models \floor{\concl(\iota)}$.
\end{itemize}
\item Given $\iota\in\HInf$ and $N\subseteq\CHH$, let $\iota\in\HRedI(N)$ if
\begin{itemize}
  \item $\iota$ is not a \infname{PosExt} inference
    and $\gnd^\GHSel(\iota)\subseteq\GHRedI(\gnd(N))$ for all $\GHSel\in\gnd(\HSel)$; or
  \item $\iota$ is a \infname{PosExt} inference
    and $\gnd(\concl(\iota))\subseteq \gnd(N)\ccup\GHRedC(\gnd(N))$.
  \end{itemize}

\end{itemize}
Occasionally, we omit the selection function in the notation when it is irrelevant.
A clause set $N$ is \emph{saturated} \wrt\ an inference system and a
redundancy criterion $(\RedI,\RedC)$ if every inference from clauses
in $N$ is in~$\RedI(N).$

\subsection{Simplification Rules}
\label{lf:ssec:simplification-rules}
The redundancy criterion $(\HRedI, \HRedC)$ is strong enough to support
most of the simplification rules implemented in Schulz's first-order prover E
\cite[Sections 2.3.1~and~2.3.2]{schulz-2002-brainiac},
some only with minor adaptions.
Deletion of duplicated literals, 
deletion of resolved literals, 
syntactic tautology deletion,  
negative simplify-reflect, and
clause subsumption
adhere to our redundancy criterion.

Positive simplify-reflect and equality subsumption are supported
by our criterion if they are applied in green contexts.
Semantic tautology deletion can be applied as well, but we must use the
entailment relation of the GF level---i.e., only rewriting in green contexts can
be used to establish the entailment.
Similarly, rewriting of positive and negative literals (demodulation) 
can only be applied in green contexts.
Moreover, for positive literals, the rewriting clause must be smaller than the rewritten
clause---a condition that is also necessary with the standard first-order redundancy criterion
but not always fulfilled by Schulz's rule.
As for destructive equality resolution, even in first-order logic the rule cannot be justified with the standard redundancy criterion,
and it is unclear whether it preserves refutational completeness. %

\section{Refutational Completeness}
\label{sec:refutational-completeness}

Besides soundness, the most important property of the four calculi introduced in
\Section~\ref{ssec:the-inference-rules}
is refutational completeness.
We will prove the static and dynamic refutational completeness of $\HInf$ \wrt\ $(\HRedI, \HRedC)$, which is defined as follows:
\begin{definitionx}[Static refutational completeness]
  Let $\Inf$ be an inference system and let $(\RedI, \RedC)$ be a redundancy criterion.
  The inference system $\Inf$ is called \emph{statically refutationally complete} \wrt\ $(\RedI, \RedC)$ if
  we have $N \models \bot$ if and only if $\bot \in N$
  for every clause set $N$ that is saturated \wrt\ $\Inf$ and $\RedI$.
\end{definitionx}
\begin{definitionx}[Dynamic refutational completeness] \label{def:dyn-complete}
Let $\Inf$ be an inference system and let $(\RedI, \RedC)$ be a redundancy criterion.
Let $(N_i)_i$ be a finite or infinite sequence over sets of clauses. 
Such a sequence is
called a \emph{derivation} if
$N_i \setminus N_{i+1} \subseteq \RedC(N_{i+1})$ for all $i$.
It is called \emph{fair} if all $\Inf$-inferences from clauses in
the limit inferior $\bigcup_i \bigcap_{\!j \geq i} N_{\!j}$ are contained in
$\bigcup_i \RedI(N_i)$.
The inference system $\Inf$ is called \emph{dynamically refutationally complete} \wrt\ $(\RedC, \RedI)$ if
for every fair derivation $(N_i)_i$ such that $N_0 \models \bot$,
we have $\bot \in N_i$ for some $i$.
\end{definitionx}
To circumvent the term order's potential nonmonotonicity, our \infname{Sup}
inference rule only considers green subterms. This is reflected in our proof by the
reliance on the $\GF$ level introduced in \Section~\ref{ssec:the-redundancy-criterion}.
The equation $\cst{g}_0 \eq \cst{f}_0 \in \CGF$,
which corresponds to the equation $\cst{g} \eq \cst{f}\in\CGH$,
cannot be used directly to rewrite
the clause $\cst{g}_1(\cst{a}_0) \noteq \cst{f}_1(\cst{a}_0) \in \CGF$,
which corresponds to $\cst{g}\;\cst{a} \noteq \cst{f}\;\cst{a}\in\CGH$.
Instead, we
first need to apply \infname{ArgCong} to derive
$\cst{g}\>x \eq \cst{f}\>x$,
which after grounding and transfer to $\GF$ yields
$\cst{g}_1(\cst{a}_0) \eq \cst{f}_1(\cst{a}_0)$.
The $\GF$ level is a device that enables us to
reuse the refutational completeness result for standard (first-order) superposition.

The proof proceeds in three steps, corresponding to the three levels
$\GF$, $\GH$, and $\HH$ introduced in \Section~\ref{ssec:the-redundancy-criterion}:
\begin{enumerate}
\item We use Bachmair and Ganzinger's work on the refutational completeness of standard
superposition~\cite{bachmair-ganzinger-1994}
to prove the static refutational completeness of $\GFInf$.
\item From the first-order model constructed in Bachmair and Ganzinger's proof,
we derive a $\lambda$-free higher-order model to prove the static refutational completeness of $\GHInf$.
\item We use the saturation framework of Waldmann et al.~\cite{waldmann-et-al-2020-saturation} to lift the static refutational completeness of $\GHInf$
to static and dynamic refutational completeness of $\HInf$.
\end{enumerate}

In the first step,
since the inference system $\GFInf$ is standard ground superposition,
we only need to work around minor differences between Bachmair and Ganzinger's
definitions and ours. %
Given a saturated clause set $N\subseteq\CGF$ with $\bot\not\in N$,
Bachmair and Ganzinger prove refutational completeness by constructing
a term rewriting system $R_N$ and showing that it can be
viewed as an interpretation that is a model of $N$.
This step is exclusively concerned with ground first-order clauses.

In the second step, we derive refutational completeness of $\GHInf$.
Given a saturated clause set $N\subseteq\CGH$ with $\bot\not\in N$,
we use the first-order model $R_{\floor{N}}$ of $\floor{N}$ constructed in step (1)
to derive a clausal higher-order interpretation that is a model of $N$.
Thanks to saturation \wrt\ \infname{GArgCong},
the higher-order interpretation can conflate the interpretations of the members
$\cst{f}_0^{\tuple{\upsilon}},\dots,\cst{f}_n^{\tuple{\upsilon}}$ of a same symbol family. %
In the extensional calculi, saturation \wrt\ \infname{GExt} can be used to show that the constructed interpretation is extensional.

In the third step,
we employ the saturation framework by Waldmann et al.~\cite{waldmann-et-al-2020-saturation},
which is largely based on Bachmair and Ganzinger's~\cite{bachmair-ganzinger-2001-resolution},
to prove $\HInf$ refutationally complete.
Like Bachmair and Ganzinger's, 
the saturation framework allows us to prove the static and dynamic refutational completeness of our calculus on the nonground level.
On top of that, it allows us to use the redundancy criterion defined in Section~\ref{ssec:the-redundancy-criterion},
which supports deletion of subsumed formulas.
Moreover, the saturation framework provides completeness theorems for prover architectures, such as the DISCOUNT loop.
The main proof obligation the saturation framework leaves to us is that there
exist %
inferences in $\HInf$ corresponding to all nonredundant inferences in $\GHInf$.
For monotone term orders, we can avoid \infname{Sup} inferences into
variables~$x$ by exploiting the clause order's compatibility with contexts:
If $t' \prec t$, we have $C\{x \mapsto\nobreak t'\} \prec C\{x \mapsto t\}$,
which allows us to show that \infname{Sup} inferences into
variables are redundant. This technique fails for variables~$x$
that occur applied in~$C$, because the order lacks compatibility
with arguments. This is why the calculi must either purify clauses to make this
line of reasoning work again or perform some \infname{Sup} inferences into
variables.

\subsection{The Ground First-Order Level}
\label{ssec:the-ground-first-order-level}

We use Bachmair and Ganzinger's results on standard superposition~\cite{bachmair-ganzinger-1994}
to prove $\GF$ refutationally complete. In the subsequent steps, we will also make use of specific properties
of Bachmair and Ganzinger's model.

Bachmair and Ganzinger's logic and inference system differ in some details from~$\GF$:
\begin{itemize}

\item
Bachmair and Ganzinger use untyped first-order logic, whereas $\GF$'s logic is
typed. Bachmair and Ganzinger's proof works verbatim for monomorphic first-order
logic as well, but we need to require that the order $\succ$ has the subterm
property to show that there exist no critical pairs in the term rewriting
system, as observed by Wand \cite[Section~3.2.1]{wand-2017}.

\item
In their redundancy criterion for clauses, Bachmair and Ganzinger require that a finite subset of $\{D \in N \mid D \prec C\}$
entails $C$, whereas we require that $\{D \in N \mid D \prec C\}$ entails~$C$.
By compactness of first-order logic, the two criteria are equivalent.

\end{itemize}
Bachmair and Ganzinger prove refutational completeness for nonground clause sets, but
we only require the ground result here.

The basis of Bachmair and Ganzinger's proof is that
a term rewriting system $R$ defines an interpretation $\TGF/R$
such that for every ground equation $s
\eq t$, we have $\TGF/R \models s \eq t$ if and only if $s
\leftrightrewrite_R^* t$.
Formally, $\TGF/R$ denotes the
monomorphic first-order interpretation
whose universes $\U_\tau$ consist of the $R$-equivalence
classes over $\TGF$ containing terms of type $\tau$.
The interpretation $\TGF/R$ is term-generated---that is,
for every element $a$ of the universe of this interpretation and for any
valuation $\xi$, there exists
a ground term~$t$ such that $\interpret{t}{\TGF/R}{\xi} = a$.
To lighten notation, we will write $R$ to refer to both the term rewriting
system $R$ and the interpretation $\TGF/R$.

The term rewriting system is constructed as follows. Let $N\subseteq\CGF$.
We first define sets of rewrite rules $E_N^C$ and $R_N^C$ for all $C\in N$ by induction on the clause order.
Assume that $E_N^D$ has already been defined for all $D \in N$
such that $D \prec C.$ Then $R_N^C = \bigcup_{D \prec C} E_N^D.$
Let $E_N^C=\{s \rewrite t\}$ if the following conditions are met:\
\begin{enumerate}[(a)]
	\item $C = C' \lor s \eq t$; \label{cond:C-eq-C'-st}
	\item $s \eq t$ is strictly maximal in $C$; \label{cond:st-strictly-max}
	\item $s \succ t$; \label{cond:s-gt-t}
	\item $C'$ is false in $R_N^C$; \label{cond:C'-false}
	\item $s$ is irreducible \wrt\ $R_N^C.$ \label{cond:s-irred}
\end{enumerate}
Then $C$ is said to \emph{produce} $s \rewrite t$%
.
Otherwise, $E_N^C = \emptyset$.
Finally, $R_N = \bigcup_{D} E_N^D.$

We call an inference $\iota\in\GFInf$ \emph{B\&G-redundant} if
some $C\in\prem(\iota)$ is true in $R_N^C$ or $\concl(\iota)$ is true in $\smash{R_N^{\mprem(\iota)}}$.
We call a set $N\subseteq\CGF$ \emph{B\&G-saturated} if
all inferences from $N$ are B\&G-redundant.

\begin{lemmax} \label{lem:redundancy-implies-bg-redundancy}
  If $\bot\not\in N$ and $N\subseteq\CGF$ is saturated \wrt\ $\GFInf$ and $\GFRedI$,
  then $N$ is B\&G-saturated.
  \end{lemmax}
  \begin{proof}
  Let $N\subseteq\CGF$ be saturated \wrt\ $\GFInf$ and $\GFRedI$.
  To show that $N$ is B\&G-saturated,
  let $\iota$ be an inference from $N$.
  We need to show that $\iota$ is B\&G-redundant \wrt\ $N$.
  We proceed by well-founded induction on $\mprem(\iota)$ \wrt\ $\succ$.
  By the induction hypothesis, for all inferences~$\iota'$ with $\concl(\iota') \prec \mprem(\iota)$,
  $\iota'$ is B\&G-redundant \wrt\ $N$.
  By Lemma~5.5 of Bachmair and Ganzinger, $\iota$ is B\&G-redundant \wrt\ $N$.
\end{proof}

\begin{lemmax} \label{lem:productive-clauses}
  Let $\bot\not\in N$ and $N\subseteq\CGF$ be saturated \wrt\ $\GFInf$
  and $\GFRedI$. If $C = C' \lor s \eq t \in N$ produces $s \rewrite t$,
  then $s \eq t$ is strictly eligible in $C$ and $C'$ is false in $R_N$.
\end{lemmax}
\begin{proof}
By Lemma~\ref{lem:redundancy-implies-bg-redundancy},
$N$ is also B\&G-saturated.
By condition \ref{cond:C'-false}, $C'$ is false in $R_N^C$. Since $s\succ t$ by condition \ref{cond:s-gt-t} and
$s$ is irreducible \wrt\ $R_N^C$ by condition \ref{cond:s-irred},
$s \eq t$ is also false in $R_N^C$. Hence, $C$ is false in $R_N^C$.
Using this and conditions \ref{cond:C-eq-C'-st}, \ref{cond:st-strictly-max}, \ref{cond:s-gt-t},
and \ref{cond:s-irred}, we can apply
Lemma~4.11 of Bachmair and Ganzinger
using $N$ for $N$ and for $N'$, $C$ for $C$ and for $D$, $s$ for $s$, and $t$ for $t$.
Part~(ii) of that lemma shows that $s \eq t$ is strictly eligible in $C$, and
part~(iv) shows that $C'$ is false in $R_N$.
\end{proof}

\begin{theoremx}[Ground first-order static refutational completeness]
  \label{thm:GF-refutational-completeness}
  The inference system $\GFInf$ is statically refutationally complete \wrt\ $(\GFRedI, \GFRedC)$.
  More precisely, if $N\subseteq\CGF$ is a clause set saturated \wrt\ $\GFInf$
  and $\GFRedI$ such that $\bot\not\in N$,
  then $R_N$ is a model of $N$.
\end{theoremx}
\begin{proof}
  By Lemma~\ref{lem:redundancy-implies-bg-redundancy},
  $N$ is also B\&G-saturated.
  It follows that $R_N$ is a model of $N$, as shown in the
  proof of Theorem 4.14 of Bachmair and Ganzinger.
\end{proof}

\subsection{The Ground Higher-Order Level}

In this subsection, let $\GHSel$ be a selection function on $\CGH$, and
let $N\subseteq\CGH$ with $\bot\not\in N$ be a clause set saturated \wrt\
$\GHInf^\GHSel$ and $\GHRedI^\GHSel$. Clearly, $\floor{N}$ is then saturated \wrt\
$\smash{\GFInf^{\floor{\GHSel}}}$ and $\smash{\GFRedI^{\floor{\GHSel}}}$.

We abbreviate $R_{\floor{N}}$ as $\RfN$.
From $\RfN$, we construct a model
$\IIIho$ of $N$. The key properties enabling us to perform this construction
are that $\RfN$ is
term-generated and that the interpretations of the members
$\cst{f}_0^{\tuple{\upsilon}},\dots,\cst{f}_n^{\tuple{\upsilon}}$ of a same symbol family behave in the same way
thanks to the $\infname{ArgCong}$ rule.

\begin{lemmax}[Argument congruence]\label{lemma:arg-cong}
For terms
$s,t,u\in\TGH$,
if $\interpretfog{\floor{s}}\allowbreak = \interpretfog{\floor{t}}$,
then $\interpretfog{\floor{s\>u}} = \interpretfog{\floor{t\>u}}$.
\end{lemmax}%
\begin{proof}
  What we want to show is equivalent to
  \[\floor{s}   \leftrightrewrite_{\RfN}^* \floor{t}
  \text{ \,implies\, }
  \floor{s\>u} \leftrightrewrite_{\RfN}^* \floor{t\>u} \]
  By induction on the number of rewrite steps and due to symmetry, it suffices to show that
  \[\floor{s}  \rewrite_{\RfN} \floor{t}
  \text{ \,implies\, }
  \floor{s\>u} \leftrightrewrite_{\RfN}^* \floor{t\>u} \]
  If the rewrite step from $\floor{s}$ is below the top level, this is obvious because there is a corresponding rewrite step from $\floor{s\>u}$.
  If it is at the top level, $\floor{s} \rewrite \floor{t}$
  must be rule of $\RfN$.
  This rule must originate from a productive clause of the form
  $\floor{C} = \floor{C' \llor s \eq t}$.
  By Lemma~\ref{lem:productive-clauses}, $\floor{s \eq t}$ is
  strictly eligible in $\floor{C}$ \wrt\ $\floor{\GHSel}$,
  and hence $s \eq t$ is
  strictly eligible in $C$ \wrt\ $\GHSel$.
  Moreover, $s$ and $t$ have functional type.
  Thus, the following \infname{GArgCong} inference $\iota$ is applicable:
	\[
	\namedinference{GArgCong}
	{C' \llor s \eq t}
	{C' \llor s\>u \eq t\>u}
  \]
  By saturation,
  $\iota$ is redundant \wrt\ $N$---i.e., we have $\concl(\iota)\in N \ccup \GHRedC(N)$ (for the nonpurifying calculi)
  or $\floor{N}\models\concl(\iota)$ (for the purifying calculi). In both cases,
  by Theorem \ref{thm:GF-refutational-completeness}, $\floor{\concl(\iota)}$ is then true in $\RfN$.
  By Lemma~\ref{lem:productive-clauses}, $\floor{C'}$ is false in $\RfN$. Therefore,
  $\floor{s\>u \eq t\>u}$ must be true in $\RfN$. \qedhere
\end{proof}

\begin{lemmax}\label{lemma:arg-fun-cong}
  For terms
  $s,t,u,v\in\TGH$,
  if $\interpretfog{\floor{s}}\allowbreak = \interpretfog{\floor{t}}$ and $\interpretfog{\floor{u}}\allowbreak = \interpretfog{\floor{v}}$,
  then $\interpretfog{\floor{s\>u}} = \interpretfog{\floor{t\>v}}$.
\end{lemmax}%
\begin{proof}
By Lemma~\ref{lemma:arg-cong}, we have $\interpretfog{\floor{s\>u}} = \interpretfog{\floor{t\>u}}$.
It remains to show that $\interpretfog{\floor{t\>u}} = \interpretfog{\floor{t\>v}}$.
Since $t$ is ground, it must be of the form $\cst{f}\typeargs{\tuple\upsilon_m}\>\tuple{t}_n$.
The interpretation $\RfN$ defined above is an interpretation $(\U, \ifo)$ in monomorphic
first-order logic.
Then
\[\interpretfog{\floor{t\>u}} = \ifo(\cst{f}_{n+1}^{\tuple\upsilon_m})(\interpretfog{\floor{\tuple{t}_n}},\interpretfog{\floor{u}})
= \ifo(\cst{f}_{n+1}^{\tuple\upsilon_m})(\interpretfog{\floor{\tuple{t}_n}},\interpretfog{\floor{v}})
= \interpretfog{\floor{t\>v}} \tag*{\qedhere}\]
\end{proof}

\begin{definitionx} %
Define a higher-order interpretation $\IIIho = (\uho,\IIty^\GH,\iho,\eho)$ as follows.
The interpretation $\RfN$ defined above is an
interpretation $(\U, \ifo)$ in monomorphic first-order logic,
where $\U_\tau$ is its universe for type
$\tau$, and $\II$ is its interpretation function.
Let
$\uho \defeq \{ \U_{\tau} \mid \tau\text{ is a ground type}\}$.
Let $\IIty^\GH(\kappa)(\U_{\tuple\tau})= \U_{\kappa(\tuple\tau)}$
for all type constructors $\kappa$ and type tuples $\tuple\tau$.
Let $\iho(\cst{f},\U_{\tuple\tau}) \defeq \ifo(\cst{f}_0^{\tuple\tau})$.

Since $\RfN$ is term-generated, for every $a\in\U_{{\tau\fun\upsilon}}$ and $b\in\U_{\tau}$, there exist
ground terms $s : \tau\fun\upsilon$ and $u : \tau$ such that $\interpretfog{\floor{s}} = a$ and $\interpretfog{\floor{u}} = b$.
We define $\eho$ by
$\smash{
\eho_{\U_{\tau},\U_{\upsilon}}(a)(b) = \interpretfog{\floor{s\>u}}
}\negvthinspace$ for any term $u$.
By Lemma~\ref{lemma:arg-fun-cong}, this definition is independent of the choice of $s$ and $u$.
\end{definitionx}

\begin{lemmax}[Model transfer to $\GH$]\label{lem:model-transfer}
   $\IIIho$ is a model of $N$. In the extensional calculi, $\IIIho$ is an
   extensional model of $N$.
\end{lemmax}
\begin{proof}
  We first prove by induction on terms $t\in\TGH$ that
  $\interprethog{t} = \interpretfog{\floor{t}}$.
  Let $t\in\TGH$, and assume as the induction hypothesis that
  $\interprethog{u} = \interpretfog{\floor{u}}$
  for all subterms $u$ of~$t$. If $t$ is of the form
  $\cst{f}\typeargs{\tuple\upsilon}$, then
  \begin{align*}
  \interprethog{t} &= \iho(\cst{f}, \U_{\tuple\upsilon})%
  =\ifo(\cst{f}_0^{\tuple\upsilon})%
  =\interpretfog{\cst{f}_0^{\tuple\upsilon}}%
  =\interpretfog{\floor{\cst{f}\typeargs{\tuple\upsilon}}}=\interpretfog{\floor{t}}
  \end{align*}
  If $t$ is an application $t = t_1\;t_2$, where $t_1$ is of type $\tau\fun\upsilon$,
  then, using the definition of the term denotation and of $\eho$, we have
  \begin{align*}
  \interprethog{t_1\;t_2}
  = \eho_{\U_\tau,\U_\upsilon}(\interprethog{t_1}) (\interprethog{t_2}) %
  \overset{\!\text{IH}\!}{=}
  \eho_{\U_\tau,\U_\upsilon}(\interpretfog{\floor{t_1}}) (\interpretfog{\floor{t_2}})%
  =
  \interpretfog{\floor{t_1\;t_2}}
  \end{align*}

  So we have shown that $\interprethog{t} = \interpretfog{\floor{t}}$ for all
  terms $t$. It follows that a ground equation $s \eq t$ is true in
  $\IIIho$ if and only if $\floor{s \eq t}$ is true in $\RfN$.
  Hence a ground clause $C$ is true in $\IIIho$ if and only if $\floor{C}$
  is true in $\RfN$.
  By Theorem~\ref{thm:GF-refutational-completeness}, $\RfN$ is a model of $\floor{N}$.
  Thus, $\IIIho$ is a model of $N$.

  For the extensional calculi, it remains to show that $\IIIho$ is extensional---i.e.,
  we have to show that for all $\tau$ and $\upsilon$ and all $a,b\in \U_{\tau\fun\upsilon}$,
  if $a \not= b$, then $\eho(a) \not= \eho(b)$.
  Since $\RfN$ is term-generated, there are terms $s,t\in\TGF$
  such that $\interpretfo{s}{} = a$ and $\interpretfo{t}{} = b$.
  By what we have shown above, it follows that
  $\interpretho{s'}{} = a$ and $\interpretho{t'}{} = b$ for $s'=\ceil{s}$ and $t'=\ceil{t}$.

  Since $N$ is saturated, the \infname{GExt} inference that generates
  the clause \[C ~=~ s'\;(\diff\typeargs{\tau,\upsilon}\>s'\>t')
  \noteq t'\>(\diff\typeargs{\tau,\upsilon}\>s'\>t') \llor s' \eq t'\] is
  redundant---i.e., $C \in N \ccup \GHRedC(N)$ (in the nonpurifying calculi) or $\floor{N}\models \floor{C}$ (in the purifying calculi).
  In both cases, it follows that
  $\RfN\models\floor{C}$ by Theorem \ref{thm:GF-refutational-completeness} and
  thus $\IIIho\models C$ by what we have shown above.
  The second literal of $C$ is false in $\IIIho$ because $\interpretho{s'}{} = a \not= b = \interpretho{t'}{}$.
  So the first literal of $C$ must be true in $\IIIho$ and thus
    \begin{align*}
      \eho(a)(\interpretho{\diff\typeargs{\tau,\upsilon}\>s'\>t' }{})
    &= \interpretho{s'\;(\diff\typeargs{\tau,\upsilon}\>s'\>t' )}{} \\
    &\not= \interpretho{t'\>(\diff\typeargs{\tau,\upsilon}\>s'\>t' )}{}
    = \eho(b)(\interpretho{\diff\typeargs{\tau,\upsilon}\>s'\>t' }{})
    \end{align*}
  It follows that $\eho(a) \not= \eho(b)$.
  \end{proof}

  We summarize the results of this subsection in the following theorem:

  \begin{theoremx}[Ground static refutational completeness]
    \label{thm:GH-refutational-completeness}
    Let $\GHSel$ be a selection function on $\CGH$.
    Then the inference system $\GHInf^\GHSel$ is statically refutationally complete
    \wrt\ $(\GHRedI, \GHRedC)$.
    That means, if $N \subseteq \CGH$ is saturated \wrt\ $\GHInf^\GHSel$ and $\GHRedI^\GHSel$,
    then $N \models \bot$ if and only if $\bot \in N$.
  \end{theoremx}

The construction of $\IIIho$ relies on the specific properties of $\RfN$. It
would not work with an arbitrary first-order interpretation. Transforming a
$\lambda$-free higher-order interpretation into a first-order interpretation is
easier:

  \begin{lemmax} \label{lem:gf-interpretation-from-gh}
    Given an interpretation $\III$
    on $\GH$,
    there exists an interpretation $\III^\GF$ on $\GF$
    such that for any clause $C\in\CGH$ the truth values of
    $C$ in $\III$ and of $\floor{C}$ in $\III^\GF$ coincide.
    \end{lemmax}
    \begin{proof}
    Let $\III = (\UU,\IIty,\II,\EE)$ be a $\lambda$-free higher-order interpretation.
    Let $\UU^\GF_\tau = \interpret{\tau}{\IIIty}{}$ be the first-order type universe for the ground type $\tau$.
    For a symbol $\smash{\cst{f}_{l}^{\tuple{\upsilon}_m}} \in
    \Sigma_\GF$ and universe elements $\tuple{a}_l$, let $\II^\GF (\smash{\cst{f}_{l}^{\tuple{\upsilon}_m}})(\tuple{a}_l) =
    \interpret{\cst{f}\typeargs{\tuple\upsilon_m}\>\tuple{x}_l}{\III}{\{\tuple{x}_l \mapsto \tuple{a}\}}$.
    This defines an interpretation $\III^\GF = (\UU^\GF,\II^\GF)$ on $\GF$.
  
    We need to show that for any $C\in\CGH$, $\III \models C$ if and only if $\III^\GF \models \floor{C}$.
    It suffices to show that $\interpret{t}{\III}{} = \interpret{\floor{t}}{\III^\GF}{}$
    for all terms $t\in\TGH$.
    We prove this by induction on $t$.
    Since $t$ is ground, it must be of the form
    $\cst{f}\typeargs{\tuple\upsilon_m}\>\tuple{s}_l$. Then
    $\floor{t} = \smash{\cst{f}_{l}^{\tuple{\upsilon}_m}(\floor{\tuple{s}_l})}$
    and hence
    \[\interpret{\floor{t}}{\III^\GF}{}
    = \II^\GF (\smash{\cst{f}_{l}^{\tuple{\upsilon}_m}})(\interpret{\floor{\tuple{s}_l}}{\III^\GF}{})
    \overset{\smash{\text{IH}}}{=}
    \II^\GF (\smash{\cst{f}_{l}^{\tuple{\upsilon}_m}})(\interpret{\tuple{s}_l}{\III}{})
    = \interpret{\cst{f}\typeargs{\tuple\upsilon_m}\>\tuple{s}_l}{\III}{}
    = \interpret{t}{\III}{}\]
    using the definition of $\II^\GF$ and Lemma~\ref{lem:subst-lemma-general} for the third step.
    \end{proof}

\subsection{The Nonground Higher-Order Level}

To lift the result to the nonground level, we employ the saturation framework of
Waldmann et al.~\cite{waldmann-et-al-2020-saturation}.
It is easy to see that the entailment relation $\models$ on $\GH$ is a consequence relation in the sense of the framework.
It remains to show that our redundancy criterion on $\GH$ is a redundancy criterion in the sense of the framework and that
$\gnd$ is a grounding function in the sense of the framework:

  \begin{lemmax} \label{lem:redundancy-criterion}
    The redundancy criterion for $\GH$ is a redundancy criterion in the sense of the framework.
  \end{lemmax}
  \begin{proof} 
  We must prove the conditions (R1) to (R4)
  defined by Waldmann et al., which, 
  adapted to our context, state the following
  for all clause sets $N,N' \subseteq \CGH$:
  \begin{enumerate}[leftmargin=3em]
    \item[(R1)] if $N \models \bot$, then $N \setminus \GHRedC(N) \models \bot$;
    \item[(R2)] if $N \subseteq N'$, then
        $\GHRedC(N) \subseteq \GHRedC(N')$ 
        and $\GHRedI(N) \subseteq \GHRedI(N')$;
    \item[(R3)] if $N' \subseteq \GHRedC(N)$, then
        $\GHRedC(N) \subseteq \GHRedC(N \setminus N')$
        and $\GHRedI(N) \subseteq \GHRedI(N \setminus N')$;
    \item[(R4)] if $\iota \in \GHInf$ and $\concl(\iota) \in N$, then $\iota \in \GHRedI(N)$.
  \end{enumerate}

  \medskip\noindent
  (R1)\enskip
  It suffices to show that $N\setminus\GHRedC(N) \models N$ for $N\subseteq\CGH$.
  Let $\III$ be a model of $N\setminus\GHRedC(N)$.
  In the extensional calculi, let $\III$ be extensional.
  Then by Lemma~\ref{lem:gf-interpretation-from-gh},
  there exists a model $\III^\GF$ of $\floor{N\setminus\GHRedC(N)} = \floor{N}\setminus\GFRedC(\floor{N})$.
  By Lemma~5.2 of Bachmair and Ganzinger,
  this is also a model of $\floor{N}$. By Lemma~\ref{lem:gf-interpretation-from-gh},
  it follows that $\III \models N$.

  \medskip\noindent
  (R2)\enskip
  By Lemma~5.6(i) of Bachmair and Ganzinger, this holds on $\GF$.
  For clauses and all inferences except \infname{GArgCong} and \infname{GExt},  this implies
  that it holds on $\GH$ as well because $\flooronly$ is a redundancy-preserving
  bijection between $\CGH$ and $\CGF$ and between these inferences.
  For \infname{GArgCong} and \infname{GExt} inferences,
  it holds because it holds on clauses.

  \medskip\noindent
  (R3)\enskip
  The proof is analogous to (R2), with Lemma~5.6(ii) of Bachmair and Ganzinger instead
  of Lemma~5.6(i).

  \medskip\noindent
  (R4)\enskip
  We must show that for all inferences with
  $\concl(\iota) \in N$, we have $\iota \in \GHRedI(N)$.
  Since $\concl(\iota) \prec \mprem(\iota)$ for all $\iota\in\GFInf$, this holds on $\GF$.
  For all inferences except \infname{GArgCong} and \infname{GExt}, since $\flooronly$ is a bijection preserving redundancy,
  it follows that it also holds also on $\GH$.
  For \infname{GArgCong} and \infname{GExt} inferences, it holds by definition.
  \end{proof}

\begin{lemmax} \label{lem:grounding-function}
  The grounding functions $\gnd^\GHSel$ for $\GHSel\in\gnd(\HSel)$
  are grounding functions in the sense of the framework.
\end{lemmax}
\begin{proof}
We must prove the conditions (G1) to (G3)
defined by Waldmann et al., which, 
adapted to our context, state the following:
\begin{enumerate}[leftmargin=3em]
\item[(G1)] $\gnd(\bot) = \{ \bot \}$;
\item[(G2)] for every $C \in \CGH$, if $\bot \in \gnd(C)$, then $C = \bot$;
\item[(G3)] for every $\iota \in \HInf$, if $\gnd^\GHSel(\iota) \not= \Undef$, 
   then $\gnd^\GHSel(\iota) \subseteq \GHRedI^\GHSel(\gnd(\concl(\iota)))$.
\end{enumerate}

Clearly, $C = \bot$ if and only if $\bot \in \gnd(C)$ if and only if $\gnd(C) =
\{\bot\}$, proving (G1) and (G2).
For (G3), we have to show for all non-\infname{PosExt} inferences $\iota\in\HInf$ that $\gnd^\GHSel(\iota)\subseteq\GHRedI^\GHSel(\gnd(\concl(\iota)))$.
Let $\iota\in\HInf$ and $\iota'\in\gnd^\GHSel(\iota)$. By the definition of $\gnd^\GHSel$,
there exists a grounding substitution~$\theta$
such that $\prem(\iota') = \prem(\iota)\theta$
and $\concl(\iota') = \preconcl(\iota)\theta$. We want to show that $\iota'\in \GHRedI^\GHSel(\gnd(\concl(\iota)))$.

If $\iota'$ is not an \infname{GArgCong} or \infname{GExt} inference, by the definition of inference redundancy,
it suffices to show that
$\{D\in\floor{\gnd(\concl(\iota))} \mid D \prec \mprem(\floor{\iota'})\} \models \concl(\floor{\iota'})$.
We define a substitution $\theta'$ that extends $\theta$ to all variables in $\concl(\iota)$.
Due to purification, the clause $\concl(\iota)$ may contain variables not present in $\preconcl(\iota)$.
For each such variable $x'$, let $x$ be the variable in $\preconcl(\iota)$ that $x'$ stems from and
define $x'\theta' = x\theta$. Then the clause $\floor{\concl(\iota)\theta'}$ differs from
the clause $\floor{\concl(\iota')} = \floor{\preconcl(\iota)\theta'}$ only
in some additional grounded purification literals,
which all have the form $t \noteq t$ and are thus trivially false in any interpretation.
Hence, $\floor{\concl(\iota)\theta'} \models \floor{\concl(\iota')}$.
Since one of the variables of a purification literal must appear applied in the clause,
for each grounded purification literal $t \noteq t$ the term $t$ must be smaller than the maximal
term of the clause $\floor{\concl(\iota')}$.

If no literals are selected in $\mprem(\floor{\iota'})$,
inspection of the rules in
$\GFInf$ shows that $\floor{\concl(\iota)\theta'} \prec \mprem(\floor{\iota'})$.
Otherwise,
$\iota'$ can only be an \infname{ERes} inference or a \infname{Sup} inference into a negative literal.
If it is an \infname{ERes} inference, due to the selection restrictions,
the substitution $\sigma$ used in $\iota$ is the identity
for all variables of functional type. Therefore, applying $\sigma$ cannot trigger any purification
and hence $\floor{\concl(\iota)\theta'} = \floor{\preconcl(\iota)\theta'} \prec \mprem(\floor{\iota'})$.
If $\iota'$ is a \infname{Sup} inference into a negative literal, due to the selection restrictions,
the substitution $\sigma = \mgu(t,u)$ used in $\iota$ is the identity
for all variables of functional type that stem from the main premise.
Therefore the variables from the main premise $C$ need not be purified.
The variables from the side premise $D$ might need to be purified, yielding purification literals
of the form $x \noteq y$ where $x\theta' = y\theta'$. Then $x$ or $y$ must appear applied in $D$
and hence $x\theta'$ is smaller than $t\theta'$. Again, it follows that
$\floor{\concl(\iota)\theta'} \prec \mprem(\floor{\iota'})$.

This proves
$\{D\in\floor{\gnd(\concl(\iota))} \mid D \prec \mprem(\floor{\iota'})\} \models \concl(\floor{\iota'})$.

In the nonpurifying calculi, if $\iota'$ is an \infname{GArgCong} or \infname{GExt} inference,
it suffices to show that $\concl(\iota')\in \gnd(\concl(\iota))$.
This holds because $\concl(\iota') = \preconcl(\iota)\theta = \concl(\iota)\theta$.
In the purifying calculi, if $\iota'$ is an \infname{GArgCong} or \infname{GExt} inference,
we must show that $\floor{\gnd(\concl(\iota))} \models \floor{\concl(\iota')}$.
Defining $\theta'$ as above, we have $\floor{\concl(\iota)\theta'} \models \floor{\concl(\iota')}$,
as desired.
\end{proof}

To lift the completeness result of the previous section to
the nonground calculus $\HInf$,
we employ Theorem~14 of Waldmann et al., which, 
adapted to
our context, is stated as follows.
The theorem uses the notation $\Inf(N)$ 
to denote the set of $\Inf$-inferences
whose premises are in $N$,
for an inference system $\Inf$ and a clause set $N$.
Moreover, it uses the 
Herbrand entailment $\Gmodels$ on $\CHH$, which is
defined as $N_1 \Gmodels N_2$ if $\gnd(N_1) \models \gnd(N_2)$.

\begin{theoremx}[Lifting theorem]
\label{thm:lifting-theorem}
If $\GHInf^\GHSel$ is statically refutationally complete
\wrt\ $(\GHRedI^\GHSel, \GHRedC)$
for every $\GHSel \in \gnd(\HSel)$,
and if for every $N\subseteq \CHH$ that is
saturated \wrt\ $\HInf$ and $\HRedI$
there exists a $\GHSel \in \gnd(\HSel)$
such that
$\GHInf^\GHSel(\gnd(N)) 
\subseteq \gnd^\GHSel(\HInf(N)) \cup \GHRedI^\GHSel(\gnd(N))$,
then
$\HInf$ is statically refutationally complete \wrt\ $(\HRedI, \HRedC)$
and $\Gmodels$.
\end{theoremx}
\begin{proof}
  This is essentially Theorem~14 %
  of Waldmann et al.
  We take $\HH$ for $\mathbf{F}$, $\GH$ for $\mathbf{G}$, and $\gnd(\HSel)$ for $Q$.
  It is easy to see that the entailment relation $\models$ on $\GH$ is a consequence relation in the sense
  of the framework. By Lemma~\ref{lem:redundancy-criterion} and~\ref{lem:grounding-function},
  $(\GHRedI^\GHSel,\GHRedC)$ is a redundancy criterion
  in the sense of the framework,
  and $\gnd^\GHSel$ are grounding functions
  in the sense of the framework,
  for all $\GHSel\in\gnd(\HSel)$.
  The redundancy criterion $(\HRedI,\HRedC)$
  matches exactly the intersected lifted redundancy criterion
   $\mathit{Red}^{\ccap\gnd,\sqsupset}$
  of Waldmann et al.
  Theorem~14 %
  of Waldmann et al.\ states the theorem only for ${\sqsupset} = \varnothing$.
  By Lemma~16 %
  of Waldmann et al.,
  it also holds if ${\sqsupset} \not= \varnothing$.
\end{proof}

Let $N\subseteq\CHH$ be a clause set saturated \wrt\ $\HInf$ and $\HRedI$.
We assume that $\HSel$ fulfills the selection restrictions introduced in Section~\ref{ssec:the-inference-rules}.
For the above theorem to apply,
we need to show that there exists a selection
function $\GHSel\in\gnd(\HSel)$ such that
all inferences $\iota\in\GHInf^\GHSel$ with $\prem(\iota)\in\gnd(N)$
are liftable or redundant.
Here, by \emph{liftable}, we mean that $\iota$ is a $\smash{\gnd^\GHSel}$-ground instance
of a $\smash{\HInf}$-inference from $N$;
by \emph{redundant}, we mean that $\iota\in\smash{\GHRedI^\GHSel(\gnd(N))}$.

To choose the right selection function $\GHSel\in\gnd(\HSel)$, we observe that each ground clause
$C\in\gnd(N)$ must have at least one corresponding clause $D\in N$ such that $C$ is a ground instance of $D$.
We choose one of them for each $C\in\gnd(N)$, which we denote by $\gnd^{-1}(C)$. Then let $\GHSel$ select those literals in $C$ that correspond to
the literals selected by $\HSel$ in $\gnd^{-1}(C)$. Given this selection function
$\GHSel$, we can show that all inferences from $\gnd(N)$ are liftable or redundant.

All non-\infname{Sup} inferences in $\GHInf$ are liftable (Lemma~\ref{lem:lifting1}).
For \infname{Sup}, some inferences
are liftable (Lemma~\ref{lem:lifting2}) and some are redundant (Lemma~\ref{lem:nonliftable-sup-redundant}).
As in standard superposition,
\infname{Sup} inferences into positions below variables are redundant.
The variable condition of each of the four calculi is
designed to cover the nonredundant
\infname{Sup} inferences into positions of variable-headed terms,
which makes these inferences liftable.

  \begin{lemmax}\label{lem:sigma-theta-eq-theta}
    Let $\sigma$ be the most general unifier of $s$ and $s'$. Let
    $\theta$ be an arbitrary unifier of $s$ and $s'$. Then $\sigma\theta =
    \theta$.
  \end{lemmax}
  \begin{proof}
    Like in first-order logic, we can assume that $\sigma$ is idempotent without
    loss of generality \cite[Corollary~7.2.11]{fitting-1996}.
    Since $\sigma$ is most general, there exists a substitution $\rho$ such
    that $\sigma\rho = \theta$. Therefore, by idempotence, $\sigma\theta =
    \sigma\sigma\rho = \sigma\rho = \theta$.
  \end{proof}

\begin{lemmax}[Lifting of \infname{ERes}, \infname{EFact}, \infname{GArgCong}, and \infname{GExt}]
  \label{lem:lifting1}
  All \infname{ERes}, \infname{EFact}, \infname{GArgCong}, and \infname{GExt} inferences are liftable.
\end{lemmax}
\begin{proof}
  \infname{ERes}: Let $\iota\in\GHInf^\GHSel$ be an \infname{ERes} inference with $\prem(\iota)\in\gnd(N)$.
  Then $\iota$ is of the form
  \[\namedinference{ERes}{C\theta~=~C'\theta \llor s\theta \noteq s'\theta}{C'\theta}\]
  where $\gnd^{-1}(C\theta) = C = C' \llor s \noteq s'$
  and the literal $s\theta \noteq s'\theta$
  is eligible \wrt\ $\GHSel$.
  Let $\sigma = \mgu(s,s')$.
  It follows that
  $s \noteq s'$ is eligible in $C$ \wrt\ $\sigma$ and $\HSel$.
  Moreover, $s\theta$ and $s'\theta$ are unifiable and ground,
  and therefore $s\theta = s'\theta$.
  Thus, the following inference $\iota'\in\HInf$ is applicable:
  \[\namedinference{ERes}{C' \mathrel\lor s \not\eq s'}{\pure(C'\sigma)}\]
  (where $\pure$ is the identity in the nonpurifying calculi).
  By Lemma~\ref{lem:sigma-theta-eq-theta}, we have $C'\sigma\theta = C'\theta$.
  Therefore, $\iota$ is the $\theta$-ground instance of $\iota'$
  and is therefore liftable.

  \medskip

  \noindent
  \infname{EFact}:\enskip Analogously, if $\iota\in\GHInf^\GHSel$
  is an \infname{EFact} inference with $\prem(\iota)\in\gnd(N)$,
  then $\iota$ is of the form
  \[\namedinference{EFact}{C\theta~=~C'\theta \llor s'\theta \eq t'\theta \llor s\theta \eq t\theta}
  {C'\theta \llor t\theta \noteq t'\theta \llor s\theta \eq t'\theta}\]
  where $\gnd^{-1}(C\theta) = C =  C' \llor s' \eq t' \llor s \eq t$,
  the literal $s\theta \eq t\theta$ is eligible in $C$ \wrt\ $\GHSel$,
  and $s\theta\not\prec t\theta$.
  Let $\sigma = \mgu(s,s')$.
  Hence, $s\eq t$ is eligible in $C$ \wrt\ $\sigma$ and $\HSel$. We have $s\not\prec t$.
  Moreover, $s\theta$ and $s'\theta$ are unifiable and ground. Hence,
  $s\theta = s'\theta$.
  Thus, the following inference $\iota'\in\HInf$ is applicable:
  \[\namedinference{EFact}{C' \llor s' \eq t' \llor s \eq t}
  {\pure((C' \llor t \noteq t' \llor s \eq t')\sigma)} \]
  By Lemma~\ref{lem:sigma-theta-eq-theta}, we have $\preconcl(\iota')\theta = \concl(\iota)$.
  Hence, $\iota$ is the $\theta$-ground instance of $\iota'$
  and is therefore liftable.

  \medskip

  \noindent
  \infname{GArgCong}:
  Let $\iota\in\GHInf^\GHSel$ be a \infname{GArgCong} inference with $\prem(\iota)\in\gnd(N)$.
  Then $\iota$ is of the form
  \[\namedinference{GArgCong}{C\theta~=~C'\theta \llor s\theta \eq s'\theta}
  {C'\theta \llor s\theta\>\tuple{u}_n \eq s'\theta\>\tuple{u}_n}\]
  where $\gnd^{-1}(C\theta) = C = C' \llor s \eq s'$,
  the literal $s\theta \eq s'\theta$
  is strictly eligible \wrt\ $\GHSel$, and
  $s\theta$~and~$s'\theta$ are of functional type.
  It follows that $s$~and~$s'$ have either a functional
  or a polymorphic type.
  Let $\sigma$ be the most general substitution
  such that $s\sigma$ and $s'\sigma$ take $n$ arguments.
  Then $s \noteq s'$ is eligible in $C$ \wrt\ $\sigma$ and $\HSel$.
  Hence the following inference $\iota'\in\HInf$ is applicable:
  \[\namedinference{ArgCong}{C' \llor s \eq s'}
  {\pure(C'\sigma \llor s\sigma\>\tuple{x}_n \eq s'\sigma\>\tuple{x}_n)}\]
  Then $\iota$ is a ground instance of $\iota'$
  and is therefore liftable.

\medskip

\noindent
\infname{GExt}: The conclusion of a \infname{GExt} inference in $\GHInf$ is by definition a
ground instance of the conclusion of the \infname{Ext} inference in $\HInf$ before purification.
Hence, the \infname{GExt} inference is a ground instance of the \infname{Ext} inference.
Therefore it is liftable.
\end{proof}

\begin{lemmax}[Lifting of \infname{Sup}]
  \label{lem:lifting2}
  Let $\iota\in\GHInf^\GHSel$ be a \infname{Sup} inference
  \[\namedinference{Sup}
  {D'\theta \mathrel\lor t\theta \eq t'\theta \hypsep C'\theta \mathrel\lor \fosubterm{s\theta}{t\theta}_p \doteq s'\theta}
  {D'\theta \mathrel\lor C'\theta \mathrel\lor \fosubterm{s\theta}{ t'\theta}_p \doteq s'\theta}
  \]
  where $\gnd^{-1}(D\theta) = D = D' \llor t \eq t'$ and
  $\gnd^{-1}(C\theta) = C = C' \llor s \doteq s'$.
  Suppose that the position $p$ exists as a green subterm in $s$.
  Let $u$ be the green subterm of $s$ at that position and $\sigma = \mgu(t,u)$
  {\upshape(}which exists since $\theta$ is a unifier{\upshape)}.
  If the variable condition holds for $C$, $t$, $t'$, $u$, and $\sigma$,
  then $\iota$ is liftable.
\end{lemmax}
\begin{proof}
  The inference conditions of $\iota$ can be lifted to $D$ and $C$.
  That $t\theta\eq t'\theta$ is strictly eligible in $D\theta$ \wrt\ $\GHSel$
  implies that $t\eq t'$ is strictly eligible in $D$ \wrt\ $\sigma$ and $\HSel$.
  If $s\theta\doteq s'\theta$ is (strictly) eligible in $C\theta$ \wrt\ $\GHSel$,
  then $s\doteq s'$ is (strictly) eligible in $C$ \wrt\ $\sigma$ and $\HSel$.
  Moreover,
  $D\theta \not\succeqC C\theta$ implies $D \not\succeqC C$,
  $t\theta \not\prec t'\theta$ implies $t \not\prec t'$, and
  $s\theta \not\prec s'\theta$ implies $s \not\prec s'$.

  By assumption, $p$ is a position of $s$ and the variable condition holds.
  Thus, the following inference $\iota'\in\HInf$ is applicable:
  \[\namedinference{Sup}
  {D' \mathrel\lor t \eq t' \hypsep C' \mathrel\lor \fosubterm{s}{u}_p \doteq s'}
  {\pure((D' \mathrel\lor C' \mathrel\lor \fosubterm{s}{ t'}_p \doteq s')\sigma)}
  \]
  By Lemma~\ref{lem:sigma-theta-eq-theta},
  we have
  $(\preconcl(\iota'))\theta
  = \concl(\iota)$.
  Hence, $\iota$ is the $\theta$-ground instance of $\iota'$
  and is therefore liftable.
\end{proof}

The other \infname{Sup} inferences might not be liftable, but they are
redundant:

\begin{lemmax}\label{lem:nonliftable-sup-redundant}
  Let $\iota\in\GHInf^\GHSel$ be a \infname{Sup} inference
  \[\namedinference{Sup}
  {D'\theta \mathrel\lor t\theta \eq t'\theta \hypsep C'\theta \mathrel\lor \fosubterm{s\theta}{t\theta}_p \doteq s'\theta}
  {D'\theta \mathrel\lor C'\theta \mathrel\lor \fosubterm{s\theta}{ t'\theta}_p \doteq s'\theta}
  \]
  where $\gnd^{-1}(D\theta) = D = D' \llor t \eq t'$ and
  $\gnd^{-1}(C\theta) = C = C' \llor s \doteq s'$.
  Suppose that Lemma~\ref{lem:lifting2} does not apply.
  This could be either because the position $p$ is below a variable in $s$
  or because the variable condition does not hold.
  Then $\iota\in\GHRedI^\GHSel(\gnd(N))$.
\end{lemmax}
  \begin{proof}
    By the definition of $\GHRedI$, to show $\iota\in\GHRedI^\GHSel(\gnd(N))$, it suffices to prove that
    $\{E \in \floor{\gnd(N)} \mid E \prec \floor{C\theta}\} \models \floor{\concl(\iota)}$.
    Let $\III$ be a first-order model of all $E\in \floor{\gnd(N)}$ with $E \prec \floor{C\theta}$.
    We must show that $\III \models \floor{\concl(\iota)}$. If $\III \models \floor{D'\theta}$, this is obvious.
    So we further assume that $\III \not\models \floor{D'\theta}$.
    Since $D\theta \precC C\theta$ by the \infname{Sup} inference conditions,
    it follows that $\III\models \floor{t\theta \eq t'\theta}$. By congruence, it suffices to show
    $\III\models \floor{C\theta}$.
    We proceed by a case distinction on the two possible reasons why Lemma~\ref{lem:lifting2} does not apply:

    \medskip

    \noindent
    \textsc{Case 1:}\enskip The position $p$ is below a variable in $s$.
    Then $t\theta$ is a proper green subterm of~$x\theta$
    and hence a green subterm of $x\theta\;\tuple{w}$ for any arguments $\tuple{w}$.
    Let $v$ be the term that we obtain by replacing $t\theta$ by $t'\theta$ in $x\theta$
    at the relevant position. It follows from our assumptions about $\III$ that
     $\III\models \floor{t\theta\eq t'\theta}$, and by congruence,
    $\III\models \floor{x\theta\;\tuple{w}\eq v\;\tuple{w}}$ for any arguments $\tuple{w}$.
    Hence, $\III\models \floor{C\theta}$ if and only if
    $\III\models \floor{C\{x\mapsto v\}\theta}$.
    By the inference conditions we have $t\theta \succ t'\theta$, which
    implies $\floor{C\theta} \succC \floor{C\{x\mapsto v\}\theta}$ by compatibility
    with green contexts. Therefore,
    we have $\III\models \floor{C\{x\mapsto v\}\theta}$ and
    hence $\III\models \floor{C\theta}$.

    \medskip

    \noindent
    \textsc{Case 2:}\enskip The variable condition does not hold.
    In the extensional calculi, it follows that $u$ has a variable head and jells with $t\eq t'$.
    By Definition~\ref{def:jell},
    this means that $u$, $t$, and $t'$ have the following form:
    $u=x\;\tuple{v}_n$ for some variable $x$ and a tuple of terms $\tuple{v}_n$
    of length $n \geq 0$; $t=\tilde{t}\;\tuple{x}_n$ and $t'=\tilde{t}\vthinspace'\>\tuple{x}_n$,
    where $\tuple{x}_n$ are variables that do not occur elsewhere in $D.$

    For the intensional calculi, we have $u\in\VV$.
    Thus, $u$, $t$, and $t'$ can be written in the same form as described above for the extensional calculi, with $n=0$.

    \medskip

    \noindent
    \textsc{Case 2.1 (Purifying calculi):}\enskip
    First, we assume that $x$ occurs only with arguments $\tuple{v}_n$ in $C.$ For
    the intensional calculus, this must be the case because $n=0$ and hence $x$
    can only occur without arguments by the definition of $\pureint$ and the
    literal selection restriction. Define a substitution $\theta'$ by $x\theta' =
    \tilde{t}\vthinspace'\theta$ and $y\theta' = y\theta$ for other variables
    $y$. Since $t\theta \succ t'\theta$ by the inference conditions, we have $C\theta \succC C\theta'$.
    Moreover, $C\theta' \in \gnd(N)$. Then $\III\models\floor{C\theta}$ by
    congruence, because $\III\models\floor{C\theta'}$ and  $\III\models
    \floor{t\theta \eq t'\theta}$.

    Now we assume that $x$ occurs with arguments other than $\tuple{v}_n$ in $C.$
    This can only happen in the extensional calculus and by the selection
    restrictions, $s\theta \doteq s'\theta$ must not be selected in
    $C\theta$. Therefore, $s\theta$ is the maximal term in $C\theta$. Then $s
    \not= x$ and hence $\tuple{v}_n\not=\varepsilon$ because otherwise $s\theta =
    x\theta$ would be smaller than the applied occurrence of $x\theta$ in
    $C\theta$.

    Define a substitution $\theta''$ such that $x\theta'' =
    \tilde{t}\vthinspace'\theta$, $y\theta'' = \tilde{t}\vthinspace'\theta$ for
    other variables $y$ if $y\theta = s\theta$ and $C$ contains the literal
    $x\noteq y$, and $y\theta'' = y\theta$ otherwise.

    We show that $C\theta \succC C\theta''$ by proving that no literal of
    $C\theta''$ is larger than the maximal literal $s\theta\doteq
    s'\theta$ of $C\theta$ and that $s\theta\doteq s'\theta$ appears more
    often in $C\theta$ than in $C\theta''$.

    For a literal of the form $x\noteq y$, we have $x\theta'' \prec s\theta$ and
    $y\theta'' \prec s\theta$. For literals that are not of this form, by the
    definition of $\pureext$ in the extensional calculus, $x$ occurs always with
    arguments $\tuple{v}_n$. Hence these literals are equal or smaller in
    $C\theta''$ than in $C\theta$, because $x\theta''\>\tuple{v}_n \prec
    x\theta\;\tuple{v}_n$ and $y\theta'' \preceq y\theta$. Therefore, no literal
    of $C\theta''$ is larger than the maximal literal $s\theta\doteq
    s'\theta$ of $C\theta$. Moreover, these inequalities show that every
    occurrence of $s\theta\doteq s'\theta$ in $C\theta''$ corresponds to
    an occurrence of $s\theta\doteq s'\theta$ in $C\theta$ that
    corresponds to a literal  in $C$ without the variable $x$. Since at least
    one occurrence of $s\theta\doteq s'\theta$  in $C\theta$ corresponds
    to a literal in $C$ containing $x$, $s\theta\doteq s'\theta$ appears
    more often in $C\theta$ than in $C\theta''$. This concludes the argument
    that $C\theta \succC C\theta''$.
    It follows that
    $\III\models \floor{C\theta''}$.

    We need to show that $\III \models\floor{C\theta}$. There is a \infname{PosExt} inference
    from $D$ to $D' \mathrel\lor \tilde{t} \eq
    \tilde{t}\vthinspace'$.
    This inference is in $\HRedI(N)$ because $N$ is saturated. Therefore, $D'\theta \mathrel\lor \tilde{t}\theta \eq
    \tilde{t}\vthinspace'\theta$ is in $\gnd(N)\ccup\GHRedC(\gnd(N))$.
    It follows that $\III \models \floor{D'\theta \mathrel\lor \tilde{t}\theta \eq
    \tilde{t}\vthinspace'\theta}$ because this clause is
    smaller than $\floor{D'\theta}$ and hence smaller than $\floor{C\theta}$.
    Since $\floor{D'\theta}$ is false in $\III$, we have
    $\III\models\floor{\tilde{t}\theta \eq \tilde{t}\vthinspace'\theta}$.

    For every literal of the form $x\noteq y$, where $y\theta = s\theta$, the
    variable $y$ can only occur without arguments in $C$ because of the
    maximality of $s\theta$.
    We distinguish two cases. If for every literal of the form $x\noteq
    y$ where $y\theta = s\theta$, we have $\III\models\floor{y\theta'' \eq
    y\theta}$, then $\III\models \floor{C\theta}$ by congruence. If for some
    literal of the form $x\noteq y$ where $y\theta = s\theta$, we have
    $\III\models\floor{y\theta'' \noteq y\theta}$, then
    $\III\models\floor{y\theta\noteq x\theta}$ because $y\theta'' =
    \tilde{t}\vthinspace'\theta$, $\III\models
    \floor{\tilde{t}\vthinspace'\theta \eq \tilde{t}\theta}$, and
    $\tilde{t}\theta = x\theta$. Hence a literal of $\floor{C\theta}$ is true in $\III$
    and therefore $\III \models C\theta$.

    \medskip

    \noindent
    \textsc{Case 2.2 (Nonpurifying calculi):}\enskip
    Since the variable condition does not hold,
    we have $C\theta \succeqC C''\theta$,
    where $C'' = C\{x\mapsto \tilde{t}\vthinspace'\}$.
    We cannot have $C\theta = C''\theta$ because $x\theta =
    \tilde{t}\theta\neq\tilde{t}\vthinspace'\theta$ and $x$ occurs in $C$.
    Hence, we have $C\theta \succC C''\theta$.

    By the definition of $\III$, $C\theta \succC C''\theta$ implies
    $\III\models\floor{C''\theta}$. We will use equalities that are true in $\III$ to
    rewrite $\floor{C\theta}$ into $\floor{C''\theta}$, which implies $\III\models
    \floor{C\theta}$ by congruence.

    By saturation of $N$, for any well-typed $m$-tuple of fresh
    variables $\tuple{z}$, we can use a \infname{PosExt} with premise $D$ (if
    $n>m$) or \infname{ArgCong} inference with premise $D$ (if $n<m$)  or using
    $D$ itself (if $n=m$) to show that $\gnd(D' \lor
    \tilde{t}\;\tuple{z} \eq \tilde{t}\vthinspace'\>\tuple{z})
    \subseteq \gnd(N) \cup
    \GHRedC(\gnd(N))$.
    Hence, $D'\theta \mathrel\lor \tilde{t}\theta\;\tuple{u} \eq
    \tilde{t}\vthinspace'\theta\;\tuple{u}$ is in $\gnd(N) \cup \GHRedC(\gnd(N))$ for
    any ground arguments $\tuple{u}$.

    We observe that whenever $\tilde{t}\theta\;\tuple{u}$ and
    $\tilde{t}\vthinspace'\theta\;\tuple{u}$ are smaller than the maximal term
    of $C\theta$ for some arguments $\tuple{u}$, we have
    \[\III\models \floor{\tilde{t}\theta\;\tuple{u}} \eq
    \floor{\tilde{t}\vthinspace'\theta\;\tuple{u}}
    \tag{\textdagger}\label{eq:congruences}  \]

    To show this, we assume that $\tilde{t}\theta\;\tuple{u}$ and
    $\tilde{t}\vthinspace'\theta\;\tuple{u}$ are smaller than the maximal term
    of $C\theta$ and we distinguish two cases: If $t\theta$ is smaller than the
    maximal term of $C\theta$, all terms in $D'\theta$ are smaller than the
    maximal term of $C\theta$ and hence $D'\theta \lor
    \tilde{t}\theta\;\tuple{u} \eq \tilde{t}\vthinspace'\theta\;\tuple{u} \precC
    C\theta$. If, on the other hand, $t\theta$ is equal to the maximal term of
    $C\theta$, $\tilde{t}\theta\;\tuple{u}$ and
    $\tilde{t}\vthinspace'\theta\;\tuple{u}$ are smaller than $t\theta$. Hence
    $\tilde{t}\theta\;\tuple{u} \eq \tilde{t}\vthinspace'\theta\;\tuple{u} \prec
    t\theta \eq t'\theta$ and $D'\theta \lor \tilde{t}\theta\;\tuple{u} \eq
    \tilde{t}\vthinspace'\theta\;\tuple{u} \precC D\theta \precC C\theta$. In
    both cases, since $\floor{D'\theta}$ is false in $\III$ by assumption,
    $\III\models \floor{\tilde{t}\theta\;\tuple{u}} \eq
    \floor{\tilde{t}\vthinspace'\theta\;\tuple{u}}$.

    We proceed by a case distinction on whether $s\theta$ appears in a selected
    or in a maximal literal of $C\theta$. In both cases we provide an algorithm
    that establishes the equivalence of $C\theta$ and $C''\theta$ via rewriting
    using (\ref{eq:congruences}). This might seem trivial at first sight, but we
    can only use the equations (\ref{eq:congruences}) if
    $\tilde{t}\theta\;\tuple{u}$ and $\tilde{t}\vthinspace'\theta\;\tuple{u}$
    are smaller than the maximal term of $C\theta$. Moreover, $\tuple{u}$ might
    itself contain positions where we want to rewrite, so the order of
    rewriting matters.

    \medskip
	
    \noindent
    \textsc{Case 2.2.1:}\enskip $s\theta$ is the maximal side of a maximal
    literal of $C\theta$. Then, since $C\theta \succC C''\theta$, every term in
    $C\theta$ and in $C''\theta $ is smaller than or equal to $s\theta$. Let $C_0$
    and $\tilde{C}_0$ be the clauses resulting from rewriting $\floor{t\theta}
    \rewrite \floor{t'\theta}$ wherever possible in $\floor{C\theta}$ and
    $\floor{C''\theta}$, respectively. Since $\floor{t\theta}$ is a subterm of
    $\floor{s\theta}$, now every term in $C_0$ and $\tilde{C}_0$ is strictly
    smaller than $\floor{s\theta}$.

    \looseness=-1
    We define $C_1, C_2, \dots$\ inductively as follows: Given $C_i$, choose a
    subterm of the form $\floor{\tilde{t}\theta\;\tuple{u}}$ where
    $\tilde{t}\theta\;\tuple{u} \succ \tilde{t}\vthinspace'\theta\;\tuple{u}$ or
    of the form $\floor{\tilde{t}\vthinspace'\theta\;\tuple{u}}$ where
    $\tilde{t}\vthinspace'\theta\;\tuple{u} \succ \tilde{t}\theta\;\tuple{u}$.
    Let $C_{i+1}$ be the clause resulting from rewriting that subterm
    $\floor{\tilde{t}\theta\;\tuple{u}}$ to
    $\floor{\tilde{t}\vthinspace'\theta\;\tuple{u}}$ or that subterm
    $\floor{\tilde{t}\vthinspace'\theta\;\tuple{u}}$ to
    $\floor{\tilde{t}\theta\;\tuple{u}}$ in $C_i$, depending on which term was
    chosen.
    Analogously, we define $\tilde{C}_1, \tilde{C}_2, \dots$ by applying the
    same algorithm to $\tilde{C}_0$. In both cases, the process terminates
    because $\succ$ is compatible with green contexts and well founded.
    Let $C_*$ and $\tilde{C}_*$ be the respective final clauses.

    The algorithm preserves the invariant that every term in $C_i$ and
    $\tilde{C}_i$ is strictly smaller than $s\theta$. By congruence via
    (\textdagger), applied at every step of the algorithm, we know that $C_{*}$
    and $\floor{C\theta}$ are equivalent in $\III$ and that $\tilde{C}_{*}$ and
    $\floor{C''\theta}$ are equivalent in $\III$ as well.

    We show that $C_{*} = \tilde{C}_{*}$. Assume that $C_{*} \not=
    \tilde{C}_{*}$. The algorithm preserves a second invariant, namely that
    $\ceil{C_i}$ and $\ceil{\tilde{C}_j}$ are equal except for positions where
    one contains $\tilde{t}\theta$ and the other one contains
    $\tilde{t}\vthinspace'\theta$. Consider a deepest position where
    $\ceil{C_*}$ and $\ceil{\tilde{C}_*}$ are different. The respective position
    in $C_*$ and $\tilde{C}_*$ then contains
    $\floor{\tilde{t}\theta\;\tuple{u}}$ and
    $\floor{\tilde{t}\vthinspace'\theta\;\tuple{u}}$ or vice versa. The
    arguments $\tuple{u}$ must be equal because we consider a deepest position.
    But then $\tilde{t}\theta\;\tuple{u} \succ
    \tilde{t}\vthinspace'\theta\;\tuple{u}$ or $\tilde{t}\theta\;\tuple{u} \prec
    \tilde{t}\vthinspace'\theta\;\tuple{u}$, which is impossible since the
    algorithm terminated in $C_*$ and $\tilde{C}_*$.
    This shows that $C_{*} = \tilde{C}_{*}$. Hence $\floor{C\theta}$ and
    $\floor{C''\theta}$ are equivalent, which proves
    $\III\models\floor{C\theta}$.

    \medskip

    \noindent
    \textsc{Case 2.2.2:}\enskip $s\theta$ is the maximal side of a selected
    literal of $C\theta$. Then, by the selection restrictions, $x$ cannot be the
    head of a maximal literal of $C.$

    At every position where $x\;\tuple{u}$ occurs in $C$ with some (or no)
    arguments~$\tuple{u}$, we rewrite $(\tilde{t}\;\tuple{u})\theta$ to
    $(\tilde{t}\vthinspace'\>\tuple{u})\theta$ in $C\theta$ if
    $(\tilde{t}\;\tuple{u})\theta\succ(\tilde{t}\vthinspace'\>\tuple{u})\theta$.
    We start with the innermost occurrences of $x$, so that the order of the
    two terms at one step does not change by later rewriting.
    Analogously, at every position where $x\;\tuple{u}$ occurs in $C$ with some
    (or no) arguments~$\tuple{u}$, we rewrite
    $(\tilde{t}\vthinspace'\>\tuple{u})\theta$ to $(\tilde{t}\;\tuple{u})\theta$
    in $C''\theta$ if
    $(\tilde{t}\vthinspace'\>\tuple{u})\theta\succ(\tilde{t}\;\tuple{u})\theta$,
    again starting with the innermost occurrences.

    We never rewrite at the top level of the maximal term of $C\theta$ or
    $C''\theta$ because $x$ cannot be the head of a maximal literal of $C.$ The
    two resulting clauses are identical because $C\theta$ and $C''\theta$ only
    differ at positions where $x$ occurs in $C.$ The rewritten terms are all
    smaller than the maximal term of $C\theta$. With (\ref{eq:congruences}),
    this implies that $\III\models \floor{C\theta}$ by congruence.
  \end{proof}

  With these properties of our inference systems in place,
  the saturation framework's lifting theorem (Theorem~\ref{thm:lifting-theorem}) 
  guarantees static and dynamic refutational completeness
  of $\HInf$ \wrt\ $\HRedI$. However, this theorem gives us refutational
  completeness \wrt\ the Herbrand entailment $\Gmodels$,
  defined as $N_1 \Gmodels N_2$ if $\gnd(N_1) \models \gnd(N_2)$,
  whereas our semantics is Tarski entailment $\models$, defined as
  $N_1 \models N_2$ if any model of $N_1$ is a model of $N_2$.
  The following lemma repairs this mismatch:

  \begin{lemmax}
    \label{lem:herbrand-tarski}
    For $N \subseteq \CHH$, we have $N \Gmodels \bot$ if and only if $N \models \bot$.
  \end{lemmax}
  \begin{proof}
    By Lemma~\ref{lem:apply-subst},
    any model of $N$ is also a model of $\gnd(N)$---%
    i.e., $N \not\models \bot$ implies $N \not\Gmodels \bot$.
    For the other direction, we need to show that
    $N \not\Gmodels \bot$ implies $N \not\models \bot$.
    Assume that $N \not\Gmodels \bot$---i.e., $\gnd(N) \not\models \bot$.
    Then there is a model $\III$ of $\gnd(N)$.
    We must show that there exists a model of $N$---i.e., $N \not\models \bot$.
    Let $\III'$ be an interpretation derived from $\III$ by removing all
    universes that are not the denotation of a type in $\TyGH$
    and removing all domain elements that are not the denotation of a term in $\TGH$,
    making $\III'$ term-generated.
    Clearly, in our clausal logic,
    this leaves the denotations of terms and the truth of ground clauses unchanged.
    Thus, $\III'\models\gnd(N)$. We will show that $\III'\models N$.
    Let $C \in N$.
    We want to show that $C$ is true in $\III'$ for all valuations $\xi$.
    Fix a valuation $\xi$.
    By construction,
    for each variable $x$, there exists a ground term $s_x$ such that
    $\interpret{s_x}{\III'}{} = \xi(x)$.
    Let $\rho$ be the substitution that maps every free variable $x$ in $C$ to $s_x$.
    Then $\xi(x) = \interpret{s_x}{\III'}{} = \interpret{x\rho}{\III'}{}$ for all $x$.
    By treating the type variables of $C$ in the same way, we can also achieve that
    $\xi(\alpha) = \interpret{\alpha\rho}{\III'}{}$ for all $x$.
    By Lemma~\ref{lem:subst-lemma-general}, $\interpret{t\rho}{\III'}{} = \interpret{t}{\III'}{\xi}$ for all terms $t$ and
    $\interpret{\tau\negvthinspace\rho}{\III'}{} = \interpret{\tau}{\III'}{\xi}$ for all types $\tau$.
    Hence, $C\rho$ and $C$ have the same truth value in $\III'$ for $\xi$.
    Since $\III'\models\gnd(N)$, $C\rho$ is true in $\III'$ and thus
    $C$ is true in $\III'$ as well.
  \end{proof}

  \begin{theoremx}[Static refutational completeness]
    \label{thm:static-refutational-completeness}
    The inference system $\HInf$ is statically refutationally complete \wrt\ $(\HRedI, \HRedC)$.
    That means, if $N \subseteq \CHH$ is a clause set saturated \wrt\ $\HInf$ and $\HRedI$,
    then we have $N \models \bot$ if and only if $\bot \in N$.
  \end{theoremx}
  \begin{proof}
    We apply Theorem~\ref{thm:lifting-theorem}.
    By Theorem~\ref{thm:GH-refutational-completeness},
    $\GHInf^\GHSel$ is statically refutationally complete for all  $\GHSel\in\gnd(\HSel)$.
    By Lemmas~\ref{lem:lifting1}, \ref{lem:lifting2}, and~\ref{lem:nonliftable-sup-redundant},
    for every saturated $N\subseteq \CHH$, there exists a selection
    function $\GHSel\in\gnd(\HSel)$ such that
    all inferences $\iota\in\GHInf^\GHSel$ with $\prem(\iota)\in\gnd(N)$
    either are $\gnd^\GHSel$-ground instances of $\HInf$-inferences from $N$ or
    belong to $\smash{\GHRedI^\GHSel(\gnd(N))}$.

    Theorem~\ref{thm:lifting-theorem} implies that
    if $N \subseteq \CHH$ is a clause set saturated \wrt\ $\HInf$ and $\HRedI$,
    then $N \Gmodels \bot$ if and only if $\bot \in N$.
    By Lemma~\ref{lem:herbrand-tarski}, this also holds for the Tarski entailment $\models$.
    That is,
    if $N \subseteq \CHH$ is a clause set saturated \wrt\ $\HInf$ and $\HRedI$,
    then $N \models \bot$ if and only if $\bot \in N$.
  \end{proof}

  From static refutational completeness, we can easily derive dynamic
  refutational completeness.

  \begin{theoremx}[Dynamic refutational completeness]
    \label{thm:dynamic-refutational-completeness}
    The inference system $\HInf$ is dynamically refutationally complete \wrt\ $(\HRedI, \HRedC)$,
    as defined in Definition~\ref{def:dyn-complete}.
  \end{theoremx}
  \begin{proof}
    By
    Theorem~17 %
    of Waldmann et al.,
    this follows from Theorem~\ref{thm:static-refutational-completeness} and Lemma~\ref{lem:herbrand-tarski}.
  \end{proof}

\section{Implementation} %
\label{sec:implementation} %

Zipperposition \cite{cruanes-2015,cruanes-2017} is an open source
superposition-based theorem prover written in OCaml.%
\footnote{\url{https://github.com/sneeuwballen/zipperposition}}
It was initially designed for polymorphic first-order logic with equality, as
embodied by TPTP TF1 \cite{blanchette-paskevich-2013}. We will refer to this
implementation as Zipperposition's \relax{first-order mode}.
Later,
Cruanes extended the prover with a pragmatic \relax{higher-order mode} with support for
$\lambda$-abstractions and extensionality, without any completeness
guarantees.
We have now also implemented complete $\lambda$-free higher-order
modes based on the four calculi described in this \paper{}.

The pragmatic higher-order mode provided a convenient basis to implement our
calculi. It employs higher-order term and type representations and orders.
Its ad hoc calculus extensions are similar to our calculi. They
include an \infname{ArgCong}-like rule and a \infname{PosExt}-like rule, and
\infname{Sup} inferences are performed only at green subterms.
One of the bugs we found during our implementation work occurred because
argument positions shift when applied variables are instantiated. We resolved
this by numbering argument positions in terms from right to left.

To implement the $\lambda$-free higher-order mode, we restricted the
unification algorithm to non-$\lambda$-abstractions. To satisfy the
requirements on selection, we avoid selecting literals that contain
higher-order variables. To comply with our redundancy notion, we disabled
rewriting of nongreen subterms.
To improve term indexing of higher-order terms,
we replaced the imperfect discrimination trees by fingerprint indices
\cite{schulz-fingerprint-2012}.
To speed up the computation of the \infname{Sup} conditions,
we omit the condition $C\sigma \not\preceq D\sigma$ in the implementation,
at the cost of performing some additional inferences.

For the purifying calculi, we implemented purification as a
simplification rule. This ensures that it is applied aggressively on
all clauses, whether initial clauses from the problem or clauses produced
during saturation, before any inferences are performed.

For the nonpurifying calculi, we added the possibility to perform
\infname{Sup} inferences at variable positions. This means that variables must
be indexed as well. In addition, we modified the variable condition.
Depending on the term ordering, it may be expensive
or even impossible to decide whether there exists a grounding
substitution~$\theta$ with $t\sigma\theta \succ t'\sigma\theta$ and
$C\sigma\theta \prec C''\sigma\theta$.
We overapproximate the condition as follows:
(1)~check whether $x$ appears with different arguments in the clause~$C$;
(2)~use a term-order-specific algorithm to
determine whether there might exist a grounding substitution~$\theta$ and terms
$\tuple{u}$ such that $t\sigma\theta\succ t'\sigma\theta$ and
$t\sigma\theta\;\tuple{u} \prec t'\sigma\theta\;\tuple{u}$;
and (3)~check whether $C\sigma \not\succeqC C''\sigma$.
If these three conditions apply, we conclude that there might exist a ground
substitution~$\theta$ witnessing nonmonotonicity.

For the extensional calculi, we add axiom (\infname{Ext}) to the clause set.
To curb the explosion associated with extensionality, this axiom and all clauses derived from it are
penalized by the clause selection heuristic.
We also added the $\infname{NegExt}$ rule described in Section~\ref{ssec:rationale-for-the-inference-rules},
which resembles Vampire's extensionality
resolution rule \cite{gupta-et-al-2014}.

The \infname{ArgCong} rule can have infinitely many conclusions on polymorphic clauses.
To capture this in the implementation, we store these infinite sequences of conclusions
in the form of finite instructions of how to obtain the actual clauses. In addition to the usual
active and passive set of the DISCOUNT loop architecture \cite{avenhaus-et-al-1995}, we use a set of scheduled inferences that
stores these instructions. We visit the scheduled inferences in this additional set
and the clauses in the passive set fairly to
achieve dynamic completeness of our prover architecture. Waldmann et al.~\cite[Example~34%
]{waldmann-et-al-2020-saturation}
and Bentkamp et al.\ \cite[Section~6]{bentkamp-et-al-lamsup-journal} 
describe this architecture in more detail.

Using Zipperposition, we can quantify the disadvantage of the applicative encoding on
Example~\ref{ex:add-l-r}.
A well-chosen KBO instance with argument coefficients allows Zipperposition to derive
$\bot$ in 4~iterations of the prover's main loop and 0.03\,s.
KBO or LPO with default settings needs 203~iterations and 0.4\,s,
whereas KBO or LPO on the applicatively encoded problem needs
203~iterations and more than 1\,s due to the larger terms.

\section{Evaluation}
\label{sec:evaluation}

We evaluated Zipperposition's implementation of our
four
calculi on
Sledgehammer-generated Isabelle/HOL benchmarks \cite{boehme-nipkow-2010}
and on benchmarks from the TPTP library \cite{sutcliffe-et-al-2012-tff,sutcliffe-et-al-2009}.
Our experimental data is available online.%
\footnote{\url{https://doi.org/10.5281/zenodo.3992618}}
We used the development version of Zipperposition, revision
\OK{\texttt{2031e216}}.\footnote{\OK{\url{https://github.com/sneeuwballen/zipperposition/tree/2031e216c1941acd76187882a073e8f1e533}}}

The Sledgehammer benchmarks,
corresponding to Isabelle's Judgment Day suite,
were regenerated to target
clausal $\lambda$-free higher-order logic.
They comprise 2506 problems in total, divided in two
groups based on the number of Isabelle facts (lemmas, definitions, etc.)\ selected
for inclusion in each problem:\ either 16 facts (SH16) or 256 facts (SH256). The
problems were generated by encoding $\lambda$-expressions as $\lambda$-lifted
supercombinators \cite{meng-paulson-2008-trans}.

From the TPTP library, we collected \OK{708}
first-order problems in TFF format and \OK{717} higher-order problems in
THF format, both groups containing both monomorphic and polymorphic problems.
We
excluded all problems
that contain interpreted arithmetic symbols,
the symbols \texttt{(@@+)}, \texttt{(@@-)}, \texttt{(@+)}, \texttt{(@-)}, \texttt{(\&)}, or tuples,
as well as
the \texttt{SYN000} problems, which are only intended to test the parser,
and
problems whose
clausal normal form
takes longer than 15\,s to compute or falls outside
the $\lambda$-free fragment described in Section~\ref{sec:logic}.

We want to answer the following %
questions:
\begin{enumerate}
\item What is the overhead of our calculi on first-order benchmarks? \label{q:vs-fo}
\item Do the calculi outperform the applicative encoding? \label{q:app-enc}
\item Do the purifying or the nonpurifying calculi achieve better results? \label{q:puri-vs-nonpuri}
\item What is the cost of nonmonotonicity? \label{q:cost-nonmono}
\end{enumerate}
Since 
the present work
is only a stepping stone towards a prover for full
higher-order logic, it would be misleading to compare this prototype with
state-of-the-art higher-order provers that support a stronger logic.
Many of the higher-order problems in the TPTP library are in fact
satisfiable for our \hbox{$\lambda$-free}
logic, even though they may be unsatisfiable for full
higher-order logic and labeled as such in the TPTP.

To answer question (\ref{q:vs-fo}) and (\ref{q:puri-vs-nonpuri}), we ran
Zipperposition's first-order mode on the first-order benchmarks and the
purifying and nonpurifying modes on all benchmarks.
To answer question (\ref{q:app-enc}), we implemented an applicative encoding mode in Zipperposition,
which performs the applicative encoding after the clausal normal form transformation
and then proceeds with Zipperposition's first-order mode.
The encoding makes all function symbols nullary and replaces all
applications with a polymorphic binary $\cst{app}$ symbol.

We instantiated all four calculi with three term orders:\
LPO \cite{blanchette-et-al-2017-rpo},
KBO~\cite{becker-et-al-2017} (without argument coefficients),
and EPO%
\ \cite{bentkamp-2018-epo-formalization}%
.
Among these, LPO is the only nonmonotonic order and therefore
the most relevant option to evaluate our calculi,
which are designed to cope with nonmonotonicity.
To answer question~(\ref{q:cost-nonmono}), we also include the monotone orders KBO and EPO.
EPO is an order designed to resemble LPO 
while fulfilling the requirements of 
a ground-total simplification order on $\lambda$-free terms.
KBO and EPO serve as a yardstick to assess the cost of nonmonotonicity.
With these monotone orders, no superposition inferences at variables are necessary
and thus the nonpurifying calculi become a straightforward generalization of the
standard superposition calculus
with the caveat that
it may be more
efficient %
to superpose at nongreen
subterms directly instead of relying on the \infname{ArgCong} rule.
On first-order benchmarks and in the applicative encoding mode, all three orders are monotone
because they are monotone on first-order terms.

Figure~\ref{fig:eval-int} summarizes, for the intensional calculi, the number
of solved satisfiable and unsatisfiable problems within 180\,s, and the
time taken to show unsatisfiability. Figure~\ref{fig:eval-ext} presents the
corresponding data for the extensional calculi.
\OK{The average time is computed over the problems
that all configurations for the respective benchmark set and term order found to be
unsatisfiable within the time limit.}
For each combination of benchmark set and term order,
the best result is highlighted in bold.
The evaluation was carried out on StarExec Iowa
\cite{stump-et-al-2014} using Intel Xeon E5-2609\,0 CPUs clocked at 2.40\,GHz.

\begin{figure}
  \def\arraystretch{1.1}
  \setlength\tabcolsep{.666ex}
  \begin{tabular}{@{}ll@{\hskip 3ex}rrr@{\hskip 3ex}rrr@{\hskip 3ex}rrr@{}}
    \toprule
  &            & \multicolumn{3}{c@{\hskip 3ex}}{\# sat}   & \multicolumn{3}{c@{\hskip 3ex}}{\# unsat}
  & \multicolumn{3}{c@{\hskip 3ex}}{\diameter~time} \\[-.1ex] %
  &            & LPO   & KBO   & EPO  & LPO   & KBO   & EPO  & LPO   & KBO   & EPO   \\
  \hline \rule{0pt}{3ex}%
SH16
  & applicative encoding & 111 & \bf 189 & 65 & 373 & 382 & 157 & 0.9 & 1.2 & 10.7 \\
  & nonpurifying calculus & \bf 136 & 165 & \bf 133 & \bf 383 & \bf 385 & \bf 381 & \bf 0.4 & \bf 0.3 & \bf 0.0 \\
  & purifying calculus & 82 & 98 & 82 & 363 & 363 & 355 & 1.3 & 2.0 & \bf 0.0 \\[1.25\jot]
SH256
  & applicative encoding & \bf 1 & \bf 1 & \bf 1 & 471 & 488 & 36 & 9.4 & 8.7 & 63.8 \\
  & nonpurifying calculus & \bf 1 & \bf 1 & \bf 1 & \bf 543 & \bf 554 & \bf 498 & \bf 2.3 & \bf 2.3 & \bf 0.1 \\
  & purifying calculus & \bf 1 & \bf 1 & \bf 1 & 523 & 528 & 484 & 2.6 & 3.4 & 0.5 \\[1.25\jot]
TFF
  & first-order mode & 0 & 0 & 0 & \bf 212 & \bf 229 & \bf 107 & \bf 1.9 & \bf 2.3 & \bf 1.5 \\
  & applicative encoding & 0 & 0 & 0 & 180 & 205 & 21 & 7.0 & 10.0 & 4.6 \\
  & nonpurifying calculus & 0 & 0 & 0 & 210 & \bf 229 & 105 & \bf 1.9 & 2.4 & \bf 1.5 \\
  & purifying calculus & 0 & 0 & 0 & 211 & \bf 229 & 105 & 2.1 & 2.6 & 1.6 \\[1.25\jot]
THF
  & applicative encoding & \bf 127 & \bf 115 & 111 & 523 & 522 & 428 & 0.9 & 0.6 & 0.8 \\
  & nonpurifying calculus & 111 & 114 & \bf 112 & \bf 529 & \bf 527 & \bf 516 & \bf 0.3 & \bf 0.3 & \bf 0.0 \\
  & purifying calculus & 108 & 109 & 108 & 528 & 526 & 514 & \bf 0.3 & 0.5 & \bf 0.0
  \\ \bottomrule
\end{tabular}
  \caption{Evaluation of the intensional calculi}
  \label{fig:eval-int}

  \vspace*{\floatsep}

  \def\arraystretch{1.1}
  \setlength\tabcolsep{.666ex}
  \begin{tabular}{@{}ll@{\hskip 3ex}rrr@{\hskip 3ex}rrr@{\hskip 3ex}rrr@{}}
    \toprule
  &            & \multicolumn{3}{c@{\hskip 3ex}}{\# sat}   & \multicolumn{3}{c@{\hskip 3ex}}{\# unsat}
  & \multicolumn{3}{c@{\hskip 3ex}}{\diameter~time} \\[-.1ex] %
  &            & LPO   & KBO   & EPO  & LPO   & KBO   & EPO   & LPO   & KBO   & EPO   \\
  \hline \rule{0pt}{3ex}%
SH16
  & applicative encoding & 79 & \bf 152 & 48 & 379 & 386 & 157 & 1.2 & 1.3 & 11.4 \\
  & nonpurifying calculus & \bf 103 & 131 & \bf 95 & \bf 386 & \bf 393 & \bf 387 & \bf 0.4 & \bf 0.1 & \bf 0.0 \\
  & purifying calculus & 32 & 57 & 32 & 367 & 365 & 363 & 2.0 & 1.7 & \bf 0.0 \\[1.25\jot]
SH256
  & applicative encoding & \bf 1 & \bf 1 & \bf 1 & 462 & 486 & 36 & 7.5 & 9.4 & 63.8 \\
  & nonpurifying calculus & \bf 1 & \bf 1 & \bf 1 & \bf 548 & \bf 572 & \bf 504 & \bf 1.9 & \bf 2.1 & \bf 0.1 \\
  & purifying calculus & \bf 1 & \bf 1 & \bf 1 & 512 & 529 & 482 & 2.2 & 5.0 & \bf 0.1 \\[1.25\jot]
TFF
  & first-order mode & 0 & 0 & 0 & \bf 212 & \bf 229 & \bf 107 & \bf 1.9 & \bf 2.5 & \bf 1.5 \\
  & applicative encoding & 0 & 0 & 0 & 178 & 202 & 21 & 7.9 & 11.7 & 4.7 \\
  & nonpurifying calculus & 0 & 0 & 0 & 207 & \bf 229 & 106 & 2.1 & 3.0 & \bf 1.5 \\
  & purifying calculus & 0 & 0 & 0 & 210 & \bf 229 & 105 & 2.2 & 3.2 & 1.6 \\[1.25\jot]
THF
  & applicative encoding & \bf 108 & \bf 109 & 105 & 526 & 527 & 436 & 0.9 & 0.6 & 1.1 \\
  & nonpurifying calculus & 106 & 108 & \bf 107 & \bf 539 & \bf 535 & \bf 526 & \bf 0.3 & \bf 0.3 & \bf 0.0 \\
  & purifying calculus & 96 & 97 & 96 & 530 & 529 & 519 & \bf 0.3 & 0.6 & \bf 0.0
  \\ \bottomrule
\end{tabular}
  \caption{Evaluation of the extensional calculi}
  \label{fig:eval-ext}
\end{figure}

The experimental results on the TFF part of the TPTP library confirm that our
calculi handle the vast majority of problems that are solvable in first-order
mode gracefully, and thus that the overhead is minimal, answering
question (\ref{q:vs-fo}).
On first-order problems, the calculi are occasionally at
variance with the first-order mode, due to the interaction of
\infname{ArgCong} with polymorphic types and due to the extensionality axiom (\infname{Ext}).
In contrast, the applicative encoding is comparatively
inefficient on problems that are already first-order.
For LPO, the success rate decreases by \OK{around 15\%}, and the average time to
show unsatisfiability \OK{triples}.

The SH16 benchmarks consist mostly of small higher-order problems.
The
small number of axioms benefits the applicative encoding enough to outperform
the purifying calculi but not the nonpurifying ones.
The SH256 benchmarks are also higher-order but much larger. Such
problems are underrepresented in the TPTP library. On these, our
calculi clearly outperform the applicative encoding, answering question
(\ref{q:app-enc}) decisively. This is hardly surprising
given that the proving effort is dominated by first-order reasoning, which they
can perform gracefully.

The THF benchmarks generally require more sophisticated higher-order reasoning than the Sledgehammer benchmarks,
as observed by Sultana, Blanchette, and Paulson \cite[Section~5]{sultana-et-al-2013}.
On these benchmarks, the empirical results are less clear;
the applicative encoding and our calculi are roughly neck-and-neck.
The nonpurifying calculi detect unsatisfiability slightly more frequently, whereas the applicative
encoding tends to find more saturations.
It seems that, due to the large amount of higher-order reasoning necessary
to solve TPTP problems, the advantage of our calculi on the first-order parts of the
derivation is not a decisive factor on these benchmarks.

Concerning question~(\ref{q:puri-vs-nonpuri}), the nonpurifying calculi
outperform their purifying relatives across all benchmarks. The raw data show
that on most benchmark sets, the problems solved by the nonpurifying calculi
are almost a superset of the problems solved by the purifying calculi. Only on
the SH256 benchmarks, the purifying calculi can solve a fair number of
problems that the nonpurifying calculi cannot solve (11~problems for the
intensional calculi with LPO and 9 problems for the extensional calculi with
LPO).

KBO tends to have a slight advantage over LPO on all benchmark sets.
But the gap between KBO and LPO is not larger on the higher-order benchmarks than on TFF.
Since LPO is monotonic on first-order terms but nonmonotonic on higher-order terms,
whereas KBO is monotonic on both, the best answer we can give to
question~(\ref{q:cost-nonmono}) is that no substantial cost seems to be
associated with nonmonotonicity.
In particular, for the nonpurifying calculi,
the additional superposition inferences at variables necessary 
with LPO do not have a negative impact on the overall performance.
EPO generally performs worse than the other two orders,
with the exception of the nonpurifying calculus on SH16 benchmarks, where it is roughly neck-and-neck with LPO.
This suggests that for small, mildly higher-order problems,
EPO can be a viable LPO-like complement to KBO
if one considers the effort to implement our calculi too high.

\section{Discussion and Related Work}
\label{sec:discussion-and-related-work}

Our calculi join a long list of extensions and refinements of superposition.
Among the most closely related is Peltier's \cite{peltier-2016} Isabelle/HOL
formalization of the refutational completeness of a superposition calculus
that operates on $\lambda$-free higher-order terms and that is parameterized by
a monotonic term order. Extensions with polymorphism and induction,
independently developed by Cruanes \cite{cruanes-2015,cruanes-2017} and Wand
\cite{wand-2017}, contribute to increasing the power of
automatic
provers. Detection of inconsistencies in axioms, as suggested by Schulz et
al.~\cite{schulz-et-al-2017}, is important for large axiomatizations.

Also of interest is Bofill and Rubio's \cite{bofill-rubio-2013}
integration of nonmonotonic orders in ordered paramodulation, a precursor of
superposition. Their work is a veritable tour de force, but it is also highly
complicated and restricted to ordered paramodulation. Lack of compatibility
with arguments being a mild form of nonmonotonicity, it seemed preferable to
start with superposition, enrich it with an \infname{ArgCong} rule, and tune
the side conditions until we obtained a complete calculus.

Most complications can be avoided by using a monotonic order such as KBO
without argument coefficients. However, coefficients can be useful
to help achieve compatibility with $\beta$-reduction. For example, the term
$\lambda x.\; x + x$ could be treated as a constant with a coefficient of~2 on
its argument and a heavy weight to ensure $(\lambda x.\; x + x)\; y
\succ y + y$. Although they do not use argument coefficients,
the recently developed 
\relax{$\lambda$-superposition} calculus by
Bentkamp et al.\ \cite{bentkamp-et-al-lamsup-journal}
and 
\relax{combinatory superposition} calculus by Bhayat and Reger
\cite{bhayat-reger-2020-combsup} need%
\ a nonmonotonic order to cope with
$\beta$-reduction.
They are modeled after our extensional nonpurifying and
intensional nonpurifying calculi, respectively.
Many researchers have proposed or used encodings of higher-order logic constructs into
first-order logic, including Robinson \cite{robinson-1970}, Kerber
\cite{kerber-1991}, Dougherty \cite{dougherty-1993}, Dowek et al.\
\cite{dowek-et-al-1995}, Hurd \cite{hurd-2003},
Meng and Paulson \cite{meng-paulson-2008-trans},
Obermeyer \cite{obermeyer-2009}, and Czajka \cite{czajka-2016}.
Encodings of types, such as those by Bobot and Paskevich \cite{bobot-paskevich-2011}
and Blanchette et al.\ \cite{blanchette-et-al-2016-types}, are also crucial
to obtain a sound encoding of higher-order logic.
These ideas are implemented in proof assistant tools such as HOLyHammer and
Sledgehammer \cite{blanchette-et-al-2016-qed}.

In the term rewriting community, $\lambda$-free higher-order logic is known as
applicative first-order logic. First-order rewrite techniques can be applied to
this logic via the applicative encoding. However, there are similar drawbacks
as in theorem proving to having $\cst{app}$ as the only nonnullary symbol.
Hirokawa et al.\ \cite{hirokawa-2013} propose a technique that resembles our
mapping $\flooronly$ to avoid these drawbacks.

Another line of research has focused on the development of automated proof
procedures for higher-order logic. 
Robinson's
\cite{robinson-1969}, Andrews's \cite{andrews-1971}, and Huet's \cite{huet-1973}
pioneering work stands out. Andrews \cite{andrews-2001} and Benzm\"uller and
Miller \cite{benzmueller-miller-2014} provide excellent surveys. The
competitive higher-order automatic theorem provers include
\textsc{Leo}-II \cite{benzmueller-et-al-2008-leo2} (based on RUE resolution),
Satallax \cite{brown-2012-ijcar} (based on a tableau calculus and a SAT solver),
agsyHOL \cite{lindblad-2014} (based on a focused sequent calculus and a
generic narrowing engine), Leo-III \cite{steen-benzmueller-2018} (based on
a pragmatic higher-order version of ordered paramodulation with no
completeness guarantees),
CVC4 and veriT \cite{barbosa-et-al-2019-hosmt} (both based on satisfiability modulo theories), and
Vampire \cite{bhayat-reger-2019-rcu,bhayat-reger-2020-combsup}
(based on superposition and $\cst{SK}$-style combinators).
The Isabelle proof assistant \cite{nipkow-et-al-2002}
(which includes a tableau reasoner and a rewriting engine)
and its Sledgehammer subsystem %
have also participated in the higher-order division of the CADE ATP System
Competition.
Zipperposition is a convenient vehicle for experimenting and prototyping
because it is easier to understand and modify than highly-optimized
C or \cpp{} provers. Our middle-term goal is to design higher-order
superposition calculi, implement them in state-of-the-art provers such as E
\cite{schulz-2013}, SPASS \cite{weidenbach-et-al-2009}, and Vampire
\cite{kovacs-voronkov-2013}, and integrate these in proof assistants to
provide a high level of automation. With its stratified architecture,
Otter-$\lambda$ \cite{beeson-2004} is perhaps the closest to what we are
aiming at, but it is limited to second-order logic and offers no completeness
guarantees. As a first step, Vukmirovi\'c, Blanchette, Cruanes, and Schulz
\cite{vukmirovic-et-al-2019} have generalized E's data structures and
algorithms to clausal $\lambda$-free higher-order logic, assuming a monotonic
KBO \cite{becker-et-al-2017}.

\section{Conclusion}
\label{sec:conclusion}

We presented four superposition calculi for intensional and extensional
clausal $\lambda$-free higher-order logic and proved them refutationally
complete. The calculi nicely generalize standard superposition and are
compatible with our $\lambda$-free higher-order LPO and KBO. Especially on
large problems, our experiments confirm what one would naturally expect:\ that
native support for partial application and applied variables outperforms the
applicative encoding.

The new calculi reduce the gap between proof assistants based on higher-order
logic and superposition provers. We can use them to reason about arbitrary
higher-order problems by axiomatizing suitable combinators. 
But perhaps more
importantly, they appear promising as a stepping stone towards complete,
highly efficient automatic theorem provers for full higher-order logic.
Indeed, the subsequent work by Bentkamp et al.\
\cite{bentkamp-et-al-lamsup-journal}, which introduces support for
$\lambda$-expressions, and Bhayat and Reger \cite{bhayat-reger-2020-combsup},
which works with $\cst{SK}$-style combinators, is largely based on our
nonpurifying calculi.

\let\Acksize=\normalsize

\def\ackname{\Acksize Acknowledgment}

\paragraph{\bfseries\upshape\ackname.}
{\Acksize
We are grateful to the maintainers of StarExec for letting us use their
service.
We want to thank Petar Vukmirovi\'c for his work on Zipperposition and for his advice.
We thank Christoph Benzm\"uller, Sander Dahmen, Johannes H\"olzl,
Anders Schlichtkrull, Stephan Schulz, Alexander Steen, \hbox{Geoff} Sutcliffe,
Andrei Voronkov, Daniel Wand,
Christoph Weidenbach, and the participants in the 2017 Dagstuhl
Seminar on Deduction beyond First-Order Logic for stimulating discussions.
We also want to thank Christoph Benzm\"uller, Ahmed Bhayat, Wan Fokkink, Alexander Steen,
Mark Summerfield, Sophie Tourret, and the anonymous reviewers for suggesting
several textual improvements.
Bentkamp and Blanchette's research has received funding from the European
Research Council (ERC) under the European Union's Horizon 2020 research and
innovation program (grant agreement No.\ 713999, Matryoshka).

}

\bibliographystyle{alpha}
\bibliography{ms}

\newcommand{\etalchar}[1]{$^{#1}$}
\begin{thebibliography}{BREO{\etalchar{+}}19}

\bibitem[ADF95]{avenhaus-et-al-1995}
J{\"{u}}rgen Avenhaus, J{\"{o}}rg Denzinger, and Matthias Fuchs.
\newblock {DISCOUNT:} {A} system for distributed equational deduction.
\newblock In Jieh Hsiang, editor, {\em RTA-95}, volume 914 of {\em LNCS}, pages
  397--402. Springer, 1995.

\bibitem[And71]{andrews-1971}
Peter~B. Andrews.
\newblock Resolution in type theory.
\newblock {\em J. Symb. Log.}, 36(3):414--432, 1971.

\bibitem[And01]{andrews-2001}
Peter~B. Andrews.
\newblock Classical type theory.
\newblock In John~Alan Robinson and Andrei Voronkov, editors, {\em Handbook of
  Automated Reasoning}, volume~II, pages 965--1007. Elsevier and {MIT} Press,
  2001.

\bibitem[BBCW18]{bentkamp-et-al-2018}
Alexander Bentkamp, Jasmin~Christian Blanchette, Simon Cruanes, and Uwe
  Waldmann.
\newblock Superposition for lambda-free higher-order logic.
\newblock In Didier Galmiche, Stephan Schulz, and Roberto Sebastiani, editors,
  {\em {IJCAR} 2018}, volume 10900 of {\em LNCS}, pages 28--46. Springer, 2018.

\bibitem[BBPS16]{blanchette-et-al-2016-types}
Jasmin~Christian Blanchette, Sascha B{\"{o}}hme, Andrei Popescu, and Nicholas
  Smallbone.
\newblock Encoding monomorphic and polymorphic types.
\newblock {\em Log. Meth. Comput. Sci.}, 12(4:13):1--52, 2016.

\bibitem[BBT{\etalchar{+}}21]{bentkamp-et-al-lamsup-journal}
Alexander Bentkamp, Jasmin Blanchette, Sophie Tourret, Petar Vukmirovi\'c, and
  Uwe Waldmann.
\newblock Superposition with lambdas.
\newblock {\em J. Autom. Reason.}, 2021.

\bibitem[BBWW17]{becker-et-al-2017}
Heiko Becker, Jasmin~Christian Blanchette, Uwe Waldmann, and Daniel Wand.
\newblock A transfinite {K}nuth--{B}endix order for lambda-free higher-order
  terms.
\newblock In Leonardo de~Moura, editor, {\em CADE-26}, volume 10395 of {\em
  LNCS}, pages 432--453. Springer, 2017.

\bibitem[BC04]{bertot-casteran-2004}
Yves Bertot and Pierre Cast\'eran.
\newblock {\em Interactive Theorem Proving and Program Development: {Coq'Art}:
  The Calculus of Inductive Constructions}.
\newblock Texts in Theoretical Computer Science. Springer, 2004.

\bibitem[Bee04]{beeson-2004}
Michael Beeson.
\newblock Lambda logic.
\newblock In David~A. Basin and Micha{\"{e}}l Rusinowitch, editors, {\em
  {IJCAR} 2004}, volume 3097 of {\em LNCS}, pages 460--474. Springer, 2004.

\bibitem[Ben18]{bentkamp-2018-epo-formalization}
Alexander Bentkamp.
\newblock Formalization of the embedding path order for lambda-free
  higher-order terms.
\newblock {\em Archive of Formal Proofs}, 2018.
\newblock \url{http://isa-afp.org/entries/Lambda_Free_EPO.html}.

\bibitem[BG94]{bachmair-ganzinger-1994}
Leo Bachmair and Harald Ganzinger.
\newblock Rewrite-based equational theorem proving with selection and
  simplification.
\newblock {\em J. Log. Comput.}, 4(3):217--247, 1994.

\bibitem[BG01]{bachmair-ganzinger-2001-resolution}
Leo Bachmair and Harald Ganzinger.
\newblock Resolution theorem proving.
\newblock In John~Alan Robinson and Andrei Voronkov, editors, {\em Handbook of
  Automated Reasoning}, volume~I, pages 19--99. Elsevier and {MIT} Press, 2001.

\bibitem[BKPU16]{blanchette-et-al-2016-qed}
Jasmin~Christian Blanchette, Cezary Kaliszyk, Lawrence~C. Paulson, and Josef
  Urban.
\newblock Hammering towards {QED}.
\newblock {\em J. Formaliz. Reas.}, 9(1):101--148, 2016.

\bibitem[BM14]{benzmueller-miller-2014}
Christoph Benzm{\"{u}}ller and Dale Miller.
\newblock Automation of higher-order logic.
\newblock In J{\"{o}}rg~H. Siekmann, editor, {\em Computational Logic},
  volume~9 of {\em Handbook of the History of Logic}, pages 215--254. Elsevier,
  2014.

\bibitem[BN10]{boehme-nipkow-2010}
Sascha B{\"{o}}hme and Tobias Nipkow.
\newblock Sledgehammer: {J}udgement {D}ay.
\newblock In J{\"{u}}rgen Giesl and Reiner H{\"{a}}hnle, editors, {\em {IJCAR}
  2010}, volume 6173 of {\em LNCS}, pages 107--121. Springer, 2010.

\bibitem[BP11]{bobot-paskevich-2011}
Fran\c{c}ois Bobot and Andrei Paskevich.
\newblock Expressing polymorphic types in a many-sorted language.
\newblock In Cesare Tinelli and Viorica Sofronie-Stokkermans, editors, {\em
  FroCoS 2011}, volume 6989 of {\em LNCS}, pages 87--102. Springer, 2011.

\bibitem[BP13]{blanchette-paskevich-2013}
Jasmin~Christian Blanchette and Andrei Paskevich.
\newblock {TFF1}: The {TPTP} typed first-order form with rank-1 polymorphism.
\newblock In Maria~Paola Bonacina, editor, {\em CADE-24}, volume 7898 of {\em
  LNCS}, pages 414--420. Springer, 2013.

\bibitem[BPTF08]{benzmueller-et-al-2008-leo2}
Christoph Benzm{\"u}ller, Lawrence~C. Paulson, Frank Theiss, and Arnaud
  Fietzke.
\newblock {LEO-II}---{A} cooperative automatic theorem prover for higher-order
  logic.
\newblock In Alessandro Armando, Peter Baumgartner, and Gilles Dowek, editors,
  {\em IJCAR 2008}, volume 5195 of {\em {LNCS}}, pages 162--170. Springer,
  2008.

\bibitem[BR13]{bofill-rubio-2013}
Miquel Bofill and Albert Rubio.
\newblock Paramodulation with non-monotonic orderings and simplification.
\newblock {\em J. Autom. Reason.}, 50(1):51--98, 2013.

\bibitem[BR19]{bhayat-reger-2019-rcu}
Ahmed Bhayat and Giles Reger.
\newblock Restricted combinatory unification.
\newblock In Pascal Fontaine, editor, {\em CADE-27}, volume 11716 of {\em
  LNCS}, pages 74--93. Springer, 2019.

\bibitem[BR20]{bhayat-reger-2020-combsup}
Ahmed Bhayat and Giles Reger.
\newblock A combinator-based superposition calculus for higher-order logic.
\newblock In Nicolas Peltier and Viorica Sofronie-Stokkermans, editors, {\em
  IJCAR 2020, Part I}, volume 12166 of {\em LNCS}, pages 278--296. Springer,
  2020.

\bibitem[Bra75]{Brand1975}
D.~Brand.
\newblock Proving theorems with the modification method.
\newblock {\em SIAM J. Comput.}, 4:412--430, 1975.

\bibitem[BREO{\etalchar{+}}19]{barbosa-et-al-2019-hosmt}
Haniel Barbosa, Andrew Reynolds, Daniel El~Ouraoui, Cesare Tinelli, and Clark
  Barrett.
\newblock Extending {SMT} solvers to higher-order logic.
\newblock In Pascal Fontaine, editor, {\em CADE-27}, volume 11716 of {\em
  LNCS}, pages 35--54. Springer, 2019.

\bibitem[Bro12]{brown-2012-ijcar}
Chad~E. Brown.
\newblock {S}atallax: {A}n automatic higher-order prover.
\newblock In Bernhard Gramlich, Dale Miller, and Uli Sattler, editors, {\em
  IJCAR 2012}, volume 7364 of {\em LNCS}, pages 111--117. Springer, 2012.

\bibitem[BWW17]{blanchette-et-al-2017-rpo}
Jasmin~Christian Blanchette, Uwe Waldmann, and Daniel Wand.
\newblock A lambda-free higher-order recursive path order.
\newblock In Javier Esparza and Andrzej~S. Murawski, editors, {\em {FoSSaCS}
  2017}, volume 10203 of {\em LNCS}, pages 461--479. Springer, 2017.

\bibitem[Chu40]{church-1940}
Alonzo Church.
\newblock A formulation of the simple theory of types.
\newblock {\em J. Symb.\ Log.}, 5(2):56--68, 1940.

\bibitem[CP03]{cervesato-pfenning-2003}
Iliano Cervesato and Frank Pfenning.
\newblock A linear spine calculus.
\newblock {\em J. Log. Comput.}, 13(5):639--688, 2003.

\bibitem[Cru15]{cruanes-2015}
Simon Cruanes.
\newblock {\em Extending Superposition with Integer Arithmetic, Structural
  Induction, and Beyond}.
\newblock {Ph.D.}\ thesis, \'Ecole polytechnique, 2015.

\bibitem[Cru17]{cruanes-2017}
Simon Cruanes.
\newblock Superposition with structural induction.
\newblock In Clare Dixon and Marcelo Finger, editors, {\em FroCoS 2017}, volume
  10483 of {\em LNCS}, pages 172--188. Springer, 2017.

\bibitem[Cza16]{czajka-2016}
{\L}ukasz Czajka.
\newblock Improving automation in interactive theorem provers by efficient
  encoding of lambda-abstractions.
\newblock In Jeremy Avigad and Adam Chlipala, editors, {\em CPP 2016}, pages
  49--57. {ACM}, 2016.

\bibitem[DH86]{DigricoliHarrison1986}
Vincent~J. Digricoli and Malcolm~C. Harrison.
\newblock Equality-based binary resolution.
\newblock {\em J. ACM}, 33(2):253--289, 1986.

\bibitem[DHK95]{dowek-et-al-1995}
Gilles Dowek, Th{\'e}r{\`e}se Hardin, and Claude Kirchner.
\newblock Higher-order unification via explicit substitutions (extended
  abstract).
\newblock In {\em LICS '95}, pages 366--374. IEEE Computer Society, 1995.

\bibitem[Dou93]{dougherty-1993}
Daniel~J. Dougherty.
\newblock Higher-order unification via combinators.
\newblock {\em Theor. Comput. Sci.}, 114(2):273--298, 1993.

\bibitem[Fit96]{fitting-1996}
Melvin Fitting.
\newblock {\em First-Order Logic and Automated Theorem Proving}.
\newblock Springer-Verlag, 2nd edition, 1996.

\bibitem[Fit02]{fitting-2002}
Melvin Fitting.
\newblock {\em Types, Tableaus, and {G}{\"o}del's God}.
\newblock Kluwer, 2002.

\bibitem[GKKV14]{gupta-et-al-2014}
Ashutosh Gupta, Laura Kov{\'{a}}cs, Bernhard Kragl, and Andrei Voronkov.
\newblock Extensional crisis and proving identity.
\newblock In Franck Cassez and Jean{-}Fran{\c{c}}ois Raskin, editors, {\em
  {ATVA} 2014}, volume 8837 of {\em LNCS}, pages 185--200. Springer, 2014.

\bibitem[GM93]{gordon-melham-1993}
M.~J.~C. Gordon and T.~F. Melham, editors.
\newblock {\em Introduction to {HOL}: A Theorem Proving Environment for Higher
  Order Logic}.
\newblock Cambridge University Press, 1993.

\bibitem[Hen50]{henkin-1950}
Leon Henkin.
\newblock Completeness in the theory of types.
\newblock {\em J. Symb. Log.}, 15(2):81--91, 1950.

\bibitem[HMZ13]{hirokawa-2013}
Nao Hirokawa, Aart Middeldorp, and Harald Zankl.
\newblock Uncurrying for termination and complexity.
\newblock {\em J. Autom. Reasoning}, 50(3):279--315, 2013.

\bibitem[Hue73]{huet-1973}
G{\'{e}}rard~P. Huet.
\newblock A mechanization of type theory.
\newblock In Nils~J. Nilsson, editor, {\em IJCAI-73}, pages 139--146. William
  Kaufmann, 1973.

\bibitem[Hur03]{hurd-2003}
Joe Hurd.
\newblock First-order proof tactics in higher-order logic theorem provers.
\newblock In Myla Archer, Ben {Di~Vito}, and C{\'{e}}sar Mu{\~{n}}oz, editors,
  {\em Design and Application of Strategies/\kern.5pt Tactics in Higher Order
  Logics}, {NASA} Technical Reports, pages 56--68, 2003.

\bibitem[Ker91]{kerber-1991}
Manfred Kerber.
\newblock How to prove higher order theorems in first order logic.
\newblock In John Mylopoulos and Raymond Reiter, editors, {\em IJCAI-91}, pages
  137--142. Morgan Kaufmann, 1991.

\bibitem[Kop12]{kop-2012-phd}
Cynthia Kop.
\newblock {\em Higher Order Termination: Automatable Techniques for Proving
  Termination of Higher-Order Term Rewriting Systems}.
\newblock {Ph.D.}\ thesis, Vrije Universiteit Amsterdam, 2012.

\bibitem[KV13]{kovacs-voronkov-2013}
Laura Kov{\'{a}}cs and Andrei Voronkov.
\newblock First-order theorem proving and {V}ampire.
\newblock In Natasha Sharygina and Helmut Veith, editors, {\em {CAV} 2013},
  volume 8044 of {\em LNCS}, pages 1--35. Springer, 2013.

\bibitem[Lei94]{leivant-1994}
Daniel Leivant.
\newblock Higher order logic.
\newblock In Dov~M. Gabbay, Christopher~J. Hogger, J.~A. Robinson, and
  J{\"{o}}rg~H. Siekmann, editors, {\em Handbook of Logic in Artificial
  Intelligence and Logic Programming, Volume 2, Deduction Methodologies}, pages
  229--322. Oxford University Press, 1994.

\bibitem[Lin14]{lindblad-2014}
Fredrik Lindblad.
\newblock A focused sequent calculus for higher-order logic.
\newblock In St{\'{e}}phane Demri, Deepak Kapur, and Christoph Weidenbach,
  editors, {\em {IJCAR} 2014}, volume 8562 of {\em LNCS}, pages 61--75.
  Springer, 2014.

\bibitem[Mil87]{miller-1987}
Dale~A. Miller.
\newblock A compact representation of proofs.
\newblock {\em Studia Logica}, 46(4):347--370, 1987.

\bibitem[MP08]{meng-paulson-2008-trans}
Jia Meng and Lawrence~C. Paulson.
\newblock Translating higher-order clauses to first-order clauses.
\newblock {\em J. Autom. Reason.}, 40(1):35--60, 2008.

\bibitem[NPW02]{nipkow-et-al-2002}
Tobias Nipkow, Lawrence~C. Paulson, and Markus Wenzel.
\newblock {\em {I}sabelle/{HOL}: A Proof Assistant for Higher-Order Logic},
  volume 2283 of {\em {LNCS}}.
\newblock Springer, 2002.

\bibitem[Obe09]{obermeyer-2009}
Fritz~H. Obermeyer.
\newblock {\em Automated Equational Reasoning in Nondeterministic
  $\lambda$-Calculi Modulo Theories $\mathcal{H}{}^{*}$}.
\newblock {Ph.D.}\ thesis, Carnegie Mellon University, 2009.

\bibitem[Pel16]{peltier-2016}
Nicolas Peltier.
\newblock A variant of the superposition calculus.
\newblock {\em Archive of Formal Proofs}, 2016.

\bibitem[Rob69]{robinson-1969}
J.A. Robinson.
\newblock Mechanizing higher order logic.
\newblock In B.~Meltzer and D.~Michie, editors, {\em Machine Intelligence},
  volume~4, pages 151--170. Edinburgh University Press, 1969.

\bibitem[Rob70]{robinson-1970}
J.A. Robinson.
\newblock A note on mechanizing higher order logic.
\newblock In B.~Meltzer and D.~Michie, editors, {\em Machine Intelligence},
  volume~5, pages 121--135. Edinburgh University Press, 1970.

\bibitem[SB18]{steen-benzmueller-2018}
Alexander Steen and Christoph Benzm\"uller.
\newblock The higher-order prover {Leo-III}.
\newblock In Didier Galmiche, Stephan Schulz, and Roberto Sebastiani, editors,
  {\em IJCAR 2018}, volume 10900 of {\em {LNCS}}, pages 108--116. Springer,
  2018.

\bibitem[SBBT09]{sutcliffe-et-al-2009}
Geoff Sutcliffe, Christoph Benzm{\"{u}}ller, Chad~E. Brown, and Frank Theiss.
\newblock Progress in the development of automated theorem proving for
  higher-order logic.
\newblock In Renate~A. Schmidt, editor, {\em CADE-22}, volume 5663 of {\em
  LNCS}, pages 116--130. Springer, 2009.

\bibitem[SBP13]{sultana-et-al-2013}
Nik Sultana, Jasmin~Christian Blanchette, and Lawrence~C. Paulson.
\newblock {LEO-II} and {Satallax} on the {Sledgehammer} test bench.
\newblock {\em J. Applied Logic}, 11(1):91--102, 2013.

\bibitem[Sch02]{schulz-2002-brainiac}
Stephan Schulz.
\newblock E - a brainiac theorem prover.
\newblock {\em {AI} Commun.}, 15(2-3):111--126, 2002.

\bibitem[Sch12]{schulz-fingerprint-2012}
Stephan Schulz.
\newblock Fingerprint indexing for paramodulation and rewriting.
\newblock In Bernhard Gramlich, Dale Miller, and Uli Sattler, editors, {\em
  {IJCAR} 2012}, volume 7364 of {\em LNCS}, pages 477--483. Springer, 2012.

\bibitem[Sch13]{schulz-2013}
Stephan Schulz.
\newblock System description: {E} 1.8.
\newblock In Kenneth~L. McMillan, Aart Middeldorp, and Andrei Voronkov,
  editors, {\em LPAR-19}, volume 8312 of {\em LNCS}, pages 735--743. Springer,
  2013.

\bibitem[SL91]{SnyderLynch1991}
Wayne Snyder and Christopher Lynch.
\newblock Goal directed strategies for paramodulation.
\newblock In Ronald~V. Book, editor, {\em RTA-91}, volume 488 of {\em LNCS},
  pages 150--161. Springer, 1991.

\bibitem[SS89]{SchmidtSchauss1989}
Manfred Schmidt-Schau{\ss}.
\newblock Unification in a combination of arbitrary disjoint equational
  theories.
\newblock {\em J. Symb. Comput.}, 8:51--99, 1989.

\bibitem[SSCB12]{sutcliffe-et-al-2012-tff}
Geoff Sutcliffe, Stephan Schulz, Koen Claessen, and Peter Baumgartner.
\newblock The {TPTP} typed first-order form with arithmetic.
\newblock In Nikolaj Bj{\o}rner and Andrei Voronkov, editors, {\em LPAR-18},
  volume 7180 of {\em LNCS}, pages 406--419. Springer, 2012.

\bibitem[SST14]{stump-et-al-2014}
Aaron Stump, Geoff Sutcliffe, and Cesare Tinelli.
\newblock {StarExec}: {A} cross-community infrastructure for logic solving.
\newblock In St{\'{e}}phane Demri, Deepak Kapur, and Christoph Weidenbach,
  editors, {\em {IJCAR} 2014}, volume 8562 of {\em LNCS}, pages 367--373.
  Springer, 2014.

\bibitem[SSUP17]{schulz-et-al-2017}
Stephan Schulz, Geoff Sutcliffe, Josef Urban, and Adam Pease.
\newblock Detecting inconsistencies in large first-order knowledge bases.
\newblock In Leonardo de~Moura, editor, {\em {CADE}-26}, volume 10395 of {\em
  LNCS}, pages 310--325. Springer, 2017.

\bibitem[V{\"a}{\"a}19]{vaananen-2019}
Jouko V{\"a}{\"a}n{\"a}nen.
\newblock Second-order and higher-order logic.
\newblock In Edward~N. Zalta, editor, {\em The Stanford Encyclopedia of
  Philosophy}. Metaphysics Research Lab, Stanford University, fall 2019
  edition, 2019.

\bibitem[VBCS19]{vukmirovic-et-al-2019}
Petar Vukmirovi\'c, Jasmin~Christian Blanchette, Simon Cruanes, and Stephan
  Schulz.
\newblock Extending a brainiac prover to lambda-free higher-order logic.
\newblock In Tomas Vojnar and Lijun Zhang, editors, {\em TACAS 2019}, volume
  11427 of {\em LNCS}, pages 192--210. Springer, 2019.

\bibitem[Wan17]{wand-2017}
Daniel Wand.
\newblock {\em Superposition: Types and Polymorphism}.
\newblock {Ph.D.}\ thesis, Universit\"at des Saarlandes, 2017.

\bibitem[WDF{\etalchar{+}}09]{weidenbach-et-al-2009}
Christoph Weidenbach, Dilyana Dimova, Arnaud Fietzke, Rohit Kumar, Martin Suda,
  and Patrick Wischnewski.
\newblock {SPASS} version 3.5.
\newblock In Renate~A. Schmidt, editor, {\em CADE-22}, volume 5663 of {\em
  LNCS}, pages 140--145. Springer, 2009.

\bibitem[WTRB20]{waldmann-et-al-2020-saturation}
Uwe Waldmann, Sophie Tourret, Simon Robillard, and Jasmin Blanchette.
\newblock A comprehensive framework for saturation theorem proving.
\newblock In Nicolas Peltier and Viorica Sofronie-Stokkermans, editors, {\em
  IJCAR 2020, Part I}, volume 12166 of {\em LNCS}, pages 316--334. Springer,
  2020.

\end{thebibliography}
\end{document}